\documentclass[aap]{imsart}

\RequirePackage{amsthm,amsmath,amsfonts,amssymb}
\RequirePackage[numbers]{natbib}
\RequirePackage[colorlinks,citecolor=blue,urlcolor=blue]{hyperref}

\usepackage{color}
\usepackage[dvipsnames]{xcolor}
\usepackage{mathrsfs}
\usepackage{mathtools}  
\usepackage[font=small,labelfont=bf]{caption}
\usepackage{enumitem}
\usepackage{chngcntr}
\usepackage{bm} 
\usepackage{algorithmic,algorithm}
\usepackage{cancel}

\startlocaldefs
\theoremstyle{plain}
\newtheorem{theorem}{Theorem}
\newtheorem{proposition}{Proposition}
\newtheorem{lemma}{Lemma}
\theoremstyle{remark}
\newtheorem{assumption}{Assumption}
\newtheorem{remark}{Remark}
\newtheorem{example}{Example}
\renewcommand{\cal}[1]{\mathcal{#1}}

\renewcommand{\r}{\mathbb{R}}

\newcommand{\rf}{\mathfrak{r}}

\newcommand{\mmag}[1]{\left|#1\right|}
\newcommand{\norm}[1]{\left|\left|{#1}\right|\right|}

\newcommand{\Ebb}[1]{\mathbb{E}\left[#1\right]}
\newcommand{\Pbb}[1]{\mathbb{P}\left(#1\right)}

\DeclareMathOperator*{\arginf}{arginf}

\newcommand{\eno}{\eta^N_1}
\newcommand{\ent}{\eta^N_2}
\newcommand{\enot}{\eta^N_1\times \eta^N_2}

\newcommand{\eo}{\eta_1}
\newcommand{\et}{\eta_2}
\newcommand{\eot}{\eta_1\times \eta_2}

\endlocaldefs

\begin{document}

\begin{frontmatter}
\title{The divide-and-conquer sequential Monte Carlo algorithm: theoretical properties and limit theorems}
\runtitle{DaC-SMC: theoretical properties and limit theorems}
\begin{aug}
\author[A]{\fnms{Juan} \snm{Kuntz}\ead[label=ea]{juan.kuntz-nussio@warwick.ac.uk}},
\author[B]{\fnms{Francesca R.} \snm{Crucinio}\ead[label=eb]{Francesca.crucinio@ensae.fr}}
\and
\author[A]{\fnms{Adam M.} \snm{Johansen}\ead[label=ec]{a.m.johansen@warwick.ac.uk}}
\address[A]{Department of Statistics, University of Warwick, Coventry, U.K.,\\%
   \printead{ea,ec}.}
   \address[B]{CREST, ENSAE, Institut Polytechnique de Paris, Palaiseau, France,\\%
   \printead{eb}.}
\end{aug}

\begin{abstract}
We provide a comprehensive characterisation of the theoretical properties of the divide-and-conquer sequential Monte Carlo (DaC-SMC) algorithm. We firmly establish it as a well-founded method by showing that it possesses the same basic properties as conventional sequential Monte Carlo (SMC) algorithms do. In particular, we derive pertinent laws of large numbers, $L^p$ inequalities, and central limit theorems; and we characterize the bias in the normalized estimates produced by the algorithm and argue the absence thereof in the unnormalized ones. We further consider its practical implementation and several interesting variants; obtain expressions for its globally and locally optimal intermediate targets, auxiliary measures, and proposal kernels; and show that, in comparable conditions, DaC-SMC proves more statistically efficient than its direct SMC analogue.  We close the paper with a discussion of our results, open questions, and future research directions.
\end{abstract}

\begin{keyword}[class=MSC]
\kwd[Primary ]{65C05}
\kwd[; secondary ]{60F05}
\kwd{60F15}
\kwd{62F15}
\kwd{68W15}
\end{keyword}

\begin{keyword}
\kwd{strong law of large numbers}
\kwd{central limit theorem}
\kwd{interacting particle systems}
\kwd{product-form estimators}
\kwd{distributed computing}
\kwd{Bayesian inference}
\end{keyword}

\end{frontmatter}
\tableofcontents


%
\section{Introduction}\label{sec:intro}
Distributed inference algorithms that scale efficiently to large models and data sets are of widespread interest at the present time. Divide-and-conquer methods, which recursively decompose an inference problem into smaller problems whose solutions can be combined in a principled manner (e.g.\ see \cite{mattson2004patterns}), provide a promising approach in this setting. In computational statistics, many methods have been developed which allow for distributed inference either by performing inference for a subset of the data at each processor and ultimately combining their results~\cite{Dai2019,Minsker2014,Neiswanger2014,Scott2016,Wang2013} or by carrying out part of the computation associated with each step of an algorithm at a separate processor and minimizing the necessary communication~\cite{Gurbuzbalaban2021,Parayil2020,Rendell2021,Vono2021}. Another strategy which is beginning to find applications ranging from general distributed Bayesian inference~\cite{chan2021} through inference in broad model classes such as  general state space hidden Markov models~\cite{Corneflos2021,ding2019,ding2018}  to inference in specific domains such as phylogenetics~\cite{jewell2015} is the \emph{divide-and-conquer sequential Monte Carlo} (\emph{DaC-SMC}) algorithm introduced in~\cite{lindsten2017divide}. The DaC-SMC algorithm recursively breaks down large inferential problems into smaller, more manageable ones and merges their solutions in a systematic manner. It is particularly amenable to distributed or parallelized implementations (e.g.~\cite{Corneflos2021} or~\cite[Section~5.3]{lindsten2017divide}).

The goal of this paper is to provide a comprehensive theoretical foundation for this algorithm. The key conceptual innovation that enables us to do so is re-interpreting the DaC-SMC algorithm as a careful combination of sequential Monte Carlo (SMC) and a little-known type of `product-form' Monte Carlo estimators (which, once identified, we studied separately in a preliminary paper~\cite{Kuntz2021}).

Product-form estimators encompass a  class of generalized U-statistics~\cite{Lee1990,korolyuk2013theory} that exploit conditional independences present in the measure targeted, or the proposal distribution employed, to achieve lower variances than standard Monte Carlo estimators~\cite{Kuntz2021}. 
The idea underpinning them is very simple: if $(X^1_1,\dots,X^N_1),\dots,(X^1_K,\dots,X^N_K)$ are $K$ independent sequences of samples respectively drawn  from probability distributions $\mu_1,\dots,\mu_K$, then every `permutation' of these samples has the same product law:
$$(X^{n_1}_1,\dots,X^{n_K}_K)\sim\mu:=\mu_1\times\cdots\times\mu_K \quad\forall (n_1,\dots,n_K) \in \{1,\ldots,N\}^K.$$
In other words, using only $N$ samples from each marginal $\mu_1,\dots,\mu_K$, we can construct $N^K$ samples from their product $\mu$. Product-form estimators approximate $\mu$, and averages thereof, using the empirical distribution of these $N^K$ samples (as opposed to the default choice involving the empirical distribution of the $N$ `unpermuted' samples $(X^n_1,\dots,X^n_K)_{n=1}^N$). This increase in sample number leads to a decrease in estimator variance, and we say that product-form estimators are more \emph{statistically efficient} than their conventional counterparts.  The empirical distribution $\mu_\times^N$ of the $N^K$ samples can be expressed as the product of the `marginal' empirical distributions $\mu_1^N:=N^{-1}\sum_{n=1}^N\delta_{X^n_1},\dots,\mu_K^N:=N^{-1}\sum_{n=1}^N\delta_{X^n_K}$   (hence, the term `product-form estimators'). For this reason, product-form estimators embody the fork-and-join model~\cite{conway1963multiprocessor}: they `fork' the problem of approximating $\mu$ into that of approximating its marginals, solve these smaller problems by computing $\mu_1^N,\dots,\mu_K^N$ separately, and then `join' their solutions by taking their product, so obtaining a solution ($\mu_\times^N$) for the original problem.

Product-form estimators find their greatest use not in isolation but embedded within more sophisticated Monte Carlo methodology to tackle the aspects of the problem exhibiting product structure or conditional independences (e.g.\ \cite{Tran2013,Aitchison2019,Schmon2021,Kuntz2021}). Here, we study such embeddings within SMC algorithms, a class of Monte Carlo methods that approximate sequences of distributions defined on state spaces of increasing dimension by iteratively mutating and selecting weighted populations of samples or \emph{particles} (see, for example, \cite{chopin2020} for a recent book-length treatment). These algorithms are very widely employed in the context of hidden Markov models~\cite{doucet2011,Kantas2015,fearnhead2018} and, following \cite{chopin2002,delmoral2006sequential} and related works, have been shown to be a useful alternative to Markov chain Monte Carlo for approximating general distributions.

The DaC-SMC algorithm itself estimates targets on medium-to-high dimensional spaces that are products of collections of lower dimensional spaces. 
In order to accomplish this, an auxiliary tree is defined. Each node of this tree has some subset of the variables within the original model assigned to it and each subtree of the tree  has associated with it an intermediate model over the  variables assigned to the nodes in the subtree.  
The algorithm  progressively evolves clouds of weighted particles up the tree (from leaves to root) so that the corresponding weighted empirical distributions target the intermediate distributions. At each node, the clouds corresponding to the node's children are merged using a product-form estimator. Hence, the algorithm exemplifies the divide-and-conquer paradigm: carrying computations in the partial order dictated by the indexing tree, DaC-SMC recursively constructs product-form estimators for each of the lower-dimensional intermediate targets and, ultimately, for the original target itself.

While some preliminary work was done in~\cite{lindsten2017divide}, the algorithm's theoretical properties remain underexplored. The main aim of this paper is to fill this gap and firmly establish DaC-SMC as a well-founded algorithm: we show that it possesses the same fundamental theoretical properties as conventional SMC  (Section~\ref{sec:contributions}). We also have several other, more minor, associated aims: 
(Section~\ref{sec:lcdacsmc}) exploring efficient implementations of the algorithm; 
(Section~\ref{sec:variants}) identifying several interesting variants thereof in addition to those already pointed out in~\cite[Section~4]{lindsten2017divide}; 
(Section~\ref{sec:optimalprop}) deriving the optimal and `locally optimal' intermediate targets, auxiliary measures, and  proposal kernels and giving   guidance on their choice in practice; and
(Section~\ref{sec:dacvssmc}) formally arguing that DaC-SMC estimators are  more statistically efficient than their direct SMC analogues whenever both algorithms are placed on an equal footing (this does not, however, necessarily mean that DaC-SMC is more computationally efficient than SMC as the former's cost is generally greater than the latter's; see Sections~\ref{sec:lcdacsmc}~and~\ref{sec:dacvssmc} for more on this).

\subsection{Relation to the literature and contributions}\label{sec:contributions}
The most relevant publication is~\cite{lindsten2017divide} where the DaC-SMC algorithm was introduced, various variants were considered, several numerical examples were given, and it was shown that some forms of DaC-SMC estimators for the normalized targets are consistent and those for the normalizing constants are unbiased (cf.\ Section~\ref{sec:probdef} for definitions the `unnormalized targets', `normalized targets', and `normalizing constants'). 
Both in the original paper~\cite{lindsten2017divide} and also in later work, e.g.~\cite{ding2019,Corneflos2021}, the algorithm was found to have good empirical properties.
Applications of the product-form estimators it builds on can be found peppered throughout the Monte Carlo literature~\cite{Tran2013,lindsten2017divide,Aitchison2019,Schmon2021,Kuntz2021}, almost always unnamed and specialized to particular contexts. Their theoretical properties, at least in any generality, were seemingly only studied recently in~\cite{Kuntz2021} which established their   consistency, unbiasedness, asymptotically normality, and statistical efficiency, examined their computational efficiency, and considered several variants and extensions. Here, we extend these theoretical properties to DaC-SMC in much the same way that the properties of basic importance sampling were extended to standard SMC over the course of a series of papers~(e.g.~\cite{DelMoral1996,DelMoral1998,DelMoral1999,Crisan2002,Chopin2004,Kunsch2005,Crisan1997}).

The contributions of this paper are several:
\begin{itemize}
\item\emph{DaC-SMC's theoretical characterization.} We obtain  a comprehensive description  of the algorithm's theoretical properties (Section~\ref{sec:theory}). In particular, we show that DaC-SMC estimators for both the unnormalized and normalized targets retain the key properties of their SMC analogues: they are consistent (Theorems~\ref{THRM:LLN},~\ref{THRM:WEAK}), asymptotically normal (Theorem~\ref{THRM:CLT}), and satisfy $L^p$ inequalities (Theorem~\ref{THRM:LP}). In the case of the unnormalized targets, we also show that the estimators are unbiased (Theorem~\ref{THRM:UNBIAS}), while in that of the normalized ones, we show that the estimators' bias decays linearly with the number of particles (Theorem~\ref{THRM:BIAS}).
\item \emph{Optimal intermediate targets, auxiliary measures, and proposal kernels.} Using the asymptotic variance expressions obtained in Theorem~\ref{THRM:CLT}, we derive  (Section~\ref{sec:optimalprop}) the intermediate targets, auxiliary measures, and proposal kernels (cf.\ Sections~\ref{sec:probdef}--\ref{sec:dacalg} for the definitions of these objects) that lead to  zero \emph{asymptotic} variance DaC-SMC estimators for the normalizing constants. A bit unexpectedly, and in contrast with the standard case, we find that these choices do not necessarily lead to these estimators possessing zero \emph{finite-sample} variance. 
We also generalize the well-known characterization of the `locally optimal' auxiliary measures and proposal kernels of standard SMC to the DaC-SMC setting (Theorem~\ref{THRM:LOCALLY}).
\item\emph{Statistical efficiency of DaC-SMC.} Also using the asymptotic variances expressions in Theorem~\ref{THRM:CLT}, we show that DaC-SMC estimators are, at least asymptotically, more statistically efficient than their SMC counterparts (Section~\ref{sec:dacvssmc}): the asymptotic variances of the former are bounded above by those of the latter (Theorem~\ref{THRM:SMCVSDAC}).
\item\emph{Mathematical contributions.}  Our proofs for DaC-SMC's theoretical characterization combine well-known arguments used to analyse SMC algorithms with novel techniques required to deal with the product-form aspects of the estimators.  
To be more specific about the former, for the laws of large numbers and $L^p$ inequalities we adapt results from~\cite{Crisan2002,miguez2013convergence,DelMoral2004}, for the asymptotic normality we borrow from~\cite{DelMoral2004,DelMoral2013}, for the unbiasedness of the unnormalized target estimators we emulate~\cite{andrieu2010particle}, and for the bias bound of the normalized target estimator we follow~\cite{olsson2004bootstrap}. As for the latter, the new techniques involve carefully exploiting conditional independences and lead to results (Theorem~\ref{THRM:CROSS} and Lemmas~\ref{lem:prodlp},~\ref{lem:lpcoa},~\ref{lem:prodbias},~and~\ref{lem:wdecomp}) that are, to the best of our knowledge, unprecedented in the SMC literature; see Section~\ref{sec:mop} for further details. The consequences of our characterization (Theorems~\ref{THRM:LOCALLY}~and~\ref{THRM:SMCVSDAC}) also necessitate novel proofs.
\end{itemize}

\subsection{Paper structure}
Section~\ref{sec:background} reviews the two building blocks underpinning the DaC-SMC algorithm:  product-form estimators (Section~\ref{sec:pfe}) and sequential Monte Carlo (Section~\ref{sec:smc}). 
Section~\ref{sec:dac} presents the DaC-SMC algorithm: Section~\ref{sec:probdef} describes the target measures the algorithm approximates, Section~\ref{sec:dacalg} the algorithm itself, Section~\ref{sec:lcdacsmc} efficient implementations thereof, and Section~\ref{sec:variants} several interesting variations (Appendix~\ref{app:example} contains an example illustrating concepts discussed in Sections~\ref{sec:lcdacsmc},~\ref{sec:variants}). Section~\ref{sec:theory} details the algorithm's theoretical properties: Section~\ref{sec:mop} describes the methods of proof we use to argue the properties (to ease the reading, we defer our proofs to Appendices~\ref{app:notation}--\ref{app:dacvssmcproofs}),  Section~\ref{sec:optimalprop} studies the choice of intermediate targets, auxiliary measures, and proposal kernels, and Section~\ref{sec:dacvssmc} compares DaC-SMC with standard SMC theoretically and computationally. We close the paper with a discussion in Section~\ref{sec:discussion}. 
\section{Background: Product-form estimators and sequential Monte Carlo}\label{sec:background}
The DaC-SMC algorithm combines ideas from two distinct Monte Carlo methods: product-form estimators and sequential Monte Carlo algorithms. Here, we provide a brief overview of each.
\subsection{Product-form estimators}\label{sec:pfe}
Product-form estimators~\cite{Kuntz2021} are Monte Carlo methods that exploit any conditional independences present in the measure they target, or the proposal distribution they use, to achieve higher statistical efficiency than normal.   In particular, consider the task of estimating
$$\mu(\varphi):=\int\varphi(x)\mu(dx)$$
for a given probability distribution $\mu$ on a measurable space $(E,\cal{E})$ and a regular enough function $\varphi$. Suppose that $\mu$ is product-form: it is the product of $K$ distributions $\mu_1, \dots, \mu_K$ on respective measurable spaces $(E_1,\cal{E}_1),\ldots, (E_K,\cal{E}_K)$ satisfying $(E,\cal{E})=\left(\prod_{k=1}^KE_k,\prod_{k=1}^K\cal{E}_k\right)$, where $\prod_{k=1}^K \cal{E}_k$ denotes the corresponding product $\sigma$-algebra. To accomplish this task, suppose that we are given $N$  samples $(\bm{X}^n)_{n=1}^N=(X^n_1,\dots,X^n_K)_{n=1}^N$ independently drawn from $\mu$. Instead of approximating $\mu(\varphi)$ with the standard Monte Carlo estimator,
$$\mu^N(\varphi):=\frac{1}{N}\sum_{n=1}^N\varphi(\bm{X}^n),$$
consider using the \emph{product-form estimator},
\begin{equation*}
\mu_\times^N(\varphi) :=\frac{1}{N^K}\sum_{\bm{n}\in[N]^K}\varphi(\bm{X}^{\bm{n}}),
\end{equation*}
obtained by averaging over all of the $N^{K}$ element-wise permutations of samples $(\bm{X}^n)_{n=1}^N$,
$$\bm{X}^{\bm{n}}:=(X^{n_1}_1,\dots,X^{n_K}_K)\quad \forall \bm{n}=(n_1,\dots,n_K)\in[N]^K,\quad\text{where}\quad [N]:=\{1,\dots,N\}.$$

It is a strongly consistent, unbiased, and asymptotically normal estimator for $\mu(\varphi)$ (\cite[Theorem~1]{Kuntz2021}) that achieves a smaller variance than $\mu^N(\varphi)$ (\cite[Corollary~1]{Kuntz2021}),
$$\text{Var}(\mu^N_\times(\varphi))\leq \text{Var}(\mu^N(\varphi))\quad\forall N>0,$$
for all square $\mu$-integrable test functions $\varphi$. (For that matter, $\mu_\times^N(\varphi)$ yields the best  unbiased estimates of $\mu(\varphi)$ that can be achieved using only the knowledge that $\mu$ is product-form and $N$ i.i.d. samples drawn from $\mu$, cf.\ \cite[Theorem~2]{Kuntz2021}.) In a nutshell, the standard estimator uses the (non-product-form) empirical distribution of the unpermuted samples,
$$\mu^N:=\frac{1}{N}\sum_{n=1}^N\delta_{\bm{X}^n}$$
as an approximation to $\mu$, while the product-form estimator uses the (product-form) empirical distribution obtained by taking the product of the marginal empirical distributions,
$$\mu^N_\times:=\left(\frac{1}{N}\sum_{n_1=1}^N\delta_{X^{n_1}_1}\right)\times\dots\times\left(\frac{1}{N}\sum_{n_K=1}^N\delta_{X^{n_K}_K}\right)=\frac{1}{N^K}\sum_{\bm{n}\in[N]^K}\delta_{\bm{X}^{\bm{n}}};$$
and the latter provides a better approximation to $\mu$ for the reasons outlined in the introduction (see \cite[Section~2.2]{Kuntz2021} for further details).

The above applies similarly to the case where $\mu$ is an importance sampling (IS) proposal rather than a target. In particular, if our target $\gamma$ is a change of measure from $\mu$ with Radon-Nikodym derivative $w:=d\gamma/d\mu$, then
$$\gamma_\times^N(\varphi):=\mu_\times^N(w\varphi)=\frac{1}{N^K}\sum_{\bm{n}\in[N]^K}w(\bm{X}^{\bm{n}})\varphi(\bm{X}^{\bm{n}}),$$
is a consistent, unbiased, and asymptotically normal estimator for $\gamma(\varphi)=\mu(w\varphi)$ whose variance is bounded above by that of the standard IS estimator $\gamma^N(\varphi):=\mu^N(w\varphi)$; cf.\ \cite[Section~3.1]{Kuntz2021}. Because this is the case for all sufficiently integrable $\varphi$, we regard
\begin{equation}\label{eq:isprod}\gamma^N_\times(d\bm{x}):=w(\bm{x})\mu_{\times}^N(d\bm{x})=\frac{1}{N^K}\sum_{\bm{n}\in[N]^K}w(\bm{X}^{\bm{n}})\delta_{\bm{X}^{\bm{n}}}(d\bm{x})\end{equation}
as an approximation to $\gamma$. Somewhat abusively, we also refer to this type of approximation as a `product-form estimator' (e.g.\ $\gamma^N_\times$ is a product-form estimator for $\gamma$).
\subsection{Sequential Monte Carlo}\label{sec:smc}
Sequential Monte Carlo (SMC) algorithms constitute a class of Monte Carlo methods that, by progressively evolving clouds of weighted particles, approximate sequences $\rho_{0},\dots,\rho_T$ of intractable unnormalized measures defined on measurable spaces $(\bm{E}_0, \bm{\cal{E}}_0), \dots,(\bm{E}_T,\mathcal{E}_T)$ of increasing dimension (in particular, $\bm{E}_t=E_0\times\cdots\times E_t$ for some low-dimensional $E_0,\dots,E_T$ and $\bm{\cal{E}_t} = \cal{E}_0\times\cdots\times \cal{E}_t$ is the appropriate product $\sigma$-algebra). 
They also yield weighted particle approximations of the corresponding normalized distributions $\mu_{0},\dots,\mu_T$ and normalization constants $Z_0,\dots,Z_T$.

The basic SMC approach, often referred to as sequential importance resampling~\cite[p.\ 15]{doucet2011}, is detailed in Algorithm~\ref{alg:smc} (we use slightly unconventional notation here to keep things consistent with Section~\ref{sec:dac}). 
It is a simple iterative procedure: given a weighted-particle approximation $\rho_{t-1}^N$ to $\rho_{t-1}$, normalize $\rho_{t-1}^N$ to obtain an approximation $\mu^N_{t-1}$ of $\mu_{t-1}$, \emph{resample} $\mu^N_{t-1}:=\rho_{t-1}^N/\rho_{t-1}^N(\bm{E}_t)$ to obtain unweighted particles $\bm{X}_{t_-}^{1,N},\ldots,\bm{X}_{t_-}^{N,N}$ (whose empirical measure also approximates $\mu_{t-1}$), and extend the path of each resampled particle $\bm{X}_{t_-}^{n,N}$ by drawing $X_{t}^{n, N}$ from a proposal kernel $K_{t}(\bm{X}_{t_-}^{n,N},dx_{t})$ and setting $\bm{X}_{t}^{n, N}:=(X^{n,N}_{t}, \bm{X}_{t_-}^{n,N})$. Then, re-weight the particles using the Radon-Nykodim derivative $w_t:=d\rho_t/d( \rho_{t-1}\times K_{t})$ of $\rho_t(d\bm{x}_t)$ w.r.t.  $(\rho_{t-1}\times K_{t})(d\bm{x}_t):=\rho_{t-1}(d\bm{x}_{t-1})K_t(\bm{x}_{t-1},dx_t)$ to obtain an approximation of $\rho_t$:
\begin{align}
\label{eq:smc}
   \rho_t^N:= \frac{{Z}_{t-1}^N}{N} \sum_{n=1}^N w_t(\bm{X}_{t}^{n, N})\delta_{\bm{X}_{t}^{n, N}},
\end{align}
where $Z_{t-1}^N:= \rho_{t-1}^N(\bm{E}_{t-1})$ denotes the estimate of the normalizing constant $Z_{t-1}= \rho_{t-1}(\bm{E}_{t-1})$ afforded by $\rho^N_{t-1}$. 
To initialize the algorithm, we use importance sampling with a proposal distribution $K_0$ and weight function $w_0:=d\rho_0/dK_0$ to obtain an approximation $\rho^N_0:=N^{-1}\sum_{n=1}^Nw_0(\bm{X}^{n}_0)\delta_{\bm{X}^{n}_0}$ to $\rho_0$.

\begin{algorithm}[t]
\begin{algorithmic}[1]
\STATE{\textbf{Input:} number of particles $N$, proposal kernels $(K_t)_{t=0}^T$, and unnormalized targets $(\rho_t)_{t=0}^T$.} 
\STATE{\textit{Propose:} for $n\leq N$, draw $\bm{X}^{n,N}_0$ independently from $K_0$. }
\STATE{\textit{Correct:} compute $\rho^N_0=N^{-1}\sum_{n=1}^Nw_0(\bm{X}^{n}_0)\delta_{\bm{X}^{n}_0}$, where $w_0=d\rho_0/dK_0$,  and  $\mu^N_0=\rho^N_0/\rho^N_0(\bm{E}_0)$.}
\FOR{$t=1,\dots, T$}
\STATE{\textit{Resample:} for $n\leq N$, draw   $\bm{X}_{t_-}^{n,N}$ independently from $\mu_{t-1}^N$.}
\STATE{\textit{Mutate:} for $n\leq N$, draw $X^{n,N}_{t}$ independently from $K_{t}(\bm{X}_{t_-}^{n,N},dx_t)$ and set $\bm{X}_{t}^{n,N}=(X^{n,N}_{t}, \bm{X}_{t_-}^{n,N})$.}
\STATE{\textit{Correct:} compute $\rho_t^N$  in~\eqref{eq:smc} and its normalization   $\mu^N_t=\rho_t^N/\rho_t^N(\bm{E}_t)$.}
\ENDFOR
\end{algorithmic}
 \caption{Sequential importance resampling}
 \label{alg:smc}
\end{algorithm}

Resampling is the process of stochastically replicating some particles within a weighted collection in such a way that the expected number of offspring of each particle is proportional to its weight. By doing so, we obtain an unweighted sample suitable for approximating the same target distribution as the original weighted sample. The simplest way to achieve this, known as multinomial resampling because the vector of offspring numbers follows a multinomial distribution, is to independently draw $N$ times from the weighed empirical measure of the original sample. To not complicate this paper's exposition and proofs, and to keep the focus on the differences between standard SMC and DaC-SMC, we restrict ourselves throughout to the case of multinomial resampling. However, 
using more sophisticated resampling schemes within the DaC-SMC setting is sensible, incurs no further complications than for SMC, and can in principle be understood using similar theoretical techniques to those employed in the conventional setting (see, in particular, \cite{Gerber2019}).

SMC algorithms (meaning iterative importance sampling schemes with resampling steps) emerged as heuristic schemes in a variety of applications including filtering (cf. \cite{Moral2014} for a list), and have been rigorously studied as interacting particle systems tied to Feynman-Kac formulae since~\cite{DelMoral1996}; see \cite[Section 1.1]{DelMoral2004} for a comprehensive account.  
The filtering problem has an intrinsic  sequential structure and the sequence $\rho_0,\dots,\rho_T$ arises naturally. For general distributions $\mu$ lacking this sort of structure, SMC can instead be applied by building an `artificial' sequence of distributions such that $\mu_T$ (or one of its marginals) coincides with $\mu$. This can be achieved by incrementally morphing an easy-to-sample-from   distribution into $\mu$, often by tempering~\cite{neal2001annealed} or sequentially introducing observations to transition from a prior to a posterior~\cite{chopin2002}. At first glance,  Algorithm~\ref{alg:smc} may not appear applicable in such settings because the artificial distributions are all defined on the same space. However, it is possible to construct auxiliary sequences of distributions on spaces of increasing dimension such that the $t^\textrm{th}$ distribution of interest coincides with the $t^\textrm{th}$ `coordinate marginal' of the corresponding auxiliary distribution, hence allowing \emph{exactly} this algorithm to be applied to these problems~\cite{delmoral2006sequential}. The same kind of constructions can be used in the DaC-SMC framework (see~\cite[Section~4.2]{lindsten2017divide} for a brief discussion) and we focus on analysing the fundamental algorithm (which, combined with appropriately extended target distributions, applies directly to these seemingly more complex settings).

We mention one variant of the SMC algorithm which is of particular relevance to what follows: the auxiliary SMC (ASMC) scheme originally proposed for filtering~\cite{pitt99filtering}. Although it was initially presented as an alternative proposal method involving an auxiliary variable construction, this scheme can be understood as the standard algorithm  with   an additional re-weighting prior to the resampling step~\cite{carpenter1999improved}---the motivation behind  the  additional re-weighting being to incorporate as much information about the next distribution in the sequence as possible  \emph{prior} to resampling. Simply put, it is an SMC  algorithm that targets a slightly different sequence of distributions and  employs an extra importance sampling correction to obtain estimates with respect to the distributions of interest~\cite{johansen2008note}. 

In particular, we introduce \emph{auxiliary measures} $\gamma_{1_-},\ldots,\gamma_{T_-}$ on $\bm{E}_0,\dots,\bm{E}_{T-1}$, and their normalized counterparts $\pi_{1_-},\ldots,\pi_{T_-}$, and run Algorithm~\ref{alg:smc} with these measures as the targets. That is, with the weights used for resampling at time $t$ set to $w_{t-}: = d\gamma_{t_-}/d\gamma_{t-1}$, where $\gamma_0:=K_0$ and $\gamma_t := \gamma_{t_-} \times K_t$ for $t=1,\dots,T$. This gives rise to the approximation 
\begin{align}
\label{eq:asmc}
   \gamma_{t_-}^N:= \frac{\cal{Z}_{t-1}^N}{N} \sum_{n=1}^N w_{t_-}(\bm{X}_{t-1}^{n, N})\delta_{\bm{X}_{t-1}^{n, N}}
\end{align}
of the auxiliary measure $\gamma_{t_-}$, where $\cal{Z}_{0}^N :=1$ and  $\cal{Z}_{t}^N := \gamma_{t_-}^N(\bm{E}_{t-1})$ for $t=1,\dots, T$. 
In the ASMC literature, it is common to find the correction step at the beginning of iteration $t$ instead of the end of iteration $t-1$ (cf.\ Algorithm~\ref{alg:asmc}). Such cyclical permutations of the resample, mutate, and correct steps have no impact beyond labeling; we adopt this ordering because it eases the presentation of the DaC-SMC algorithm in Section~\ref{sec:dac}.

\begin{algorithm}[t]
\begin{algorithmic}[1]
\STATE{\textbf{Input:} number of particles $N$, proposal kernels $(K_t)_{t=0}^T$, and auxiliary measures $(\gamma_{t_-})_{t=1}^T$.} 
\STATE{\textit{Propose:} for $n\leq N$, draw $\bm{X}^{n,N}_0$ independently from $K_0$. }
\STATE{\textit{Compute:} $\gamma_0^N=N^{-1}\sum_{n=1}^N\delta_{\bm{X}^{n,N}_0}$.}
\FOR{$t=1,\dots, T$}
\STATE{\textit{Correct:} re-weight $\gamma_{t-1}^N$ with $w_{t_-}$ to obtain $\gamma_{t_-}^N=w_{t_-}\gamma_{t-1}^N$ and  $\pi_{t_-}^N=\gamma_{t_-}^N/\gamma_{t_-}^N(\bm{E}_{t-1})$.}
\STATE{\textit{Resample:} for $n\leq N$,  draw $\bm{X}_{t_-}^{n,N}$ independently from $\pi_{t_-}^N$.}
\STATE{\textit{Mutate:} for $n\leq N$, draw $X^{n,N}_{t}$ independently from $K_{t}(\bm{X}_{t_-}^{n,N},dx_t)$ and set $\bm{X}_{t}^{n,N}=(X^{n,N}_{t}, \bm{X}_{t_-}^{n,N})$.}
\STATE{\textit{Compute:} $\gamma_{t}^N=N^{-1}\cal{Z}_{t}^N\sum_{n=1}^N\delta_{\bm{X}_{t}^{n,N}}$ where $\cal{Z}_{t}^N=\gamma_{t_-}^N(\bm{E}_{t-1})$.}
\ENDFOR
\newline {\textbf{Note:} At each step $t$, one can use~\eqref{eq:asmcests} to compute estimates of $\rho_t$, $Z_t$, and $\mu_t$.}
\end{algorithmic}
\caption{Auxiliary sequential importance resampling}
\label{alg:asmc}
\end{algorithm}
 
The \emph{extended auxiliary measure} $\gamma_t$ is then approximated by
$$  \gamma_{t}^N:= \frac{\cal{Z}_{t}^N}{N} \sum_{n=1}^N\delta_{\bm{X}_{t}^{n,N}},$$
where, in the case of multinomial resampling, $\bm{X}_t^{1,N},\dots,\bm{X}_t^{N,N}$ denote samples drawn independently  from the normalization of $\gamma_{t_-}^N \times K_t$. (We note in passing that $\gamma_{t-1}$ and $\gamma_{t_-}$ can be interpreted as the `predictive' and `updated' time marginals of a  Feynman-Kac flow in the sense of~\cite{DelMoral2004}.) The above is then corrected using importance sampling and the weight function  $w_t := d\rho_t/d\gamma_t
$ to obtain approximations of the objects of interest $\rho_t$, $Z_t$, and $\mu_t$:
\begin{align}
  \rho_t^N(d\bm{x}_t) := w_t(\bm{x}_t)\gamma_t^N(d\bm{x}_t),\quad {Z}_t^N := \rho_t^N(\bm{E}_t),\quad 
  \mu_t^N(d\bm{x}_t) := \frac{1}{{Z}_t^N} \rho_t^N(d\bm{x}_t)\label{eq:asmcests},
\end{align}
with $\gamma_0^N:=N^{-1}\sum_{n=1}^N\delta_{\bm{X}^{n,N}_0}$ denoting the empirical distribution of the samples drawn from $\gamma_0=K_0$. 
This construction gives rise to the  auxiliary SMC algorithm described in Algorithm~\ref{alg:asmc}. 
\section{The divide-and-conquer sequential Monte Carlo  algorithm} \label{sec:dac}
In this section, we define our problem setting (Section~\ref{sec:probdef}), give the DaC-SMC algorithm (Section~\ref{sec:dacalg}), consider its implementation (Section~\ref{sec:lcdacsmc}), and discuss various variants thereof (Section~\ref{sec:variants}).
\subsection{Problem definition}\label{sec:probdef} 
We tackle a generalization of the problem in Section~\ref{sec:smc}: obtaining weighted-particle approximations to a collection $(\rho_u)_{u\in\mathbb{T}}$ of intractable unnormalized (but finite) measures  indexed by the nodes of a finite rooted tree $\mathbb{T}$. In doing so, we also obtain tractable approximations to the corresponding collections of normalizing constants $(Z_u)_{u\in\mathbb{T}}$ and normalized measures $(\mu_u)_{u\in\mathbb{T}}$. Typically, our main interest will lie in approximating the target $\rho_{\mathfrak{r}}$ indexed by the tree's root $\mathfrak{r}$ and all other targets have been introduced to facilitate this process. Hence, we refer to $\rho_{\mathfrak{r}}$ and $\mu_{\mathfrak{r}}$ as the `final targets' and to all other $\rho_u$ and $\mu_u$ as the `intermediate targets'. 
\label{page:8}

Although  this may seem a somewhat specialized setting at first glance, a great many problems (including, for example, essentially any in which one wishes to approximate a probability distribution over $\mathbb{R}^d$) may be viewed within this framework. 
In some cases, the tree structure is intrinsic to the model at hand and there are natural candidates for the intermediate targets (e.g.\ \cite{jewell2015,paige2016inference,Gelman2006} or Example~\ref{ex:schools} below). In others, the tree decomposition and intermediate targets are entirely artificial constructs   introduced for the purpose of providing a computationally tractable approach to approximating the final targets (e.g.\ see \cite[Section 3.4]{lindsten2017divide} or \cite{ding2019}) and DaC-SMC serves as a generalization of the approaches in~\cite{delmoral2006sequential, chopin2002}.

In all of these cases, the intermediate targets transition incrementally from a collection of measures on low-dimensional spaces indexed by $\mathbb{T}$'s leaves to the final target. Hence, the measures are defined on spaces whose dimension grows as we progress up the tree. In particular, the underlying measurable space $(\bm{E}_u,\bm{\cal{E}}_u)$ for $\rho_u$ is defined by taking the corresponding partial product of a collection $(E_v,\cal{E}_v)_{v\in \mathbb{T}}$ of low-dimensional spaces indexed by  $\mathbb{T}$'s nodes:
$$\bm{E}_u=\prod_{v\in\mathbb{T}_u}E_v,\quad \bm{\cal{E}}_u=\prod_{v\in\mathbb{T}_u}\cal{E}_v,\quad\forall u\in\mathbb{T},$$
where $\mathbb{T}_u$ denotes  the sub-tree  of $\mathbb{T}$ rooted at $u$ (obtained by removing all nodes from $\mathbb{T}$ except for $u$ and its descendants).\label{page:subtree}

Similarly to Algorithm~\ref{alg:asmc}, we introduce a collection $(\gamma_{u_-})_{u\in\mathbb{T}^{\not \partial}}$ of \emph{auxiliary measures} \label{page:auxiliary}indexed by set $\mathbb{T}^{\not \partial}$ non-leaf nodes.  In particular, for any such node $u$ with set of children $\cal{C}_u$, $\gamma_{u_-}$ is a finite measure on the corresponding intermediate product space:\label{page:spaces}
\begin{align*}&\bm{E}_{\cal{C}_u}:=\prod_{v\in\mathbb{T}_u^{\not u}}E_v,\quad \bm{\cal{E}}_{\cal{C}_u}:=\prod_{v\in\mathbb{T}_u^{\not u}}\cal{E}_v,\end{align*}
where $\mathbb{T}_u^{\not u}:=\mathbb{T}_u\backslash \{u\}$. These auxiliary distributions fulfill essentially the same role as the auxiliary distributions in the ASMC  algorithm: introducing as much information as possible about the measure $\rho_u$ associated with the current node $u$   \emph{before} resampling. As we will see in the next section, this mechanism is particularly important in the divide and conquer context because it allows us to correlate the components of the resampled particles associated with $u$'s children.

\begin{remark}[A connection with SMC]\label{rem:smc1}If the tree  is a line (i.e.\ each node has a single child, except for the first which has none), then our setting here reduces to that of standard ASMC (Algorithm~\ref{alg:asmc}). If we choose auxiliary distributions $\gamma_{t_-} := \rho_{t-1}$, then it further reduces to a standard SMC (Algorithm~\ref{alg:smc}).
\end{remark}

For the reasons discussed in Section~\ref{sec:optimalprop}, 
we would like to choose the intermediate targets and auxiliary measures to be as close as possible to the corresponding marginals of the final target $\rho_{\mathfrak{r}}$. 
Of course, in general, this is more easily said than done and, in practice, $\rho_u$ and $\gamma_{u_-}$ will often be obtained by ignoring~\cite[Section 3.4]{lindsten2017divide}, rather than `integrating-out', the terms in $\rho_{\mathfrak{r}}$ featuring variables associated with $u$'s ancestors (and $u$ itself in the case of $\gamma_{u_-}$). To illustrate some of these concepts, consider the following simplified version of the example in \cite[Section~5.2]{lindsten2017divide}:

\begin{example}[Simplified  mathematics test dataset]\label{ex:schools}Suppose we have the following data: the number $M_{sy}$  of students who took a certain exam in school $s$ and year $y$ and the number $m_{sy}$ which passed the exam, for all years $y$ that the exam was taken in school $s$ and all schools $s$ in a given city (we denote the set of the former by $\cal{Y}_s$ and that of the latter by $\cal{S}$). A standard hierarchical model~\cite{Gelman2006} for inference\footnote{For instance, our aim may be to infer whether there are any significant differences in pass rates across the schools and, if so, which schools perform better. Or whether these rates are changing over time and whether the trends are city-wide or school-specific.} on such a dataset is as follows:  $m_{sy}$ is binomially distributed with success probability $\alpha(\theta_{sy})$ (i.e.\ $m_{sy}$ has law $\cal{B}(\cdot;M_{sy},\alpha(\theta_{sy}))$), where $\theta_{sy}$ denotes a latent parameter and $\alpha$ denotes the standard logistic function.  The parameters are related as follows:
$$\theta_{sy}=\theta_s+\varepsilon_{sy}\enskip\forall y\in\cal{Y}_s,\enskip s\in\cal{S},\quad \theta_{s}=\theta_{\mathfrak{r}}+\varepsilon_{s}\enskip \forall  s\in\cal{S},$$
where $(\theta_s)_{s\in\cal{S}}$ and $\theta_\mathfrak{r}$ denote further latent parameters and $((\varepsilon_{sy})_{y\in\cal{Y}_s},\varepsilon_{s})_{s\in\cal{S}}$ a collection of independent noises with $\varepsilon_{sy}\sim\cal{N}(0,\sigma_s^2)$ and $\varepsilon_{s}\sim\cal{N}(0,\sigma^2_\mathfrak{r})$ for some unknown variances $(\sigma_s^2)_{s\in\cal{S}}$ and $\sigma^2_\mathfrak{r}$. For each parameter variable, we choose as a prior the pullback $g(d\theta)=\alpha(\theta)(1-\alpha(\theta))d\theta$ of the uniform distribution on $(0,1)$ via $\alpha$; while for each variance variable, we choose a unit-mean exponential distribution $f(d\sigma^2;1)$. The full (unnormalized) posterior is then given by
\begin{align}f(d\sigma^2_{\mathfrak{r}};1)g(d\theta_{\mathfrak{r}})\prod_{s\in\cal{S}}\Big[&\cal{N}(\theta_{s}-\theta_{\mathfrak{r}};0,\sigma^2_{\mathfrak{r}})f(d\sigma^2_s;1)g(d\theta_s)\label{eq:schools}\\
&\times\prod_{y\in\cal{Y}_s}\cal{N}(\theta_{sy}-\theta_s;0,\sigma_s^2)\cal{B}(m_{sy};M_{sy},\alpha(\theta_{sy}))g(d\theta_{sy})\Big].\nonumber\end{align}

This distribution's dependence structure is encoded by a tree $\mathbb{T}$ whose paths from root to leaves have the following form: root, school, year. In other words, $\mathbb{T}$'s root node $\mathfrak{r}$ has a child for each school ($\cal{C}_{\mathfrak{r}}=\cal{S}$), each of these has a child for each year the exam took place in that school ($\cal{C}_s=\cal{Y}_s$ for all $s$ in $\cal{S}$), and the latter are all leaves (note that we treat elements in different $\cal{Y}_s$s as distinct nodes, even if they represent the same year). With each leaf node, we associate the corresponding latent variable (so that $E_u:=\r$ if $u$ is a leaf), and, with each non-leaf node, we associate the corresponding parameter variable and unknown variance (so that $E_u:=\r\times[0,\infty)$ if $u$ is not a leaf). 

We set the final target $\rho_{\mathfrak{r}}$ to be our measure of interest~\eqref{eq:schools}. Analytical expressions for its marginals are unavailable, so we instead pick the intermediate targets as follows: we obtain $\rho_u$ by  removing from~\eqref{eq:schools} all terms featuring variables not associated with $u$ or its descendants:
\begin{align}
\rho_{s}(d\theta_s,d\sigma^2_s,(d\theta_{sy})_{y\in\cal{Y}_s})&:=f(d\sigma^2_s;1)g(d\theta_s)\label{eq:schools2}\\
&\quad\times\prod_{y\in\cal{Y}_s}\cal{N}(\theta_{sy}-\theta_s;0,\sigma_s^2)\cal{B}(m_{sy};M_{sy},\alpha(\theta_{sy}))g(d\theta_{sy}),\nonumber\\
\rho_{sy}(d\theta_{sy})&:=\cal{B}(m_{sy};M_{sy},\alpha(\theta_{sy}))g(d\theta_{sy})\quad\forall y\in\cal{Y}_s,\nonumber
\end{align}
for all $s$ in $\cal{S}$. It is also not possible to analytically integrate over $(\theta_{\mathfrak{r}},\sigma^2_{\mathfrak{r}})$ in~\eqref{eq:schools} or $(\theta_s,\sigma^2_s)$ in~\eqref{eq:schools2}, so we instead pick the   auxiliary measures  by further  doing away with the terms in~(\ref{eq:schools},\ref{eq:schools2}) involving the variables associated with the indexing node:
$$\gamma_{\mathfrak{r}_-}:=\prod_{s\in\cal{S}}\rho_s,\quad\gamma_{s_-}:=\prod_{y\in\cal{Y}_s}\rho_{sy}\enskip\forall s\in\cal{S}. $$
\end{example}

\begin{table}
\begin{tabular}{{cllcl}}
Symbol &  Description & Value & Loc.& Approximation\\
\hline\noalign{\smallskip}
$\mathbb{T}$&Index tree (subtree rooted at $u$)&--&pp.~\pageref{page:8},\pageref{page:tree}&--\\
$\mathfrak{r}$&$\mathbb{T}$'s root&--&pp.~\pageref{page:8},\pageref{page:tree}&--\\
$\mathbb{T}_u$&Subtree rooted at $u$&--&p.~\pageref{page:subtree}&--\\
$\mathbb{T}^{\partial}\enskip (\mathbb{T}^{\not\partial})$&Set of leaves (its complement)&--$\enskip(\mathbb{T}\backslash \mathbb{T}^{\partial})$&p.~\pageref{page:tree}&--\\
$\cal{C}_u \enskip (c_u)$&Set of $u$'s children (its cardinality)&--&p.~\pageref{page:tree}&--\\
$\mathbb{T}_u^{\not v}\enskip (\cal{C}_u^{\not v})$&$\mathbb{T}_u$ excluding $v$ ($\cal{C}_u$ excluding $v$)&--&--&--\\
\noalign{\smallskip}
$E_u$&Low-dimensional space&--&p.~\pageref{page:spaces}&--\\
$\cal{E}_u$&Sigma-algebra on $E_u$&--&p.~\pageref{page:spaces}&--\\
$\bm{E}_u$&Intermediate space&$\prod_{v\in\mathbb{T}_u}E_v$&--&--\\
$\bm{\cal{E}}_u$&Sigma-algebra on $\bm{E}_u$&$\prod_{v\in\mathbb{T}_u}\cal{E}_v$&--&--\\
\noalign{\smallskip}
$\rho_{u}$ &  Target measure & -- & p.~\pageref{page:8}&$\rho_{u}^N$ in eq.~\eqref{eq:adacests}\\
$Z_{u}$ & $\rho_u$'s normalizing constant& $\rho_u(\bm{E}_u)$& p.~\pageref{page:8}&$Z_{u}^N=\rho_u^N(\bm{E}_u)$ in eq.~\eqref{eq:adacests}\\
$\mu_{u}$ & Normalized target measure & $\rho_u/Z_u$& p.~\pageref{page:8}&$\mu_{u}^N=\rho_u^N/Z_u^N$ in eq.~\eqref{eq:adacests}\\
$\gamma_{u_-}$ & Auxiliary measure & -- & p.~\pageref{page:auxiliary}&$\gamma_{u_-}^N=w_{u_-}\gamma_{\cal{C}_u}^N$ in eq.~\eqref{eq:gammauN}\\
$\mathcal{Z}_{u}$ & $\gamma_{u_-}$'s normalizing constant & $\gamma_{u_-}(\bm{E}_{\cal{C}_u})$ & eq.~\eqref{eq:piumn}&$\cal{Z}_u^N=\gamma_{u_-}^N(\bm{E}_{\cal{C}_u})$ in eq.~\eqref{eq:piumn}\\
$\pi_{u_-}$ & Normalized auxiliary measure & $\gamma_{u_-}/\mathcal{Z}_u$& eq.~\eqref{eq:piumn}& $\pi_{u_-}^N=\gamma_{u_-}^N/\mathcal{Z}_u^N$ in eq.~\eqref{eq:piumn}\\
$K_u$ & Proposal distribution or kernel & -- & p.~\pageref{page:10}&--\\
$\gamma_u$ & Extended auxiliary measure & $\gamma_{u_-}\times K_u$& eq.~\eqref{eq:otimes}&$\gamma_u^N=\gamma_{u_-}^N\times K_u$ in eq.~\eqref{eq:gammauN}\\
$w_{u_{-}}$ & Auxiliary weights & $d\gamma_{u-}/d\gamma_{\cal{C}_u}$& p.~\pageref{page:auxiliary_weights}&--\\
$w_u$ & Inferential weights& $d\rho_u/d\gamma_u$& p.~\pageref{page:inferential_weights}&--\\
\noalign{\smallskip}
$N$&Number of particles&--&&--\\
$\bm{X}_{\cal{C}_u}^{n,N}$&Tuple of mutated particles&$(\bm{X}_{v}^{n,N})_{v\in\cal{C}_u}$&--&--\\
$\bm{X}_{\cal{C}_u}^{\bm{n},N}$&Permuted tuple of mut. particles&$(\bm{X}_{v}^{n_v,N})_{v\in\cal{C}_u}$&--&--\\
$\bm{X}_{u_-}^{n,N}$&Resampled particle&$\sim\pi_{u_-}^N$&p.~\pageref{page:resammutapar}&--\\
$X_u^{n,N}$&Mutation&$\sim K_u(\bm{X}_{u_-}^{n,N},\cdot)$&p.~\pageref{page:resammutapar}&--\\
$\bm{X}_u^{n,N}$&Mutated particle&$(X_u^{n,N},\bm{X}_{u_-}^{n,N})$&p.~\pageref{page:resammutapar}&--\\
\end{tabular}
\caption{Notation introduced throughout Section~\ref{sec:probdef}--\ref{sec:dacalg}. Other commonly used symbols are obtained by combining the above with the product notation introduced in p.~\pageref{page:products} (e.g., $\bm{E}_{\cal{C}_u}=\prod_{v\in\cal{C}_u}\bm{E}_v$ and $\gamma_{\cal{C}_u}=\prod_{v\in\cal{C}_u}\gamma_u$).}
\label{tab:notation}
\end{table}

\paragraph*{Notation} We finish this section with some notation that we use throughout the paper:
\begin{itemize}
\item (Trees) \label{page:tree}By a `tree' $\mathbb{T}$ we mean a rooted tree: a connected directed acyclic graph in which no node has more than one parent and possessing a single root node, $\rf$, that has no parent. 
We denote the number of children a node $u$ has as  $c_u$, refer to them as $u1,u2,\dots,uc_{u}$, 
and denote the set $\{u1,\dots, uc_u\}$ of them by $\cal{C}_u$. For any subset $A$ of the tree $\mathbb{T}$, we use $\mmag{A}$ to denote its cardinality, $A^{\partial}:=\{v\in A:c_v=0\}$ to denote the set of all leaves in $A$, and $A^{\not\partial}:=A\backslash A^{\partial}$ that of all other nodes in $A$. In particular, $\mathbb{T}^\partial$ denotes the set of $\mathbb{T}$'s leaves and $\mathbb{T}^{\not \partial}$ that of  all non-leaf nodes.
\item (Products)\label{page:products}  We assign the subscript `$_A$' to a symbol, where $A$ is any  subset of  $\mathbb{T}$, if the object represented by the symbol depends on the nodes in $A$. If the symbol is bold, the object also depends on the descendants of the nodes in $A$.  With few exceptions, this dependence takes the form of a product. For example,
\begin{align*}
E_A:=&\prod_{u\in A}E_u,&
\bm{E}_{A}:=&E_{\cup_{u\in A}\mathbb{T}_u},&
\cal{E}_A:=&\prod_{u\in A}\cal{E}_u,&
\bm{\cal{E}}_{A}:=&\cal{E}_{\cup_{u\in A}\mathbb{T}_u}, \\
\rho_{A}:=&\prod_{u\in A}\rho_{u},&
\mu_{A}:=&\prod_{u\in A}\mu_{u},&
Z_A:=&\prod_{u\in A}Z_{u}.& &
\end{align*}
for any non-empty subset $A$ of $\mathbb{T}$. 
The exceptions are $x_A$ and $\bm{x}_A$  which denote elements in $E_A$  and $\bm{E}_A$ rather than any sort of product over $A$.
\end{itemize}
To facilitate the reading, we summarize in Table~\ref{tab:notation} notation frequently used throughout the paper (many of these objects will be introduced in the following section).
\subsection{The algorithm}\label{sec:dacalg}
To introduce the divide-and-conquer SMC algorithm, we require one last ingredient: the generalization of the proposal distributions and kernels used in SMC/ASMC (cf.\ Section~\ref{sec:smc}) to initialize and extend the particles paths. In our case, we assign to each leaf $u$ a distribution $K_u:\cal{E}_{u}\to[0,1]$ on the low-dimensional space $E_u$. To each non-leaf node $u$, we instead assign a Markov kernel $K_u:\bm{E}_{\cal{C}_u}\times\cal{E}_u\to[0,1]$ mapping from the product $\bm{E}_{\cal{C}_u}=\prod_{v\in\cal{C}_u}\bm{E}_v$ of the intermediate spaces $\bm{E}_v$ indexed by $u$'s children to the low-dimensional space $E_u$ indexed by $u$ itself. For the reasons discussed in Section~\ref{sec:optimalprop}, $K_u$ \label{page:10}should be chosen as close as possible to normalized target $\mu_u$, if $u$ is a leaf, or to the conditional distribution over $E_u$ under $\mu_u$ given values for its remaining coordinates, if $u$ is not a leaf.

DaC-SMC is a particle approximation to the measure-valued recursion in Algorithm~\ref{alg:adacsmcexact}. The recursion starts at leaves and works its way up the tree. At a leaf $u$ it just returns the corresponding proposal distribution $K_u$. At a non-leaf node $u$, it calls itself on each child $v$ of $u$ and recovers the corresponding \emph{extended auxiliary measure} $\gamma_{v}$ obtained by taking the outer product of the child's auxiliary measure $\gamma_{v_-}$ and it's proposal kernel $K_v$:
\begin{align}\label{eq:otimes}
\gamma_v(d\bm{x}_v) :=(\gamma_{v_-}\times K_{v})(d\bm{x}_v)=\gamma_{v_-}(d\bm{x}_{\cal{C}_v})K_{v}(\bm{x}_{\cal{C}_v},dx_{v})\quad\forall v\in\mathbb{T}^{\not\partial}.
\end{align}\label{page:extended_auxiliary}
(If $v$ is a leaf then the algorithm just retrieves the  proposal distribution $\gamma_v:=K_v$ index by $v$.) The recursion then takes the product $\gamma_{\cal{C}_u}=\prod_{v\in\cal{C}_u}\gamma_v$\label{page:product_gamma} of these and re-weights it using $w_{u_-} := d{\gamma_{u_-}}/d\gamma_{\cal{C}_u}$\label{page:auxiliary_weights} to obtain $u$'s auxiliary measure $\gamma_{u_-}=w_{u_-}\gamma_{\cal{C}_u}$. Lastly, it takes the product of $\gamma_{u_-}$ with the node's proposal kernel $K_u$ to obtain the node's extended auxiliary measure $\gamma_u$ and returns $\gamma_u$. At any node $u$, re-weighting $\gamma_u$ using $w_u:=d\rho_u/d\gamma_u$\label{page:inferential_weights} yields the target $\rho_u=w_u\gamma_u$, its normalization constant $Z_u=\rho_u(\bm{E}_u)$, and its normalization $\mu_u=\rho_u/Z_u$.

\begin{algorithm}[t]
\begin{algorithmic}[1]
\STATE{\textbf{Input:} proposal kernels $(K_u)_{u\in\mathbb{T}}$ and auxiliary measures $(\gamma_{u_-})_{u\in\mathbb{T}^{\not\partial}}$.} 
\IF{$u$ is a leaf (i.e.\ $u\in\mathbb{T}^\partial$)}
\STATE \emph{Return:} $\gamma_u = K_u$.
\ELSE
\FOR{$v$ in $\cal{C}_u$}
\STATE{\textit{Recurse:} set $\gamma_v =\text{rec}(v)$.}
\ENDFOR
\STATE{\textit{Product:} compute $\gamma_{\cal{C}_u} = \prod_{v\in\cal{C}_u} \gamma_v$.}
\STATE{\textit{Correct:} re-weight $\gamma_{\cal{C}_u}$ with $w_{u_-}= d{\gamma_{u_-}}/d\gamma_{\cal{C}_u}$ to obtain $\gamma_{u_-}= w_{u_-}\gamma_{\cal{C}_u}$.}
\STATE{\textit{Return:} $\gamma_u = \gamma_{u_-} \times K_u$.}
\ENDIF
\newline {\textbf{Note:} Re-weighting $\gamma_u$ with $w_u=d\rho_u/d\gamma_u$, we obtain the target $\rho_u=w_u\gamma_u$, its normalizing constant $Z_u=\int w_ud\gamma_u$, and its normalization $\mu_u=w_u\gamma_u/\int w_ud\gamma_u$.}
\end{algorithmic}
\caption{rec($u$): the measure-valued recursion approximated by dac\_smc$(u)$ for $u$ in $\mathbb{T}$.}\label{alg:adacsmcexact}
\end{algorithm}

The changes of measure in the recursion are typically intractable, so DaC-SMC\footnote{\label{footnote}Note that, in the original paper~\cite{lindsten2017divide}, Algorithm~\ref{alg:adacsmc} was referred to as `DaC-SMC with mixture resampling' while the term `DaC-SMC' was reserved for cases where the auxiliary measures are product-form and the algorithm's running time can be sped up through careful implementation (cf.\ Section~\ref{sec:lcdacsmc}). For all, at least  theoretical, purposes  the latter is a special case of the former and we simply refer to both as the `DaC-SMC Algorithm'.} (Algorithm~\ref{alg:adacsmc}) instead approximates each of these steps using particles. At a leaf $u$, the algorithm draws $N$ \emph{particles} $\bm{X}_u^{1,N},\dots,\bm{X}_u^{N,N}$ from the leaf's proposal distribution $K_u$ and returns the corresponding unweighted sample approximation  $\gamma_u^N:=N^{-1}\sum_{n=1}^N\delta_{\bm{X}_u^{n,N}}$ (with unit mass $\cal{Z}_u^N:=1$)  to $K_u$. At a non-leaf node $u$, it first calls itself on each child $v$ of $u$ to obtain an unweighted particle approximation $\gamma_{v}^N=\cal{Z}_v^NN^{-1}\sum_{n=1}^N\delta_{\bm{X}_v^{n,N}}$ (of mass $\cal{Z}_v^N$) to the child's extended auxiliary measure $\gamma_v$. It then takes the product $\gamma_{\cal{C}_u}^N := \prod_{v \in \cal{C}_u} \gamma_{v}^N$\label{page:product_gamma} of these and re-weighs it using $w_{u_-} = d{\gamma_{u_-}}/d\gamma_{\cal{C}_u}$\label{page:auxiliary_weights} to obtain a product-form estimator (cf.\ \eqref{eq:isprod}) for the auxiliary distribution $\gamma_{u_-}$ associated with node $u$:
\begin{equation}\label{eq:gammauN}\gamma_{u_-}^N:=w_{u_-}\gamma_{\cal{C}_u}^N=\frac{\cal{Z}^N_{\cal{C}_{u}}}{N^{c_u}}\sum_{\bm{n}\in[N]^{c_u}}w_{u_-}(\bm{X}_{\cal{C}_u}^{\bm{n},N})\delta_{\bm{X}_{\cal{C}_u}^{\bm{n},N}}\end{equation}\label{page:gammauN}
where $\cal{Z}^N_{\cal{C}_{u}}:=\prod_{v\in \cal{C}_u}\cal{Z}_{v}^N$ and, for each $\bm{n}=(n_1,\dots,n_{c_u})$ in $[N]^{c_u}$, $\bm{X}^{\bm{n},N}_{\cal{C}_u}$ denotes the `permuted particle' $(\bm{X}^{n_1}_{u1},\dots,\bm{X}^{n_{c_u}}_{uc_u})$.  Just as in the standard ASMC case (Algorithm~\ref{alg:asmc}), the algorithm proceeds by normalizing $\gamma_{u_-}^N$  to obtain an approximation $\pi_{u_-}^N$ of the  \emph{normalized auxiliary measure} $\pi_{u_-}$,\label{page:normalized_auxiliary}
\begin{equation}\label{eq:piumn}
\pi_{u_-}^N:=\frac{\gamma_{u_-}^N}{\cal{Z}^N_{u}}\approx \frac{\gamma_{u_-}}{\cal{Z}_{u}}=:\pi_{u_-}\enskip\text{where}\enskip\cal{Z}^N_{u}:=\gamma_{u_-}^N(\bm{E}_{\cal{C}_u}),\enskip \cal{Z}_{u}:=\gamma_{u_-}(\bm{E}_{\cal{C}_u}),
\end{equation}
and it draws  $N$ particles $\bm{X}_{u_-}^{1,N},\dots,\bm{X}_{u_-}^{N,N}$ from $\pi_{u_-}^N$ using (multinomial) resampling. The empirical distribution $N^{-1}\sum_{n=1}^N\delta_{\bm{X}_{u_-}^{n,N}}$ of these particles approximates $\pi_{u_-}$. 
Hence, by extending the  path of each resampled particle from $\bm{E}_{\cal{C}_u}$ to $\bm{E}_u$ using the kernel $K_u$,\label{page:resammutapar}
$$\bm{X}_u^{n,N}:=(X^{n,N}_u,\bm{X}_{u_-}^{n,N})\enskip\text{where}\enskip X^{n,N}_u\sim K_u(\bm{X}_{u_-}^{n,N},dx_u),$$
and defining
\begin{align}\label{eq:actualtargetdistributions}
\gamma_u^N :=& \frac{\cal{Z}_u^N}{N} \sum_{n=1}^N \delta_{\bm{X}_u^{n,N}} \quad \text{and} \quad \pi_u^N:=\frac{1}{N} \sum_{n=1}^N \delta_{\bm{X}_u^{n,N}},
\end{align}
Algorithm~\ref{alg:adacsmc} returns a finite-sample  approximation of the extended auxiliary measure $\gamma_u:=\gamma_{u_-}\times K_u$ indexed by $u$ itself.

\begin{algorithm}[t]
\begin{algorithmic}[1]
\STATE{\textbf{Input:} number of particles $N$, proposal kernels $(K_u)_{u\in\mathbb{T}}$, and auxiliary measures $(\gamma_{u_-})_{u\in\mathbb{T}^{\not\partial}}$.} 
\IF{$u$ is a leaf (i.e.\ $u\in\mathbb{T}^\partial$)}
\STATE{\textit{Propose:} for $n\leq N$, draw $\bm{X}_u^{n,N}$ independently from $K_u$. }
\STATE{\textit{Return:} $\gamma_u^N=N^{-1}\sum_{n=1}^N\delta_{\bm{X}_u^{n,N}}$.}
\ELSE
\FOR{$v$ in $\cal{C}_u$}
\STATE{\textit{Recurse:} set $\gamma_v^N=\text{dac\_smc}(v)$.}
\ENDFOR
\STATE{\textit{Product:} compute $\gamma_{\cal{C}_u}^N=\prod_{v\in\cal{C}_u}\gamma_v^N$.}
\STATE{\textit{Correct:} re-weight $\gamma_{\cal{C}_u}^N$ with $w_{u_-}$ to obtain $\gamma_{u_-}^N$ and its normalization $\pi^N_{u_-}=\gamma_{u_-}^N/\gamma_{u_-}^N(\bm{E}_{\cal{C}_u})$.}
\STATE{\textit{Resample:} for $n\leq N$, draw $\bm{X}_{u_-}^{n,N}$ independently from $\pi_{u_-}^N$.}
\STATE{\textit{Mutate:} for $n\leq N$, draw $X^{n,N}_u$ independently from $K_u(\bm{X}_{u_-}^{n,N},\cdot)$ and set $\bm{X}_u^{n,N}=(X^{n,N}_u,\bm{X}_{u_-}^{n,N})$.}
\STATE{\textit{Return:} $\gamma_u^N=N^{-1}\cal{Z}_{u}^N\sum_{n=1}^N \delta_{\bm{X}_{u}^{n,N}}$ where $\cal{Z}_{u}^N=\gamma_{u_-}^N(\bm{E}_{\cal{C}_u})$.} 
\ENDIF
\newline {\textbf{Note:} One can use~\eqref{eq:adacests} to compute estimates of $\rho_u$, $Z_u$, and $\mu_u$.}
\end{algorithmic}
\caption{dac\_smc$(u)$ for $u$ in $\mathbb{T}$.}\label{alg:adacsmc}
\end{algorithm}

Also similarly to the standard ASMC case (Section~\ref{sec:smc}; cf. Section~\ref{sec:dacvssmc} for a detailed comparison), we obtain approximations to the targets $\rho_u$ and $\mu_u$ by running dac\_smc($u$) in Algorithm~\ref{alg:adacsmc} and applying a simple importance sampling correction:
\begin{align}\label{eq:adacests}
\rho_u^N(d\bm{x}_u) = w_u(\bm{x}_u) \gamma_u^N(d\bm{x}_u),\quad  Z_u^N = \rho_u^N(\bm{E}_u),\quad  \mu_u^N(d\bm{x}_u) = \frac{\rho_u^N(d\bm{x}_u)}{Z_u^N},
\end{align}
where $w_u=d\rho_u/d\gamma_u$ \label{page:inferential_weights}denotes the appropriate weight function. As we will see in Section~\ref{sec:theory}, these approximations have appealing theoretical properties.  The computation time of Algorithm~\ref{alg:adacsmc} can be lowered by  parallelizing the for loop in lines 5--7.
Of course, for the approach to work, we must choose the auxiliary measures $(\gamma_{u_-})_{u\in\mathbb{T}^{\not \partial}}$ and proposal kernels $(K_u)_{u\in\mathbb{T}}$ such that the relevant Radon-Nikodym derivatives exist. To avoid the technical machinery necessary to allow for the possibility of all particles simultaneously being assigned zero weight (cf.\ \cite[Section 7.4]{DelMoral2004}), we further assume that these derivatives are positive everywhere:
\begin{assumption}\label{ass:abscont}For all $u$ in $\mathbb{T}$ and $v$ in $\mathbb{T}^{\not \partial}$,  $\rho_u$ is absolutely continuous w.r.t. $\gamma_u$, $\gamma_{v_-}$ is absolutely continuous w.r.t. $\gamma_{\cal{C}_v}$,  and the Radon-Nikodym derivatives  $w_u:=d\rho_u/d\gamma_u$ and $w_{v_-}:=d\gamma_{v_-}/d\gamma_{\cal{C}_v}$ are  positive everywhere. 
\end{assumption}
To finish this section, we illustrate how one may choose the proposal kernels and obtain the weight functions by revisiting the example of the previous section:
\begin{example}[Simplified  mathematics test dataset, kernels and weights]To apply DaC-SMC to Example~\ref{ex:schools},  we need to pick proposal kernels. In the case of the leaves, the choice is clear: for all $s,y,n$, we generate $\theta_{sy}^{n,N}$ by drawing a sample from a Beta distribution with parameters $m_{sy}+1$ and $M_{sy}-m_{sy}+1$ and mapping the sample through the logit function; so that 
$$K_{sy}(d\theta_{sy})=(M_{sy}+1)\binom{M_{sy}}{m_{sy}}\cal{B}(m_{sy};M_{sy},\alpha(\theta_{sy}))g(d\theta_{sy})$$ 
and the weight functions are constant (for all $s$ in $\cal{S}$, $w_{s_-}=\prod_{y\in\cal{Y}_s}M_{sy}^{-1}$ and, for all $y$ in $\cal{Y}_s$, $w_{sy}=M_{sy}^{-1}$). 
In the case of the other nodes, we would ideally set the kernel to be the appropriate conditional distribution (cf.~Section~\ref{sec:optimalprop}),
\begin{align*}H_s((\theta_{sy})_{y\in\cal{Y}_s},d\theta_s,d\sigma^2_s)&\propto f(d\sigma^2_s;1)g(d\theta_s)\prod_{y\in\cal{Y}_s}\cal{N}(\theta_{sy}-\theta_s;0,\sigma_s^2)\quad\forall s\in\cal{S};\\
H_{\mathfrak{r}}((\theta_s,\sigma^2_s)_{s\in\cal{S}},d\theta_{\mathfrak{r}},d\sigma^2_{\mathfrak{r}})&\propto f(d\sigma^2_{\mathfrak{r}};1)g(d\theta_{\mathfrak{r}})\prod_{s\in\cal{S}}\cal{N}(\theta_{s}-\theta_{\mathfrak{r}};0,\sigma^2_{\mathfrak{r}}).\end{align*}
These, however, are intractable. So, in an effort to approximate them, we note that
$$\prod_{s\in\cal{S}}\cal{N}(\theta_{s}-\theta_{\mathfrak{r}};0,\sigma^2_{\mathfrak{r}})=\frac{e^{\frac{1}{2\sigma^2_\mathfrak{r}}\left[\sum_{s\in\cal{S}}\theta_s\left(\theta_s-\frac{1}{\mmag{\cal{S}}}\sum_{s'\in\cal{S}}\theta_{s'}\right)\right]}}{\sqrt{\mmag{\cal{S}}}\sqrt{2\pi\sigma^2_{\mathfrak{r}}}^{\mmag{S}-1}}\cal{N}\left(\theta_{\mathfrak{r}};\frac{1}{\mmag{\cal{S}}}\sum_{s\in\cal{S}}\theta_{s},\frac{\sigma^2_{\mathfrak{r}}}{\mmag{\cal{S}}}\right)$$
and similarly for $\prod_{y\in\cal{Y}_s}\cal{N}(\theta_{sy}-\theta_s;0,\sigma_s^2)$ and all $s$ in $\cal{S}$; and we instead set
\begin{align*}
K_s((\theta_{sy})_{y\in\cal{Y}_s},d\theta_s,d\sigma^2_s)&:=f(d\sigma^2_s;1)\cal{N}\left(d\theta_s;\frac{1}{\mmag{\cal{Y}_s}}\sum_{y\in\cal{Y}_s}\theta_{sy},\frac{\sigma^2_s}{\mmag{\cal{Y}_s}}\right),\\
K_{\mathfrak{r}}((\theta_s,\sigma^2_s)_{s\in\cal{S}},d\theta_{\mathfrak{r}},d\sigma^2_{\mathfrak{r}})&:=f(d\sigma^2_{\mathfrak{r}};1)\cal{N}\left(d\theta_{\mathfrak{r}};\frac{1}{\mmag{\cal{S}}}\sum_{s\in\cal{S}}\theta_{s},\frac{\sigma^2_{\mathfrak{r}}}{\mmag{\cal{S}}}\right);
\end{align*}
in which case $ w_{\mathfrak{r}_-}=\prod_{s\in\cal{S}}w_s$,
\begin{align*}
&w_s(\theta_s,\sigma^2_s,(\theta_{sy})_{y\in\cal{Y}_s})=\frac{g(\theta_s)e^{\frac{1}{2\sigma^2_s}\left[\sum_{y\in\cal{Y}_s}\theta_{sy}\left(\theta_{sy}-\frac{1}{\mmag{\cal{Y}_s}}\sum_{y'\in\cal{Y}_s}\theta_{sy'}\right)\right]}}{\sqrt{\mmag{\cal{Y}_s}}\sqrt{2\pi\sigma^2_s}^{\mmag{\cal{Y}_s}-1}}\quad\forall s\in\cal{S},\\
&w_{\mathfrak{r}}(\theta_{\mathfrak{r}},\sigma^2_{\mathfrak{r}},(\theta_s,\sigma^2_s,(\theta_{sy})_{y\in\cal{Y}_s})_{s\in\cal{S}})=\frac{g(\theta_{\mathfrak{r}})e^{\frac{1}{2\sigma^2_{\mathfrak{r}}}\left[\sum_{s\in\cal{S}}\theta_s\left(\theta_s-\frac{1}{\mmag{\cal{S}}}\sum_{s'\in\cal{S}}\theta_{s'}\right)\right]}}{\sqrt{\mmag{\cal{S}}}\sqrt{2\pi\sigma^2_{\mathfrak{r}}}^{\mmag{\cal{S}}-1}}.
\end{align*}
To sample $(\theta_s,\sigma^2_s)$ from $K_s((\theta_{sy})_{y\in\cal{Y}_s},d\theta_s,d\sigma^2_s)$, first sample $\sigma^2_s$ from $f(d\sigma^2_s;1)$ and then sample $\theta_s$ from  $\cal{N}(d\theta_s;\mmag{\cal{Y}_s}^{-1}\sum_{y\in\cal{Y}_s}\theta_{sy},\mmag{\cal{Y}_s}^{-1}\sigma^2_s)$. Similarly for $K_{\mathfrak{r}}$.
\end{example}
\subsection{Efficient computation}\label{sec:lcdacsmc}
The main drawback of Algorithm~\ref{alg:adacsmc} is the $\cal{O}(N^{c_u})$ computational cost of resampling from the product-form $\pi_{u_-}^N$ in \eqref{eq:piumn} for a general auxiliary measure $\gamma_{u_-}$ (to draw $N$ samples from $\pi_{u_-}^N$ in $\cal{O}(N^{c_u})$ operations one can, for example, use the alias method~\cite{Walker1977}).  This results in a total algorithmic cost of   $\cal{O}(N^d)$, where $d$ denotes the tree's degree (i.e.\ the largest number of children that any of the tree's nodes possess). 
In some cases, $d$ is small enough that the cost is manageable. For instance,  when the collection of target distributions $(\rho_u)_{u\in\mathbb{T}_u}$ is an artificial construct introduced for computational purposes and the interest lies in the final target, it is common~\cite{ding2018, ding2019,Corneflos2021} to choose  the indexing tree to be binary, in which case  $d=2$   and the algorithm's cost is  $\cal{O}(N^2)$.

In cases where each weight function is bounded above by a known constant, one can lower the cost to $\cal{O}(N)$ using rejection sampling as proposed in~\cite[Section~4.2]{Corneflos2021}. 
Otherwise, one can avoid  the $\cal{O}(N^{c_u})$ cost by  choosing auxiliary measures that factorize. 
For instance, suppose that, as in Example~\ref{ex:schools},  they factorize fully:
$$\gamma_{u_-}:=\prod_{v\in \cal{C}_u}\bar{\gamma}_v\quad\forall u\in\mathbb{T}^{\not \partial},$$
where $(\bar{\gamma}_v)_{v\in\mathbb{T}}$ denotes a collection of measures on the respective spaces $(\bm{E}_v,\bm{\cal{E}}_v)_{v\in\mathbb{T}}$. In this case, the weight function $w_{u_-}$ decomposes into the product
\begin{align*}
w_{u_-}=\frac{d \gamma_{u_-}}{d\gamma_{\cal{C}_u}}=\prod_{v\in \cal{C}_u}\frac{d \bar{\gamma}_v}{d\gamma_v}=:\prod_{v\in \cal{C}_u}\bar{w}_v,
\end{align*}
and the correction step (line 8 in Algorithm~\ref{alg:adacsmc}) breaks down into the following: compute 
$$\gamma_{u_-}^N:=\prod_{v\in \cal{C}_u}\bar{\gamma}_v^N,\quad \pi_{u_-}^N:=\prod_{v\in \cal{C}_u}\bar{\pi}_v^N,$$
where, for all children $v$ of $u$,
\begin{equation}\label{eq:gammauNlc}\bar{\gamma}_v^N:=\bar{w}_v\gamma_v^N=\frac{\cal{Z}_v^N}{N}\sum_{n=1}^N\bar{w}_v(\bm{X}_v^{n,N})\delta_{\bm{X}_v^{n,N}}.\end{equation}
Resampling from $\pi_{u_-}^N:=\bar{\gamma}_v^N/\bar{\gamma}_v^N(\bm{E}_v)$ can then be achieved in $\cal{O}(N)$ operations by independently drawing $N$ samples  from $\bar{\pi}_{u1}^N,\dots,\bar{\pi}_{uc_u}^N$, e.g.\ by applying the alias method~\cite{Walker1977} to each approximation separately, and concatenating them.   In other words, we can replace lines 5--9 in Algorithm~\ref{alg:adacsmc} with those in Algorithm~\ref{alg:adacsmclc} (and obtain the algorithm referred to as `DaC-SMC' in~\cite{lindsten2017divide}, cf.\ Footnote~\ref{footnote}). The running time can then be  lowered further by parallelizing these operations across $u$'s children.

\begin{algorithm}[t]
\begin{algorithmic}[1]
\FOR{$v$ in $\cal{C}_v$}
\STATE{\textit{Recurse:} set $\gamma_v^N:=\text{dac\_smc}(v)$.}
\STATE{\textit{Correct:} compute $\bar{\gamma}_v^N$   in~\eqref{eq:gammauNlc} and $\bar{\pi}^N_{v}:=\bar{\gamma}_v^N/\bar{\gamma}_v^N(\bm{E}_v)$.}
\STATE{\textit{Resample:} for $n\leq N$, $\bm{\bar{X}}_{v}^{n,N}$ independently from $\bar{\pi}_{v}^N$.}
\ENDFOR
\STATE{\textit{Concatenate:} for $n\leq N$, set $\bm{X}_{u_-}^{n,N}:=(\bm{\bar{X}}_{v}^{n,N})_{v\in\cal{C}_u}$.}
\end{algorithmic}
\caption{$\cal{O}(N)$-cost replacement of lines 5--9 in Algorithm~\ref{alg:adacsmc} for fully-factorized $\gamma_{u_-}$.}\label{alg:adacsmclc}
\end{algorithm}

These choices result in DaC-SMC employing extended auxiliary measures with partial product structure:
$$\gamma_u(d\bm{x}_u)=K_u(\bm{x}_{\cal{C}_u},dx_u)\prod_{v\in\cal{C}_u}\bar{\gamma}_v(d\bm{x}_v)\quad\forall u\in\mathbb{T}^{\not \partial}.$$
(I.e.\ if $(Y_u,\bm{Y}_{u1},\dots,\bm{Y}_{uc_u})\sim\pi_u$, then $\bm{Y}_{u1},\dots,\bm{Y}_{uc_u}$ are independent.) If the targets $\rho_u$ and $\mu_u$ do not share the above structure,  then this $\cal{O}(N)$ approach will generally require large $\gamma_u^N\mapsto\rho_u^N$ corrections and, consequently, result in greater estimator variance (which may, or may not, be compensated by the computational gains).   
As suggested in~\cite[Section~4.2]{lindsten2017divide}, one way to mitigate this issue is via tempering~\cite{delmoral2006sequential}. 

Another way to incorporate non-product structure into the auxiliary measures without sacrificing the $\cal{O}(N)$ cost is setting them to be sums of fully-factorized measures:
$$\gamma_{u_-}:=\sum_{i=1}^{I}\prod_{v\in \cal{C}_u}\bar{\gamma}_v^i=:\sum_{i=1}^{I}\gamma_{u_-}^i\quad\forall u\in\mathbb{T}^{\not \partial},$$
where  $(\bar{\gamma}_v^1,\dots,\bar{\gamma}_v^{I})_{v\in\mathbb{T}}$ denote collections of measures on the respective spaces $(\bm{E}_v,\bm{\cal{E}}_v)_{v\in\mathbb{T}}$. In this case, the correction step becomes: return
$$\gamma_{u_-}^N:=\sum_{i=1}^{I}\gamma_{u_-}^{i,N},\quad \pi_{u_-}^N:=\sum_{i=1}^{I}\frac{\gamma_{u_-}^{i,N}(\bm{E}_{\cal{C}_u})}{\gamma_{u_-}^{N}(\bm{E}_{\cal{C}_u})}\pi_{u_-}^{i,N}=:\sum_{i=1}^{I}\omega_{u_-}^{i,N}\pi_{u_-}^{i,N},$$
where  $\gamma_{u_-}^{i,N}:=\prod_{v\in \cal{C}_u}\bar{\gamma}_v^{i,N}$ and $\pi_{u_-}^{i,N}:=\prod_{v\in \cal{C}_u}\bar{\pi}_v^{i,N}$, with 
$\bar{\gamma}_v^{i,N},\bar{\pi}_v^{i,N}$ defined analogously to $\bar{\gamma}_v^{N},\bar{\pi}_v^{N}$ in~\eqref{eq:gammauNlc}. We can then draw $N$ samples from $\pi_{u_-}^N$ in $\cal{O}(N)$ operations by, for instance, drawing $(m_1,\dots,m_I)$ from a multinomial distribution with weights $(\omega_{u_-}^{i,N})_{i=1}^I$ and, for each $i$, generating $m_i$ samples  from $\pi_{u_-}^{i,N}$ as described above for the $I=1$ case.

Similar considerations apply if the auxiliary measures are chosen to be sums of partially factorized functions rather than fully factorized ones. The main difference is that resampling from $\pi_{u_-}^{N}$ has a `conditional aspect' to it and costs  $\cal{O}(N^a)$ where the exponent $1<a<c_u$ depends on the amount of factorization, see Appendix~\ref{app:timevar} for an example.

\subsection{Variants}\label{sec:variants}
To simplify the algorithm's presentation and analysis, we focus throughout on the particular version given in Algorithm~\ref{alg:adacsmc}.
There are, however, many variants that one might wish to consider in different settings and whose analysis may be tackled using straightforward modifications of the arguments given in this paper. 

First off, we have  the low-cost variant referred to as `lightweight mixture resampling' in \cite[Section~4.1]{lindsten2017divide} (see `multiple matching' in \cite[Section 2.2]{LinZCC:2005} for similar ideas) where the product-form estimator $\gamma_{u_-}^N(d\bm{x}_{\cal{C}_u})$  in~\eqref{eq:gammauN} is replaced  with its `incomplete' version,
\begin{equation}\label{eq:incomplete}\frac{\cal{Z}^N_{\cal{C}_{u}}}{\mmag{\cal{M}}}\sum_{\bm{n}\in\cal{M}}w_{u_-}(\bm{X}_{\cal{C}_u}^{\bm{n},N})\delta_{\bm{X}_{\cal{C}_u}^{\bm{n},N}},\end{equation}
with $\cal{M}\subseteq[N]^{c_u}$ indexing a user-chosen subset of permuted particles $\bm{X}_{\cal{C}_u}^{\bm{n},N}$. In particular, by picking $\cal{M}$ carefully, it might be possible to substantially reduce the cost without sacrificing too much of the variance reduction (cf.\ \cite{kong2021design} and references therein for similar feats in the U-statistics literature). While one may opt for more sophisticated approaches (e.g.\ ones along the lines of those in~\cite{kong2021design}), we can offer two simple rules of thumb for choosing $\cal{M}$ that are motivated by the considerations in~\cite[Section~4]{Kuntz2021}: (a) set $\cal{M}$'s size to be such that the cost of resampling from~\eqref{eq:incomplete}'s normalization is comparable to that of generating the unpermuted particles, $(\bm{X}_{u1}^{\bm{n},N},\dots,\bm{X}_{uc_u}^{\bm{n},N})_{n=1}^N$; and (b) minimize the number of over-lapping components in the elements of $\cal{M}$. 
In terms of theoretical analysis for this variant, one can combine the methods in Appendices~\ref{app:lplln}--\ref{app:cross} with the techniques used to study incomplete U-statistics~\cite[Chapter~4.3]{Lee1990} and their generalizations.

We also need not resample at every step and may instead choose to do so adaptively (e.g.\ resample only when the effective sample size drops underneath a predetermined threshold~\cite{Kong1994}). To prevent the number of particles from blowing up, any product-form approximation $\pi_{u_-}^N$  to a normalized auxiliary distribution $\pi_{u_-}$ that is not resampled must be replaced with an $N$ sample one, e.g.\
\begin{equation}\label{eq:standard}\frac{\sum_{n=1}^Nw_{u_-}(\bm{X}_{u1}^{n,N},\dots,\bm{X}_{uc_u}^{n,N})\delta_{(\bm{X}_{u1}^{n,N},\dots,\bm{X}_{uc_u}^{n,N})}}{\sum_{n=1}^Nw_{u_-}(\bm{X}_{u1}^{n,N},\dots,\bm{X}_{uc_u}^{n,N})}\end{equation}
as in standard ASMC.  We should point out here that, for DaC-SMC, and in contrast with SMC/ASMC,  there are benefits to resampling beyond just mitigating weight degeneracy. In particular, by permuting and resampling the particles, we explore areas of the target space that would be  missed were we to only use the $N$ unpermuted particles~\cite[Section~3.1]{Kuntz2021}.

There is also flexibility in the resampling scheme employed: the theoretical analysis in Section~\ref{sec:theory} applies to Algorithm~\ref{alg:adacsmc} in which multinomial resampling is explicitly encoded. But one could extend our analysis along the lines of~\cite{Gerber2019} to cover a much broader class of resampling schemes, including substantially lower-variance ones.  An interesting question of practical importance is how best to efficiently implement low-variance resampling schemes when the auxiliary measures decompose into sums of partially-factorized measures (e.g.\ consider the example in Appendix~\ref{app:timevar}).

One may also wish to implement minor variations to Algorithm~\ref{alg:adacsmc} dictated by the particular target of interest.
For instance, in some cases, it might be much easier to specify the auxiliary measure $\gamma_{u-}$ over the variables taking values in $\bm{E}_u$ than over those taking values in $\bm{E}_{\cal{C}_u}$. 
In this case, mutation would need to be carried out before correction and resampling (line 10 in Algorithm~\ref{alg:adacsmc} before lines 8--9) leading to a slightly different algorithm with essentially the same properties. This approach allows one to specify an analogue of the extended auxiliary distribution independently of the proposal kernel and directly exploit the information it encodes in the resampling step. The cost of doing so is having to perform mutation $N^{c_u}$ times rather than $N$ times, although this does not increase the overall complexity of the basic algorithm (at least for unfactorized auxiliary measures).

Along similar lines, for targets with particular dependence structures, one may use generalizations of product-form estimators (e.g.\ \cite[Section~3.2~and~Appendix~F]{Kuntz2021}) that directly account for these structures; see Appendix~\ref{app:variantex} for an example. The idea is that, by doing so, one can ease the burden placed on the correction steps which, in Algorithm~\ref{alg:adacsmc}, are the sole responsible 
for introducing dependencies among children into the approximations.

Lastly, one need not generate the same number of particles at each node and may instead opt to allocate more computational power to the nodes that prove most problematic (e.g.\ the $u$s whose associated space $E_u$ have greatest dimension), or even take a more sophisticated approach of the sort in~\cite{Lee2018}. One would just have to replace   $\gamma_u^N$ in~\eqref{eq:gammauN} with
$$\left(\prod_{v\in\cal{C}}\frac{\cal{Z}^{N_v}_{v}}{N_v}\right)\sum_{n_1=1}^{N_{u1}}\dots\sum_{n_1=1}^{N_{uc_u}} w_{u_-}(\bm{X}_{u1}^{n_1,N_{u1}},\dots,\bm{X}_{uc_u}^{n_{c_u},N_{uc_u}})\delta_{(\bm{X}_{u1}^{n_1,N_{u1}},\dots,\bm{X}_{uc_u}^{n_{c_u},N_{uc_u}})}(d\bm{x}_{\cal{C}_u})$$
where, for each child $v$ of $u$, $N_v$ denotes the amount of particles resampled at node $v$. 
\section{Theoretical characterization}\label{sec:theory}
We now turn to the main results of the paper showing that Algorithm~\ref{alg:adacsmc} is well-founded. We give an overview of the proofs for these results in Section~\ref{sec:mop}, postponing the full details until Appendices~\ref{app:lplln}--\ref{app:cross}.  To simplify the exposition, we only state results for the particle approximations  to the targets; however, analogous statements for the approximations to the auxiliary measures and their extensions can be extracted from Appendices~\ref{app:lplln}--\ref{app:clt}. Also to keep the exposition simple,  we focus throughout on the case of bounded test functions, writing $\cal{B}_b(S)$ for the space of bounded measurable real-valued test functions on a measurable space $(S,\cal{S})$ and $\norm{\varphi}:=\sup_{x\in S}\mmag{\varphi(x)}<\infty$ for the supremum norm  on $\cal{B}_b(S)$, and we assume that the weight functions are bounded:
\begin{assumption}\label{ass:boundweights}For all $u$ in $\mathbb{T}^{\not \partial}$ and $v$ in $\mathbb{T}$, $w_{u_-}=d\gamma_{u_-}/d\gamma_{\cal{C}_u}$  and $w_v=d\rho_v/d\gamma_v$ are bounded: $||w_{u_-}||<\infty$ and $\norm{w_v}<\infty$.
\end{assumption}
To begin with, Algorithm~\ref{alg:adacsmc} produces strongly consistent estimators for the targets: 
\begin{theorem}[Strong laws of large numbers]\label{THRM:LLN}If Assumptions~\ref{ass:abscont}--\ref{ass:boundweights} are satisfied, $u$ belongs to $\mathbb{T}$, and $\varphi$ belongs to $\cal{B}_b(\bm{E}_u)$, then
\begin{align*}\lim_{N\to\infty}\rho^N_u(\varphi)=\rho_u(\varphi),\quad\lim_{N\to\infty}\mu^N_u(\varphi)=\mu(\varphi),\quad \lim_{N\to\infty}Z^N_u= Z_u,\quad\text{almost surely.}\end{align*}
\end{theorem}
\begin{proof}
See Appendix~\ref{app:lplln}.
\end{proof}

If the underlying spaces possess nice enough topological properties, the above almost sure pointwise convergence can be strengthened to almost sure weak convergence:
\begin{theorem}[Almost sure weak convergence]\label{THRM:WEAK}If, in addition to Assumptions~\ref{ass:abscont}--\ref{ass:boundweights}, the spaces $(E_u)_{u\in\mathbb{T}}$ are Polish and $(\cal{E}_u)_{u\in\mathbb{T}}$ are the corresponding Borel sigma algebras, then
  \begin{align*}
    \rho_u^N \rightharpoonup \rho_u , \quad
   \mu_u^N \rightharpoonup \mu_u, \quad\text{almost surely},
  \end{align*}
for each $u$ in $\mathbb{T}$, where $\rightharpoonup$ denotes weak convergence as $N\to\infty$.\end{theorem}
\begin{proof}See Appendix~\ref{app:lplln}.\end{proof}
The estimators for the unnormalized targets are unbiased:
\begin{theorem}[Unbiasedness of $(\rho_u^N)_{u\in\mathbb{T}}$]\label{THRM:UNBIAS}
If Assumptions~\ref{ass:abscont}--\ref{ass:boundweights} hold, then
\begin{align*}
\Ebb{\rho^N_u(\varphi)}=\rho_u(\varphi),\quad\Ebb{Z^N_u}=Z_u,\quad\forall N>0,\enskip\varphi\in\cal{B}_b(\bm{E}_u),\enskip u\in\mathbb{T}.
\end{align*}
\end{theorem}
\begin{proof}
See Appendix~\ref{app:unbias}.
\end{proof}
The nonlinearity in the correction step introduces a bias in the estimators for the normalized targets. However, the bias decays linearly with the number of particles $N$: 
\begin{theorem}[Bias estimates for $(\mu_u^N)_{u\in\mathbb{T}}$]\label{THRM:BIAS} If Assumptions~\ref{ass:abscont}--\ref{ass:boundweights} hold, the weight functions are bounded below (i.e., for all $u$ in $\mathbb{T}^{\not \partial}$ and $v$ in $\mathbb{T}$, $w_{u_-}\geq \beta_{u_-}$ and $w_v\geq\beta_v$ for some constants $\beta_{u_-},\beta_v>0$), and $u$ belongs to $\mathbb{T}$, then there exists a constant $C_u<\infty$ such that
\begin{equation*}
\label{eq:bias}
\mmag{\Ebb{\mu^N_u(\varphi)}-\mu_u(\varphi)}\leq\frac{C_u\norm{\varphi}}{N}\quad\forall N>0,\enskip \varphi\in\cal{B}_b(\bm{E}_u).
\end{equation*}
\end{theorem}
\begin{proof}
See Appendix~\ref{app:bias}.
\end{proof}
All estimators converge at a rate proportional to the square root of the number of particles:
\begin{theorem}[$L^p$ inequalities]\label{THRM:LP}If Assumptions~\ref{ass:abscont}--\ref{ass:boundweights} hold, then, for each $p\geq 1$ and $u$ in $\mathbb{T}$, then there exist constants $C_u^\rho,C_u^\mu<\infty$ such that
\begin{align*}
\Ebb{\mmag{\rho^N_u(\varphi)-\rho(\varphi)}^p}^{\frac{1}{p}}\leq\frac{C_u^\rho\norm{\varphi}}{N^{1/2}},\quad \Ebb{\mmag{\mu^N_u(\varphi)-\mu_u(\varphi)}^p}^{\frac{1}{p}}\leq\frac{C_u^\mu\norm{\varphi}}{N^{1/2}},
\end{align*}
for all $N>0$ and $\varphi$ in $\cal{B}_b(\bm{E}_u)$. In particular, $\Ebb{\mmag{Z^N_u-Z_u}^p}^{1/p}\leq C_u^\rho/N^{1/2}$ for all $N>0$.
\end{theorem}
\begin{proof}
See Appendix~\ref{app:lplln}.
\end{proof}
To further characterize the $\sqrt{N}$ rate of convergence of $\rho_u^N$ and $\mu^N_u$, we obtain a central limit theorem for each (given in Theorem~\ref{THRM:CLT} below). To state these, we first need to introduce some notation: given two measurable spaces $(A,\cal{A})$ and $(B,\cal{B})$, a kernel $M:A\times\cal{B}\to[0,\infty)$ from the first to the second, and a test function $\varphi:B\rightarrow \r$, we use $M\varphi$ to denote the real-valued function on $A$ defined by
\begin{align*}
(M\varphi)(a):=\int M(a, db)\varphi(b)\quad \forall a \in A,
\end{align*}
assuming that the above integrals are well-defined. 

As we will see in Theorem~\ref{THRM:CLT}, the asymptotic variances of our estimators decompose into sums of $\mmag{\mathbb{T}_u}$ terms, each one accounting for the variance introduced by the computations carried out in a different descendant $v$ of $u$ (out of convenience, we include $u$ itself among $u$'s descendants). The term corresponding to a particular descendant $v$ involves a kernel   $\Gamma_{v,u}:\bm{E}_v\times\bm{\cal{E}}_{u}\to[0,\infty)$ mapping $\gamma_v$ to $\gamma_u$:
\begin{align}\label{eq:Pivu} \gamma_v(\Gamma_{v,u}\varphi)=\gamma_u(\varphi)\quad \forall \varphi\in\cal{B}_b(\bm{E}_{u}).\end{align}
The kernel $\Gamma_{v,u}$ encapsulates the relationship between a node, $u$, in the tree and one of its descendants, $v$, upon marginalizing out the other descendants of $u$. It is defined recursively: 
\begin{align}\Gamma_{u,u}(\bm{x}_u,d\bm{y}_u)&:=\delta_{\bm{x}_{u}}(d\bm{y}_u),\nonumber\\
\label{eq:Pivu1}\Gamma_{v,u}(\bm{x}_{v},d\bm{y}_{u})&:=\delta_{\bm{x}_{v}}(d\bm{y}_{v})\gamma_{\cal{C}_u}^{\not v}(d\bm{y}_{\cal{C}_u^{\not v}})w_{u_-}(\bm{y}_{\cal{C}_u})K_u(\bm{y}_{\cal{C}_u},dy_u)\quad\forall v\in \cal{C}_u,
\end{align}
where $\gamma_{\cal{C}_u}^{\not v}:=\prod_{r\in\cal{C}_u^{\not v}}\gamma_{r}$ with $\cal{C}_u^{\not v}:=\cal{C}_u\setminus\{v\}$.
For all other descendants $v\neq u$, we set
\begin{equation}\label{eq:Pivu2}\Gamma_{v,u}=\Gamma_{v,r_1}\Gamma_{r_1,r_2}\dots \Gamma_{r_l,u},\end{equation}
where $v,r_1,\dots,r_l,u$  denotes the branch of $\mathbb{T}_u$ connecting $v$ and $u$.   Because $\Gamma_{v,u}$ trivially satisfies~\eqref{eq:Pivu} if $v$ is a child of $u$, it follows from~\eqref{eq:Pivu2} that~\eqref{eq:Pivu} also holds for all other descendants $v$ of $u$. 
Note that the collection $(\Gamma_{u,v})_{u\in\mathbb{T},v\in\mathbb{T}_{u}}$ of all these kernels amounts to a generalization of the semigroup describing the propagation of local sampling errors in standard SMC that is central to much of its analysis  (e.g.\ see~\cite{DelMoral2004}).
\begin{theorem}[Central limit theorems]\label{THRM:CLT}If Assumptions~\ref{ass:abscont}--\ref{ass:boundweights} hold, then, as $N\to\infty$,
\begin{align*}N^{1/2}\left(\rho^N_u(\varphi)-\rho_u(\varphi)\right)&\Rightarrow\cal{N}(0,\sigma^2_{\rho_u}(\varphi)),\quad N^{1/2}\left(\mu^N_u(\varphi)-\mu_u(\varphi)\right)\Rightarrow\cal{N}(0,\sigma^2_{\mu_u}(\varphi)),\end{align*}
for any given $u$ in $\mathbb{T}$ and $\varphi$ in $\cal{B}_b(\bm{E}_u)$, where $\Rightarrow$ denotes convergence in distribution,
\begin{align*}%
\sigma^2_{\rho_u}(\varphi)&:=\sum_{v\in\mathbb{T}_u}\pi_v([\cal{Z}_{v}\Gamma_{v,u}[w_u\varphi]-\rho_u(\varphi)]^2),\\
\sigma^2_{\mu_u}(\varphi)&:=\sum_{v\in\mathbb{T}_u}\pi_v([\cal{Z}_{v}\Gamma_{v,u}[w_uZ_u^{-1}[\varphi-\mu_u(\varphi)]]]^2).
\end{align*}
In particular, $N^{1/2}\left( Z_u^N-Z_u\right)\Rightarrow\cal{N}(0,\sigma^2_{Z_u})$ as $N\to\infty$ with
\begin{equation}\label{eq:sigZ}\sigma^2_{Z_u} := Z_u^2\sum_{v\in\mathbb{T}_u}\pi_v\left(\left[\frac{d\mu_u^v}{d\pi_v}-1\right]^2\right),\end{equation}
where $\mu_u^v$ denotes the $\bm{E}_v$-marginal of $\mu_u$ (i.e.\ $\mu_u^v(A):=\mu_u(A\times E_{\mathbb{T}_u
\backslash \mathbb{T}_v})$ for all $A$ in $\bm{\cal{E}}_v$).
\end{theorem}
\begin{proof} See Appendix~\ref{app:clt}.
\end{proof}
\subsection{Methods of proof}\label{sec:mop}
The approximations to the targets are obtained from those to the extended auxiliary measures via a single importance sampling step, cf.\ \eqref{eq:adacests}. Hence, just as most results for ASMC are easily extracted from those for SMC~\cite{johansen2008note}, the fundamental objects we need to analyze here are the approximations to the extended auxiliary measures.

To this end, note that the main difference between DaC-SMC (Algorithm~\ref{alg:adacsmc}) and standard ASMC (Algorithm~\ref{alg:asmc}) is the correction step: in the former case, we employ the product-form estimator~\eqref{eq:gammauN} while, in the latter, we instead use the usual estimator~\eqref{eq:asmc}. Consequently, our proofs for the above results  combine well-known methods previously used to establish analogous results for standard SMC and ASMC with novel techniques that control the errors introduced by the product-form estimators embedded within DaC-SMC.

At a conceptual level, we expand the approximation error into a sum of products of `local errors', each accounting for the approximations made at a different node, and we control them separately (in particular, we show that, just as in the standard case, these errors are $\cal{O}(N^{-1/2})$). However, in contrast with standard SMC and ASMC whose local errors only get `propagated forward in time', those of DaC-SMC get multiplied together as we move up the tree. Consequently, we end up with a far greater number of terms in our expansion than normal.  Herein lies the main novelty of our analysis: dispatching these extra products-of-local-errors terms. Because the products are taken over nodes on separate branches, the errors in the products are independent of each other. Hence, the product of $k$ of these $\cal{O}(N^{-1/2})$ errors should intuitively be $\cal{O}(N^{-k/2})$.

More specifically, we begin by proving the $L^p$ inequalities (Theorem~\ref{THRM:LP}). To do so, we follow the iterative approach taken in \cite{Crisan2002,miguez2013convergence}: we call on  the Marcinkiewicz-Zygmund inequality (\cite[Lemma~7.3.3]{DelMoral2004}) to obtain $L^p$ inequalities for the empirical distributions of the particles indexed by the tree's leaves and we show, step-by-step, that the algorithm preserves these inequalities. The resampling and mutation steps then follow from arguments similar to those in~\cite{Crisan2002,miguez2013convergence}. For the correction step, we need to do some extra work and show that the product $\gamma_{\cal{C}_u}^N=\prod_{v\in\cal{C}_v}\gamma_v^N$ also maintains the inequalities. That is, we need to prove that the $L^p$ norm of the product of the errors is $\cal{O}(N^{-1/2})$; something we do roughly in Lemma~\ref{lem:prodlp} by  using the boundedness of the test and weight functions to uniformly bound all but one of the errors in the product, and exploiting the $\cal{O}(N^{-1/2})$ size of the remaining error. To complete the argument for the correction step, we then resume with the approach of~\cite{Crisan2002,miguez2013convergence} and show that the re-weighting in~\eqref{eq:gammauN} preserves the inequalities. Next, we extract the laws of large numbers (Theorem~\ref{THRM:LLN}) from the $L^p$ inequalities using the usual approach involving Chebyshev's inequality and the Borel-Cantelli lemma (e.g.\ \cite[p.~17]{boustati2020generalised}). The weak convergence (Theorem~\ref{THRM:WEAK}) then follows using standard techniques described in detail for SMC in \cite[Section S1]{Schmon2021}.

To prove the unbiasedness in Theorem~\ref{THRM:UNBIAS}, we generalize the arguments used in \cite[Section~16.4.1]{chopin2020} to establish the analogous result for standard SMC. The proof also works recursively: we show that $\gamma^N_{u_-}$ is an unbiased estimator for $\gamma_{u_-}$ if $\gamma^N_{v_-}$ is a unbiased estimator for $\gamma_{v_-}$ for each non-leaf child $v$ of $u$. The unbiasedness of $\rho_u^N$ then follows with a simple use of the tower property.

To obtain the bias bounds in Theorem~\ref{THRM:BIAS}, we follow the approach of~\cite{olsson2004bootstrap}. This argument also works its way  through the algorithm step by step, this time showing that each step increases the bias by at most $\cal{O}(N^{-1})$. Just as with the $L^p$ inequalities, our innovation  here (Lemma~\ref{lem:prodbias}) deals with the product $\gamma_{\cal{C}_u}^N=\prod_{v\in\cal{C}_v}\gamma_v^N$. As in that case, we obtain  the necessary bias bound for the product  by uniformly bounding the bias of all but one of the approximations in the product and using a previously derived $\cal{O}(N^{-1})$ bound for the remaining approximation.

Our proof for the central limit theorems (Theorem~\ref{THRM:CLT}) follows the conceptual approach outlined above and extends the arguments in~\cite[Chapter~9]{DelMoral2004}. Similarly as in that chapter, we obtain an expression for the approximation error in terms of local errors (and their `propagations up the tree'), the difference being that the expression is not just a  sum of the local errors, but of their products too. An argument similar to~\cite[Corollary~9.3.1]{DelMoral2004} then shows that the local errors are asymptotically normal and independent. Hence, if we are able to demonstrate that their products are $o(N^{-1/2})$, then the CLTs  follow as in~\cite[p.~301]{DelMoral2004} by applying Slutsky's theorem and the continuous mapping theorem, and exploiting the fact that a linear combination of independent normal random variables is normal. To this end, we prove the following product version of the $L^2$ inequality in Theorem~\ref{THRM:LP}:

\begin{theorem}[Product $L^2$ inequality]\label{THRM:CROSS} If, in addition to Assumptions~\ref{ass:abscont}--\ref{ass:boundweights}, $u$ and $v$ lie in separate branches (i.e.\ $u\not\in\mathbb{T}_v$ and $v\not\in\mathbb{T}_u$), then there exists a constant $C_{u,v}<\infty$ such that
$$\Ebb{(\gamma^N_u-\gamma_u)\times(\gamma_v^N-\gamma_v)(\varphi)^2}^{\frac{1}{2}}\leq \frac{C_{u,v}\norm{\varphi}}{N}\quad\forall N>0,\enskip \varphi\in\cal{B}_b(\bm{E}_u\times\bm{E}_v).$$
\end{theorem}
\begin{proof}See Appendix~\ref{app:cross}.\end{proof}
To the best of our knowledge, this theorem and its proof are unprecedented in the SMC literature. To argue it, we use an inductive approach reminiscent of those we take for the $L^p$ inequalities,  unbiasedness, and bias bounds: one by one, we show that each step of the algorithm preserves the above inequality. For the base case dealing with empirical distributions of the particles indexed by the tree's leaves, we call on~\cite[Lemma~1]{Kuntz2021}.
\subsection{Optimal intermediate targets, auxiliary measures, and proposal kernels}\label{sec:optimalprop} 

While the optimal choice of intermediate targets, auxiliary measures, and proposal kernels generally depends on the particular average we are interested in estimating, there is one choice that leads to zero variance estimates of the final target's normalizing constant $Z_{\mathfrak{r}}$. In particular, suppose that the final target $\rho_{\mathfrak{r}}$  and the underlying space $(\bm{E}_{\mathfrak{r}},\bm{\cal{E}}_{\mathfrak{r}})$ are nice enough (e.g.\ $(\bm{E}_{\mathfrak{r}},\bm{\cal{E}}_{\mathfrak{r}})$ is Borel~\cite[Theorem~8.5]{Kallenberg2021}, see also~\cite{Faden1985}) that, for each $u$ in $\mathbb{T}^{\not \partial}$, there exists a regular conditional probability distribution   mapping the $\bm{E}_{\cal{C}_u}$-marginal $\mu_{\mathfrak{r}}^{u_-}$ (i.e.\ $\mu_{\mathfrak{r}}^{u_-}(A):=\mu_{\mathfrak{r}}(A\times  \bm{E}_{\mathbb{T}\backslash(\mathbb{T}_u^{\not u})})$ for all $A$ in $\bm{\cal{E}}_{\cal{C}_u}$) to its $\bm{E}_u$-marginal $\mu_{\mathfrak{r}}^u$. 
That is, a Markov kernel $M_u:\bm{E}_{\cal{C}_u}\times \cal{E}_u\to[0,1]$ satisfying $\mu_{\mathfrak{r}}^{u_-}\times M_u= \mu_{\mathfrak{r}}^u$. In this case,
setting
$$\mu_{u}:=\mu_{\mathfrak{r}}^u,\quad K_u:=M_u\enskip\forall u\in\mathbb{T},\quad \pi_{u_-}:=\mu_{\mathfrak{r}}^{u_-}\enskip\forall u\in\mathbb{T}^{\not\partial},$$
we have that 
$$\pi_v=\pi_{v_-}\times M_v=\mu_{\mathfrak{r}}^{v_-}\times M_v=\mu_{\mathfrak{r}}^{v}=\mu_u^v\quad\forall v\in\mathbb{T}_u,\enskip u\in\mathbb{T}.$$
It follows from~\eqref{eq:sigZ} that our estimator $Z_{\mathfrak{r}}^N$ for $Z_{\mathfrak{r}}$  achieves zero (asymptotic) variance (and similarly for $Z_u^N$  for all other nodes $u$). The above choices generalize those for which ASMC yields zero-variance estimates of the marginal likelihood in the context of filtering, cf.\ \cite[Proposition~2]{Guarniero2017}. Note also that, just as in the ASMC case, the normalizing constants in the intermediate targets auxiliary measures are immaterial and we are free to choose them as we find convenient: something unsurprising given that they do not influence the resampling operation and ultimately cancel in the computation of the normalizing constant estimates. 
However, perhaps a bit unexpectedly, and in contrast with the ASMC case, these choices do not necessarily lead to zero \emph{finite-sample} variance estimates for the normalizing constants. In particular, given the above,
\begin{align*}&w_u=\frac{d\rho_u}{d\gamma_u}=\frac{Z_u}{\cal{Z}_u}\frac{d\mu_u}{d\pi_u}=\frac{Z_u}{\cal{Z}_u},\quad w_{u_-}=\frac{d\gamma_{u_-}}{d\gamma_{\cal{C}_u}}=\frac{\cal{Z}_u}{\cal{Z}_{\cal{C}_u}}\frac{d\pi_{u_-}}{d\pi_{\cal{C}_u}}=\frac{\cal{Z}_u}{\cal{Z}_{\cal{C}_u}}\frac{d\mu^{u_-}_{\mathfrak{r}}}{d\prod_{v\in\cal{C}_u} \mu^{v}_{\mathfrak{r}}},\\
&\Rightarrow Z^N_u=\gamma_u^N(w_u)=\frac{Z_u}{\cal{Z}_u}\cal{Z}_u^N=\frac{Z_u}{\cal{Z}_u}\gamma_{\cal{C}_u}^N(w_{u_-})=\frac{Z_u}{\cal{Z}_{\cal{C}_u}}\gamma_{\cal{C}_u}^N\left(\frac{d\mu^{u_-}_{\mathfrak{r}}}{d\prod_{v\in\cal{C}_u} \mu^{v}_{\mathfrak{r}}}\right),\end{align*}
and the rightmost term will not be constant unless $\mu^{u_-}_{\mathfrak{r}}=\prod_{v\in\cal{C}_u} \mu^{v}_{\mathfrak{r}}$. On the other hand, were this to be the case for all nodes $u$, we could iterate the above down the tree to find that $Z^N_u=Z_u$ for all $u$. 

In short, the above choices lead to zero-finite-sample-variance estimates for the normalizing constants if and only if $\mathbb{T}$ precisely matches the target's dependence structure. Some insight as to why this might be the case can be gleaned by comparing the asymptotic and finite variance expressions given in~\cite[Theorem~1]{Kuntz2021} for product-form estimators. In essence, these choices ensure that the `one-dimensional' marginals of the sampling distribution $\pi_u$ at node $u$ coincide with those of $\mu_u$ and 
 negate the corresponding $\cal{O}(N^{-1/2})$ terms in the finite-sample variance expansions (i.e.\ those also featuring in the asymptotic variance expansions). However, the sampling distribution factorizes over $u$'s children and, unless the target does so too, the higher-dimensional marginals are mismatched and the corresponding  $o(N^{-1/2})$ terms in the expansion persevere.

Just as is the case in standard ASMC~\cite{Guarniero2017}, the above choices can rarely be implemented in practice: with few exceptions, $\mu_{\mathfrak{r}}$'s marginals  cannot be computed explicitly  and the conditional distributions are unknown (much less can be sampled from). A more common situation in practice is that where the intermediate targets are fixed and we seek to choose only the proposal kernels and/or auxiliary measures. In standard ASMC (cf.\ \cite[Chapter~10]{chopin2020} and references therein), this is typically done by deriving the `locally optimal' proposal kernel and auxiliary measure for each step of the algorithm (i.e.\ those that minimize the finite-sample variance of the normalizing constant estimate  obtained at that step) and, if these prove intractable, approximating them as best as one can. To extend this approach to DaC-SMC, we provide the following generalization of the characterization for the locally optimal kernels and measures previously obtained for SMC and ASMC (e.g.\ see  \cite[Theorems~10.1--10.2]{chopin2020}).
\begin{theorem}\label{THRM:LOCALLY}Fix any $u$ in $\mathbb{T}^{\not\partial}$ and $N>0$, and let $\cal{S}$ denote the set of all $(\gamma_{u_-},K_u)$ pairs satisfying    Assumption~\ref{ass:abscont} and $\cal{S}(K_u)$ and $\cal{S}(\gamma_{u_-})$ be its slices:
\begin{align*}&\cal{S}:=\{(\gamma_{u_-},K_u):\gamma_{u_-}\sim \gamma_{\cal{C}_u},\enskip \rho_u\sim \gamma_{u_-}\times K_u\},\quad\cal{S}(K_u):=\{\gamma_{u_-}:(\gamma_{u_-},K_u)\in\cal{S}\}\enskip\forall K_u,\\
&\cal{S}(\gamma_{u_-}):=\{K_u:(\gamma_{u_-},K_u)\in\cal{S}\}\enskip\forall \gamma_{u_-}.
\end{align*}
If $f(\gamma_{u_-},K_u)$ denotes $\textrm{Var}(Z_u^N)$ written explicitly as a function of $(\gamma_{u_-},K_u)$, then, for all $(\gamma_{u_-},K_u)$~in~$\cal{S}$ and $C>0$,
\begin{align}
C\sqrt{K_u\omega_u(K_u)^2}\rho_{\cal{C}_u}&\in\arginf_{\gamma\in\cal{S}(K_u)}f(\gamma,K_u),\label{eq:locally11}\\
M_u&\in\arginf_{K\in\cal{S}(\gamma_{u_-})}f(\gamma_{u_-},K),\label{eq:locally12}\\
(C\mu_u^{u_-},M_u)&\in\arginf_{(\gamma,K)\in\cal{S}}f(\gamma,K),\label{eq:locally2}
\end{align}
where, assuming that it exists, $M_u$ denotes the regular conditional probability distribution mapping $\mu_u^{u_-}$ to $\mu_u$ (i.e.\ $\mu_u^{u_-}\times M_u=\mu_u$), $\omega_u(K_u):=d\rho_u/d\rho_{\cal{C}_u}\times K_u$, and
$$(K_u\omega_u(K_u)^2)(\bm{x}_{\cal{C}_u})=\int K_u(\bm{x}_{\cal{C}_u},dx_u)\omega_u(K_u)(\bm{x}_{\cal{C}_u},x_u)^2\quad\forall \bm{x}_{\cal{C}_u}\in\bm{E}_{\cal{C}_u}.$$
\end{theorem}
\begin{proof}See Appendix~\ref{app:locally}.
\end{proof}
\subsection{Comparison with ASMC}\label{sec:dacvssmc}Theorems~\ref{THRM:LLN}--\ref{THRM:CLT} demonstrate that DaC-SMC is a well-founded algorithm: it leads to estimators for the unnormalized  and normalized targets with the same basic properties as those possessed by  SMC and ASMC estimators~\cite{DelMoral2004}. It is  only natural to now ask whether it performs better than its  ASMC analogue (Algorithm~\ref{alg:asmc}).

To answer this question, suppose we are interested in approximating the final target $\rho_{\mathfrak{r}}$ and its normalization $\mu_{\mathfrak{r}}$. We first need to figure out how to approximate $\rho_{\mathfrak{r}}$ and $\mu_{\mathfrak{r}}$ using ASMC. This requires somehow lumping the nodes in $\mathbb{T}$ to obtain a line $\mathbb{L}$ (Remark~\ref{rem:smc1}). While there are many ways that it could be done, we focus on perhaps the most obvious: merge together all nodes in each level of the tree to obtain a line $\cal{L}_0\to\cal{L}_1\to\dots\to\cal{L}_T=\{\rf\}$ whose length  equals $\mathbb{T}$'s depth.  
Taking the corresponding products of the auxiliary measures,
$$\gamma_{t_-}:=\prod_{u\in \cal{L}_t^{\not \partial}}\gamma_{u_-},\quad\pi_{t_-}:=\prod_{u\in \cal{L}_t^{\not \partial}}\pi_{u_-},\quad\forall t=1,\dots,T,$$
where $\cal{L}_t^{\not \partial}:=\cal{L}_t\cap \mathbb{T}^{\not\partial}$ denotes set of non-leaf nodes in the $t^{th}$ level,  we obtain sequences $(\gamma_{t_-})_{t=1}^T$ and $(\pi_{t_-})_{t=1}^T$ that can be used in a standard ASMC set-up. In particular, replacing $(\gamma_{u_-})_{u\in\mathbb{T}^{\not \partial}}$, $(\pi_{u_-})_{u\in\mathbb{T}^{\not \partial}}$, and $(K_u)_{u\in\mathbb{T}}$ with $(\gamma_{t_-})_{t=1}^T$, $(\pi_{t_-})_{t=1}^T$,  and $(K_t)_{t=0}^{T}$, for some proposal kernels $K_t:\bm{E}_{t_-}\times\cal{E}_{t}\to [0,1]$ mapping from $\bm{E}_{t_-}:=\prod_{u\in \cal{L}_t^{\not \partial}}\bm{E}_{\cal{C}_u}$ to $\cal{E}_{t}:=\cal{E}_{\cal{L}_{t}}$, Algorithm~\ref{alg:adacsmc} reduces to Algorithm~\ref{alg:asmc}. To keep both algorithms on as equal a footing as possible, we set the ASMC proposal kernels to be the corresponding products of the DaC-SMC ones:
\begin{equation}\label{eq:factorizedkernels}K_t(\bm{x}_{t_-},dx_{t}):=\left(\prod_{u\in\cal{L}_t^{\partial}} K_u(dx_u)\right)\left(\prod_{u\in\cal{L}_t^{\not\partial}} K_u(\bm{x}_{\cal{C}_u},dx_u)\right)\quad\forall t=0,\dots,T,\end{equation}
where $\cal{L}_t^{\partial}:=\cal{L}_t\cap \mathbb{T}^{\partial}$ denotes set of leaves in the $t^{th}$ level. 

Theoretically, Theorems~\ref{THRM:LLN}--\ref{THRM:CLT} tell us that both ASMC and DaC-SMC produce consistent estimators of $\rho_{\mathfrak{r}}$ and $\mu_{\mathfrak{r}}$, unbiased in the case of $\rho_{\mathfrak{r}}$ and with an $\cal{O}(N^{-1})$ bias in that of $\mu_{\mathfrak{r}}$, whose errors are  $\cal{O}(N^{-1/2})$ and asymptotically normal. Comparing the asymptotic variance expressions in Theorem~\ref{THRM:CLT} for both algorithms, we find that the DaC-SMC estimators are more statistically efficient than the ASMC ones: 
\begin{theorem}\label{THRM:SMCVSDAC}If Assumptions~\ref{ass:abscont}--\ref{ass:boundweights} are satisfied, $\varphi$ belongs to $\cal{B}_b(\bm{E}_{\mathfrak{r}})$, and  $\sigma^2_{\rho_{\mathfrak{r}},smc}(\varphi)$ and $\sigma^2_{\mu_{\mathfrak{r}},smc}(\varphi)$ denote the respective asymptotic variances of the ASMC estimators for $\rho_{\mathfrak{r}}(\varphi)$~and~$\mu_{\mathfrak{r}}(\varphi)$,   then
$$ \sigma^2_{\rho_{\mathfrak{r}}}(\varphi)\leq \sigma^2_{\rho_{\mathfrak{r}},smc}(\varphi),\qquad \sigma^2_{\mu_{\mathfrak{r}}}(\varphi)\leq \sigma^2_{\mu_{\mathfrak{r}},smc}(\varphi).$$
In particular, the asymptotic variance of the DaC-SMC normalizing constant estimate is bounded above by that of the SMC estimate.
\end{theorem}
\begin{proof}See Appendix~\ref{app:dacvssmcproofs}.
\end{proof}

Of course, the above theorem, and this entire comparison,  applies only if we are using factorized proposal kernels in~\eqref{eq:factorizedkernels} for ASMC, a needless restriction for this algorithm and one that will often be detrimental to its performance. Indeed, unless the intermediate targets exhibit the same kind of conditional independence encoded in the factorized kernels, standard results for ASMC (e.g.\ \cite[Theorems~10.1,10.2]{chopin2020} or \cite[Proposition~2]{Guarniero2017}) tell us that these kernels will not be optimal in any usual sense.

However, for such factorized proposals, the gains in statistical efficiency can sometimes be drastic (e.g.\ exponential in the tree's degree); see Examples~2~and~4 and Section~3.4 in~\cite{Kuntz2021}. 
In exchange, DaC-SMC generally has a higher computational cost. More specifically, the mutation steps of DaC-SMC require drawing $N$ samples from $K_u$ for each $u$ in $\mathbb{T}$, while ASMC  requires drawing $N$ samples from $K_t$ for each $t=0,\dots, T$. However, $K_t$'s definition implies that drawing a sample from it is equivalent to drawing one from each $K_u$ with $u$ in $\cal{L}_t$ and concatenating them. Hence, the cost incurred by the mutation steps of both algorithms is the same. 
The cost of the correction and resampling steps, however, is not the same. In the case of ASMC, we need to sample $N$ particles from $\pi_{t_-}^N$ for each $t$ in $0,\dots, T$ which requires only $N$ evaluations of $w_{t_-}=\prod_{u\in\cal{L}_t^{\not \partial}}w_{u_-}$ per $t$. For DaC-SMC, we must instead sample $N$ particles  from $\pi_{u_-}^N$ for each $u$ in $\mathbb{T}^{\not\partial}$. In the absence of any sort of special structure in the auxiliary measures, this involves $N^{c_u}$ evaluations of $w_{u_-}$ per $u$ in $\mathbb{T}^{\not \partial}$. Hence, for general auxiliary measures, DaC-SMC's cost is $\cal{O}(N^d)$, where $d$ denotes $\mathbb{T}$'s  degree,   while ASMC's is just $\cal{O}(N)$.

However, this extra cost materializes only in nodes with more than one child and concentrates in those with $d$ children. If these nodes only feature high up in the tree, then the high cost incurred by the generation of the particles indexed by the nodes' children can outweigh the extra overhead incurred by the resampling and correction steps (and similarly if the proposal kernels for those nodes are particularly expensive to sample from). It is worth noting that in many settings of practical interest (e.g.\ \cite{Corneflos2021,ding2019,ding2018}), $d=2$ and this cost can be borne. Otherwise, as described in Section~\ref{sec:lcdacsmc}, DaC-SMCs cost may be brought down by choosing  auxiliary measures that are partially or fully factorized (or sums of such measures)---and when these differ too much from the optimal auxiliary measures the impact of this discrepancy upon the variance can be mitigated via tempering~\cite{delmoral2006sequential}. Lastly, in cases where  $d$ is not small and the preferred auxiliary measures do not factorize sufficiently, one can instead turn to the `incomplete' estimators~\eqref{eq:incomplete} similar to those successfully employed in the U-statistic literature to tackle analogous issues (cf.\ \cite{kong2021design} and references therein).
\section{Discussion}\label{sec:discussion}
In this paper, we study DaC-SMC algorithm introduced in~\cite{lindsten2017divide} and show that it is theoretically well-founded:  it possesses the same basic properties that standard SMC algorithms do (Section~\ref{sec:theory}). To achieve this, we combine well-known methods previously used to study standard SMC algorithms with novel techniques  that control the errors introduced by the product-form estimators~\eqref{eq:gammauN} embedded within DaC-SMC (see Section~\ref{sec:mop} for an overview and Appendices~\ref{app:lplln}--\ref{app:cross} for the details). 
Our analysis here can be sharpened and refined in various ways. 
The positivity requirements in Assumptions~\ref{ass:abscont} can be circumvented by emulating the use of stopping times in \cite{DelMoral2004}. The boundedness ones in Assumptions~\ref{ass:boundweights} can be avoided   by introducing families of appropriately integrable functions at each node, e.g.\ similarly as in~\cite{Chopin2004,Douc2008}. 
While the additional requirement in Theorem~\ref{THRM:BIAS} that the weight functions are lower bounded is common in the SMC literature (e.g.\ \cite[Proposition 9.5.6]{DelMoral2013} and \cite{olsson2004bootstrap}), we anticipate that it could be relaxed by considering a more restricted class of test functions (than that of all bounded and measurable ones) and dealing with the possible `extinction' of the particle cloud as done in~\cite[Theorem 7.4.3]{DelMoral2004}. 
Lastly, we do not believe that the particular lumping construction in Section~\ref{sec:dacvssmc} is crucial for the variance bounds in Theorem~\ref{THRM:SMCVSDAC}  to hold, only that the ASMC auxiliary measures and proposals are obtained by taking appropriate products of the DaC-SMC ones. In fact, in Appendix~\ref{app:locally} we actually show that lumping the children of a node never reduces the  variance of DaC-SMC estimators. However, to argue the bounds in full generality one would likely have to further  show that lumping `cousins' does not lower the variance, and we opted to focus on that particular lumping construct so not to overly complicate Section~\ref{sec:dacvssmc} and Theorem~\ref{THRM:LOCALLY}'s proof.

Our analysis sheds light on the algorithm, its limitations, and its practical use. First, one can improve the algorithm's performance by approximating optimal intermediate targets, auxiliary measures, and proposal kernels identified in Section~\ref{sec:optimalprop}  (this provides a formal explanation for the success of the ad hoc strategies taken in~\cite[Section~4.1]{Corneflos2021} and \cite[Sections~3.7.4,~3.10.5]{ding2019}). Next, whenever the auxiliary measures  and proposal kernels used for DaC-SMC are the corresponding products of those used for standard ASMC, the former proves more statistically efficient than the latter (Section~\ref{sec:dacvssmc}): using the same number of particles, DaC-SMC estimators achieve lower variances than their standard counterparts. In exchange, DaC-SMC generally carries a higher computational cost. How much higher depends on the amount of product structure present in the auxiliary measures $(\gamma_{u_-})_{u\in\mathbb{T}^{\not \partial}}$ employed (Section~\ref{sec:lcdacsmc}). If they are fully factorized (or sums of such measures), then the DaC-SMC's cost is $\cal{O}(N)$ where $N$ denotes the number of particles (the same as ASMC). If they do not factorize at all (nor break down into sums of partially-factorized functions), then DaC-SMC's cost is $\cal{O}(N^d)$ where $d$ denotes the tree's degree. For anything between these two extremes, the cost is $\cal{O}(N^k)$ where $1<k<d$ depends on the amount of factorization (the greater it is, the lower $k$ is).

At each node $u$, the algorithm approximates the auxiliary measure $\gamma_{u_-}$ with the product-form $\gamma^N_{\cal{C}_u}$ and then corrects this approximation using the weight function $w_{u_-}$~\eqref{eq:gammauN}. The corrected approximation, $\gamma_{u_-}^N$, is then extended using a mutation step, and the resulting extended approximation, $\gamma_u^N$, is corrected using the weight function $w_u$~\eqref{eq:adacests} to produce an approximation, $\rho_u^N$, of the intermediate target $\rho_u$. This means that any non-product structure in $\rho_u$ must be introduced into our approximations via these two re-weighting steps (aside from structure linking the components indexed by $u$'s children with $u$ itself, which can be introduced using the mutation kernels). Consequently, if the intermediate target $\rho_u$ possesses pronounced non-product structure, we end up with competing interests: (a) we would like to choose the auxiliary measure $\gamma_{u_-}$ as close to  factorized as possible  so that the algorithm's cost is low (Section~\ref{sec:lcdacsmc}) and the first re-weighting does not necessitate large corrections and lead to high-variance approximations of $\gamma_{u_-}$; and (b) we would like to choose the auxiliary measure $\gamma_{u_-}$ close $\rho_u$. In fact, the analysis in Section~\ref{sec:optimalprop} suggests that, purely in terms of variance and disregarding any considerations of the algorithm's cost, $\gamma_{u_-}$ should be chosen as close as possible to the corresponding marginal of $\rho_u$. How to balance these two competing considerations  is an interesting  question beyond the scope of this paper and likely best dealt with on a case-by-case basis.

Alternatively, it might also be possible to introduce non-product structure without blowing up the computational cost or the variance in the correction steps using tempering~\cite{delmoral2006sequential}, `incomplete' versions of product-form estimators~\eqref{eq:incomplete}, or generalizations of these estimators that account for more complicated dependence structures (see Appendix~\ref{app:variantex} for an example). We find each of these variants, and the question of which one should be used under what circumstances, worthy of further investigation; and we believe this can be done using extensions of the techniques we employed in Appendices~\ref{app:lplln}--\ref{app:cross}. Similarly for the other, perhaps more familiar, refinements that will also likely improve the algorithm's performance in practice (e.g.\ adaptive resampling, more efficient resampling schemes, and varying the computational power spent at each node; see Section~\ref{sec:variants}). 

More broadly, we believe that the techniques and intermediate results developed in this paper might find applications beyond DaC-SMC and its variants. Here, we have in mind recently developed Monte Carlo algorithms for inference in models defined as a merger of several submodels, and motivated  by either  a genuine interest in the submodel merger (as in~\cite{Goudie2019,Manderson2021}) or computational reasons such as facilitating distributed implementations and privacy concerns (e.g.~\cite{Dai2019,Rendell2021,Dai2021}).

Lastly, we anticipate that a variety of interesting applications of the DaC-SMC algorithm will involve its merger with other well-known Monte Carlo methodology and our results pave the way for such combinations. For instance, its merger with pseudo-marginal MCMC~\cite{andrieu2009pseudo} to tackle targets with intractable densities and that with particle-marginal Metropolis-Hastings~\cite{andrieu2010particle} to allow for Metropolis-Hastings proposals with intractable densities; two approaches  justified by the unbiasedness in Theorem~\ref{THRM:UNBIAS}. It can also be used within general particle MCMC~\cite{andrieu2010particle}, including particle Gibbs (e.g.\, pertinent for models such as Markov random fields with collections of correlated parameter variables lacking time-series structure), the formal justification of which requires the characterization of the joint law of all random variables generated during the running of the Algorithm~\ref{alg:adacsmc}. This characterization can be found in \cite[Chapter 8]{crucinio2021some} which generalizes the argument given in \cite[Appendix A.1]{lindsten2017divide} for balanced binary  trees.

In summary, DaC-SMC is a theoretically-sound algorithm and a promising extension of SMC. Its development, however, still lies in its infancy and there remain many open questions regarding its use. For instance, `how exactly should we carry out the resampling and which nodes should we allocate more computational power to?', `how do we introduce non-product structure into the approximations while balancing the algorithm's cost and the variance of the correction steps?', and `with which other Monte Carlo algorithms would it prove fruitful to merge DaC-SMC and for what type of targets should which merger be used for?'. We look forward to their resolution.

\begin{appendix}

\begin{center}
\huge Appendices
\end{center}

Appendix~\ref{app:example} contains an example illustrating several concepts discussed in Sections~\ref{sec:lcdacsmc}--\ref{sec:variants}. The other appendices contain the proofs of Theorems~\ref{THRM:LLN}--\ref{THRM:SMCVSDAC} given in the main text.
In particular, Appendix~\ref{app:notation} contains notation used throughout Appendices~\ref{app:lplln}--\ref{app:dacvssmcproofs}, Appendix~\ref{app:lplln} contains the proofs for the DaC-SMC estimators' consistency and $L^p$ inequalities (Theorems~\ref{THRM:LLN},~\ref{THRM:WEAK}, and~\ref{THRM:LP}), Appendix~\ref{app:unbias} that for the unbiasedness of the unnormalized target estimators (Theorem~\ref{THRM:UNBIAS}), Appendix~\ref{app:bias} that for the bias bound of the normalized target ones (Theorem~\ref{THRM:BIAS}),  Appendix~\ref{app:clt} that for the asymptotic normality (Theorem~\ref{THRM:CLT}), Appendix~\ref{app:cross} that for the product $L^2$ inequality (Theorem~\ref{THRM:CROSS}), Appendix~\ref{app:locally} the derivation of the locally optimal proposals and auxiliary measures (Theorem~\ref{THRM:LOCALLY}), and Appendix~\ref{app:dacvssmcproofs} the argument showing that the asymptotic variances of DaC-SMC estimators are bounded above by those of their SMC counterparts (Theorem~\ref{THRM:SMCVSDAC}).
\section{A further example}\label{app:example}
Consider the following time-varying version of the toy hierarchical model in~\cite[Section~3.1]{Kuntz2021}:
$$Y_{t,l} \sim\cal{N}(X_{t,l},1),  \enskip X_{t,l} \sim\cal{N}(0,\theta_t),\enskip \forall l\in[L], \quad\Theta_{t}\sim f(d\theta;\Theta_{t-1}), \enskip \forall t\in [T],\enskip \Theta_0\sim f(d\theta;1),$$
where $T,L>0$ denote integers,  $Y_{0:T,1:L}:=((Y_{t,l})_{l=1}^L)_{t=0}^T$ some observed variables,  $X_{0:T,1:L}$ $:=((X_{t,l})_{l=1}^L)_{t=0}^T$  some latent variables, $\Theta_{0:T}=(\Theta_t)_{t=0}^T$ some unknown parameters, and $f(d\theta;\eta)$ the exponential distribution with mean $\eta$. Suppose we observe data $y_{0:T,1:L}$ and wish to draw inferences from the smoothing distribution proportional to
\begin{align*}\rho(d\theta_{0:T},dx_{0:T,1:L}):=&f(d\theta_0;1)\left(\prod_{t=1}^Tf(d\theta_t;\theta_{t-1})\right)
\\&\times
\left(\prod_{t=0}^T\prod_{l=1}^L\cal{N}(y_{t,k};x_{t,l},1)\cal{N}(dx_{t,l};0,\theta_t)\right).\end{align*}
We consider two approaches with which we could tackle this problem:
\subsection{The DaC-SMC approach}\label{app:timevar}
One way to obtain tractable approximations of the posterior $\rho$ is to apply DaC-SMC with the following tree. Each parameter variable $\Theta_t$ has a node $t$ assigned to it (with $E_t=[0,\infty)$), each latent variable $X_{t,l}$ has a node $(t,l)$ assigned to it (with $E_{t,l}=\r$), and the tree has an extra `dummy' root node $T+1$ (which we introduce purely for notational convenience).  The nodes assigned to $\Theta_0$ and $X_{1:T,1:L}$ are the tree's leaves and, for each $t$ in $[T+1]$, the children of $t$ are $t-1$ and $(t-1,1),\dots,(t-1,L)$. For the auxiliary measure $\gamma_{t_-}$ associated with the node $t$ assigned to $\Theta_t$, we simply remove all terms in $\rho$ that involve $\theta_{s}$ and $x_{s,1},\dots,x_{s,L}$ with $s\geq t$:
\begin{align*}\gamma_{(t+1)_-}(d\theta_{0:t},dx_{0:t,1:L}):=&f(d\theta_0;1)\left(\prod_{t'=1}^{t}f(d\theta_{t'};\theta_{t'-1})\right)
\\&\times
\left(\prod_{t'=0}^{t}\prod_{l=1}^L\cal{N}(y_{t',l};x_{t',l},1)\cal{N}(dx_{t',l};0,\theta_{t'})\right)\quad\forall t=0,\dots,T,\end{align*}
where, with the usual abuse of notation, we allow $\mathcal{N}(\cdot;\mu,\Sigma)$ to denote both a normal probability measure of mean $\mu$ and variance $\Sigma$ and the Lebesgue denbsity of that measure.
(Note that $\gamma_{(T+1)_-}$ coincides with the unnormalized smoothing distribution $\rho$, so $\pi_{(T+1)_-}^N$  approximates the smoothing distribution.) For the proposal kernels we pick
\begin{align}K_0(d\theta_0)&:=f(d\theta_0;1),\quad K_t(\theta_{t-1},d\theta_t):=f(d\theta_t;\theta_{t-1})\enskip\forall t=1,\dots,T,\label{eq:exker1}\\
K_{t,k}(dx_{t,l})&:=\cal{N}(dx_{t,l};y_{t,l},1)\enskip \forall  t=0,\dots,T,\enskip l=1,\dots,L;\label{eq:exker2}\end{align}
in which case
$$w_{(t+1)_-}(\theta_{t},x_{t,1:L})=\prod_{l=1}^L\cal{N}(x_{t,l};0,\theta_{t})\quad\forall t=0,\dots,T.$$
Hence, with $g_{t,l}(\theta):=\sum_{n_l=1}^N\cal{N}(X_{t,l}^{n_l};0,\theta)$  and $g_t(\theta):=\prod_{l=1}^Lg_{t,l}(\theta)$,
\begin{align}
\frac{N^{L+1}}{\cal{Z}_{\cal{C}_t}^N}\gamma_{(t+1)_-}^N&=\sum_{n=1}^N\left(\prod_{l=1}^L \sum_{n_l=1}^N\cal{N}(X_{t,l}^{n_l};0,\Theta_{t}^{n})\delta_{X_{t,l}^{n_l}}\right)\delta_{(\Theta_{1:t}^n,X^n_{1:t-1,1:L})}\label{eq:dna8dgragbwabyfbwafaf}\\
&=\sum_{n=1}^N\left(\prod_{l=1}^L \sum_{n_l=1}^N\frac{\cal{N}(X_{t,l}^{n_l};0,\Theta_{t}^{n})}{g_{t,l}(\Theta_t^n)}\delta_{X_{t,l}^{n_l}}\right)g_t(\Theta_t^n)\delta_{(\Theta_{1:t}^n,X^n_{1:t-1,1:L})}.\nonumber
\end{align}
(Note that we are omitting the $N$ superscripts from $\Theta^n_t$ and $X^n_{t,l}$ to simplify the notation.) Setting $\omega_{t,l}^{n,n_l}:=\frac{\cal{N}(X_{t,l}^{n_l};0,\Theta_t^n)}{g_{t,l}(\Theta_t^n)}$ and $\omega_t^n:=\frac{g_t(\Theta_t^n)}{\sum_{m=1}^Ng_{t}(\Theta_t^m)}$, and normalizing, we find that
$$\pi_{(t+1)_-}^N=\sum_{n=1}^N\left(\prod_{l=1}^L \sum_{n_l=1}^N\omega_{t,l}^{n,n_l}\delta_{X_{t,l}^{n_l}}\right)\omega_t^n\delta_{(\Theta_{1:t}^n,X^n_{1:t-1,1:L})}.$$
Hence, we are able to draw $N$ samples 
from $\pi_{(t+1)_-}^N$ in $\cal{O}(N^2)$ operations by: $(a)$ evaluating $(\omega_t^n)_{n=1}^N$; $(b)$ drawing $N$ indices $(m_n)_{n=1}^N$ from $(\omega_t^n\delta_{n})_{n=1}^N$ using the alias method; $(c)$ for each $n$ and $l$, evaluating $(\omega_{t,l}^{m_n,k})_{k=1}^N$; $(d)$ for each $n$ and $l$, drawing an index $k_{n,l}$  from $(\omega_t^{m_n,k}\delta_{k})_{k=1}^N$; and $(e)$ setting 
$$\bm{X}_{(t+1)_-}^n:=(\Theta_{1:t}^{m_n},X_{1:t-1,1:L}^{m_n},X_{t,1}^{k_{n,1}},\dots,X_{t,L}^{k_{n,L}})\quad\forall n\in[N].$$
This process can be trivially sped up by  storing and re-using the weights $(\omega_t^{m_n,k})_{k=1}^N$ for indices $m_n$ drawn more than once. If there are $m_n$s drawn many times, one may be better off computing probability and alias tables for the corresponding weights $(\omega_t^{m_n,k})_{k=1}^N$ and using these to generate the corresponding $k_n$s.
\subsection{An alternative approach}\label{app:variantex}
This model has two types of direct dependencies linking its variables: $\Theta_{t+1}$ is drawn from an exponential distribution with mean $\Theta_t$ and $X_{t,l}$ is drawn from a zero-mean Gaussian distribution with variance $\Theta_t$. The first of these two is encoded into the kernels in~(\ref{eq:exker1},\ref{eq:exker2}): we propose a value for $\Theta_{t+1}^{n}$ by sampling $f(d\theta_{t+1};\Theta_t^{n})$.  The second type, however, is not:   $X_{t,l}^{n}$ are drawn independently of $\Theta_t^{n}$ by sampling from $\cal{N}(y_{t,l},1)$.

To   account for the second type of dependency, we could instead try  drawing 
$$\Theta^{n}_t\sim f(d\theta_t;\Theta_{t-1}^n)\enskip\text{(or $f(d\theta_0;1)$ if $t=0$)},\quad X_{t,1:L}^{n,m}\sim \prod_{l=1}^L\cal{N}\left(\frac{\Theta_t^n}{1+\Theta_t^n}y_{t,l},\frac{\Theta_t^n}{1+\Theta_t^n}\right),$$
for each for all $n$ and $m$ in $[N]$; resulting in a total cost of $\cal{O}(N^2)$.  We would then have to replace the product-form approximation to $\gamma_{(t+1)_-}$ in~\eqref{eq:dna8dgragbwabyfbwafaf} with the following `partially product-form' one~\cite[Section~3.2]{Kuntz2021}:
\begin{align*}\gamma_{(t+1)_-}^N:&=\frac{\cal{Z}_t^N}{N^{L+1}}\sum_{n=1}^N\left[\prod_{l=1}^L\sum_{n_l=1}^N\cal{N}(y_{t,l};0,1+\Theta_{t}^n)\delta_{X^{n,n_l}_{t,l}}\right]\delta_{(\Theta_{1:t}^n,X^n_{1:t-1,1:L})}\\
&=\frac{\cal{Z}_t^N}{N^{L+1}}\sum_{n=1}^N\left[\prod_{l=1}^L\frac{1}{N}\sum_{n_l=1}^N\delta_{X^{n,n_l}_{t,l}}\right]g_t(\Theta_t^n)\delta_{(\Theta_{1:t}^n,X^n_{1:t-1,1:L})},
\end{align*}
where $g_t(\theta):=\prod_{l=1}^L\cal{N}(y_{t,l};0,1+\theta)$ and $\cal{Z}_t^N$ denotes the mass of the previously computed $\gamma_{t_-}^N$ (with $\cal{Z}_0^N:=1$). Drawing $N$ samples  from the approximation's normalization,
$$\pi_{(t+1)_-}^N:=\frac{\gamma_{(t+1)_-}^N}{\cal{Z}_{(t+1)_-}^N}=\sum_{n=1}^N\left[\prod_{l=1}^L\frac{1}{N}\sum_{n_l=1}^N\delta_{X^{n,n_l}_{t,l}}\right]\omega_t^n\delta_{(\Theta_{1:t}^n,X^n_{1:t-1,1:L})}$$
where $\omega_t^n:=g_t(\Theta_t^n)/\sum_{m=1}^Ng_t(\Theta_t^m)$, can then be done in $\cal{O}(N)$ operations by: $(a)$ evaluating $(\omega_t^n)_{n=1}^N$; $(b)$ drawing $N$ indices $(m_n)_{n=1}^N$ from $(\omega_t^n\delta_{n})_{n=1}^N$ using the alias method; $(c)$ for each $n$ and $l$, drawing an index $k_{n,l}$ from $(N^{-1}\delta_{k})_{k=1}^N$  ; and $(d)$ setting 
$$\bm{X}_{(t+1)_-}^n:=(\Theta_{1:t}^{m_n},X_{1:t-1,1:L}^{m_n},X_{t-1,1}^{m_n,k_{n,1}},\dots,X_{t-1,L}^{m_n,k_{n,L}})\quad\forall n\in[N].$$%
\section{Important notation for the proofs}\label{app:notation}
Throughout  Appendices~\ref{app:lplln}--\ref{app:dacvssmcproofs}, we employ the notation introduced in Sections~\ref{sec:probdef}--\ref{sec:dacalg} (cf. Table~\ref{tab:notation} for a summary) and the following:
\begin{itemize}
\item Unless specified otherwise, $(\Omega,\cal{F},\mathbb{P})$ denotes the underlying probability space on which all random variables introduced in Algorithm~\ref{alg:adacsmc} (including those indexed by different $N$s) are jointly defined, and  $\mathbb{E}$ denotes expectations with respect to $\mathbb{P}$.
\item $C$ denotes a generic positive constant dependent on $u$, $v$, and $p$ but not on $\varphi$ and $N$ that may change from one line to the next. We use this notation whenever we are interested only in the existence of such a constant and its precise value is unimportant. To not overly repeat ourselves, throughout the appendices, we often introduce $C$ without mention (e.g.\ Lemma~\ref{lem:borelcantelli} below) and leave the corresponding `there exists some $C<\infty$ such that' or `for some $C<\infty$' statements implicit. 
\item Given any two subsets $A\subseteq B\subseteq\mathbb{T}$, a measure $\nu$ on $(E_A,\cal{E}_A)$, and measurable function $\psi$ on $(E_B,\cal{E}_B)$, we use $\nu(\psi)$ to denote the measurable function on $(E_{B\backslash A},\cal{E}_{B\backslash A})$ obtained by integrating the arguments of $\psi$ indexed by $A$ with respect to $\nu$:
$$\nu(\psi)(x_{B\backslash A}):=\int\psi(x_A,x_{B\backslash A})\nu(dx_A)\quad\forall x_{B\backslash A}\in E_{B\backslash A},$$
under the assumption that the integral is well-defined for all $x_{B\backslash A}$ in $E_{B\backslash A}$.
\item For any $u$ in $\mathbb{T}^{\not \partial}$ and $N>0$,  $\epsilon_u^N:=N^{-1}\sum_{n=1}^N\delta_{\bm{X}^{n,N}_{u_-}}$ denotes the empirical distribution of the resampled particles (not to be confused with $\pi_u^N$ in~\eqref{eq:actualtargetdistributions} which is empirical distribution of these particles after mutation) and $\pi_{\cal{C}_u}^N:=\prod_{v\in\cal{C}_u}\pi_v^N$ denotes the product of the particle approximations to the extended normalized auxiliary measures indexed by $u$'s children. Additionally, $\cal{F}_{\cal{C}_u}^N$ denotes the sigma algebra generated by the mutated particles $(\bm{X}_{\cal{C}_u}^{n,N})_{n=1}^N:=(\bm{X}_{u1}^{n,N},\dots,\bm{X}_{uc_u}^{n,N})_{n=1}^N$ indexed by $u$'s children, and that $\cal{F}_{u_-}^N$ that generated by the resampled particles $(\bm{X}_{u_-}^{n,N})_{n=1}^N$ indexed by $u$.
\end{itemize}
\section{Proofs of Theorems~\ref{THRM:LLN},~\ref{THRM:WEAK},~and~\ref{THRM:LP}}\label{app:lplln}
Throughout this appendix, we use the notation described in Appendix~\ref{app:notation}. The aim of the appendix is to prove  $L^p$~inequalities in Theorem~\ref{THRM:LP}. The laws   of large numbers in Theorem~\ref{THRM:LLN} follow from the inequalities and the ensuing consequence of the Borel-Cantelli lemma:
\begin{lemma}[{e.g.\ \cite[p.~17]{boustati2020generalised}}]\label{lem:borelcantelli}Let $(Z_N)_{N=1}^\infty$ be a sequence of real-valued random variables. If there exists a $p>2$ such that for all $N > 0$, 
$\Ebb{\mmag{Z_N}^p}\leq {C} {N^{-p/2}}$,
then $Z_N\to0$ with probability one as $N\to\infty$.
\end{lemma}
The weak convergence in Theorem~\ref{THRM:WEAK} follows directly:
\begin{proof}[Proof of Theorem~\ref{THRM:WEAK}]
Because the product of (finitely-many) Borel $\sigma$-algebras on Polish spaces coincides with the Borel $\sigma$-algebra on the product space, as the former are separable cf.\ \cite[Lemma 1.2]{Kallenberg2021}, and the latter is also Polish, this follows immediately from Theorem~\ref{THRM:LLN} and the equivalence established in \cite[Theorem 2.2]{berti2006almost}. See \cite[Section S1, Theorem 1]{Schmon2021} for details in an SMC context.
\end{proof}
To establish the $L^p$ inequalities in Theorem~\ref{THRM:LP}, we first obtain such results for $(\pi_u^N)_{u\in\mathbb{T}}$:
\begin{theorem}[$L^p$ inequalities for $(\pi_u^N)_{u\in\mathbb{T}}$]\label{thrm:lpapp}If Assumptions~\ref{ass:abscont}--\ref{ass:boundweights} hold, then
\begin{align*}
\Ebb{\mmag{\pi^N_u(\varphi)-\pi_u(\varphi)}^p}^{\frac{1}{p}}\leq\frac{C\norm{\varphi}}{N^{1/2}},\quad \forall N>0,\enskip\varphi\in\cal{B}_b(\bm{E}_u),\enskip p\geq1,\enskip u\in\mathbb{T}.
\end{align*}
\end{theorem}  
To prove the above we follow the  approach  for standard SMC in \cite{Crisan2002} adapted to cover the case of general $p$ (as in, for example, \cite{miguez2013convergence}). In particular, we derive an $L^p$ inequality for the empirical distributions of the particles indexed by the tree's leaves and show that each step of the algorithm preserves this inequality. Let's start: the  Marcinkiewicz-Zygmund-type  inequality in \cite[Lemma~7.3.3]{DelMoral2004} shows that
$$\Ebb{\mmag{\pi^N_u(\varphi)-\pi_{u}(\varphi)}^m}^{\frac{1}{m}}\leq\frac{C\norm{\varphi}}{N^{1/2}}\quad\forall N>0,\enskip \varphi\in \cal{B}_b(\bm{E}_{u}),$$
for all $u$ in $\mathbb{T}^\partial$ and  positive integers $m$. We extend the above to all $p\geq 1$ using Jensen's inequality: if $m$ is the smallest integer no smaller than $p$, then 
\begin{align}
\Ebb{\mmag{\pi^N_u(\varphi)-\pi_{u}(\varphi)}^p}^{\frac{1}{p}}&=\Ebb{\left(\mmag{\pi^N_u(\varphi)-\pi_{u}(\varphi)}^m\right)^{\frac{p}{m}}}^{\frac{1}{p}}\leq\Ebb{\mmag{\pi^N_u(\varphi)-\pi_{u}(\varphi)}^m}^{\frac{1}{m}}\nonumber\\
& \leq\frac{C\norm{\varphi}}{N^{1/2}}\quad\forall N>0,\enskip \varphi\in \cal{B}_b(\bm{E}_{u}).\label{eq:lp0}
\end{align}
Taking the product of $\pi_v^N$ over all children $v$ of a node $u$ preserves~\eqref{eq:lp0}:
\begin{lemma}[Product step]\label{lem:lpcoa}If, in addition to Assumptions~\ref{ass:abscont}--\ref{ass:boundweights},~\eqref{eq:lp0} is satisfied for each child $v$ (i.e.\ it holds with $v$ replacing $u$ therein)  of a node $u$ in $\mathbb{T}^{\not \partial}$ and some $p\geq1$, 
then
\begin{equation}\label{eq:lpcoa}
\Ebb{\mmag{\pi^N_{\cal{C}_u}(\varphi)-\pi_{\cal{C}_u}(\varphi)}^p}^{\frac{1}{p}}\leq\frac{C\norm{\varphi}}{N^{1/2}}\quad\forall N>0,\enskip \varphi\in \cal{B}_b(\bm{E}_{\cal{C}_u}).
\end{equation}
\end{lemma}
The key to proving the above is the following lemma giving an $L^p$ inequality for  products of independent random probability measures satisfying their own individual $L^p$ inequalities:
\begin{lemma}\label{lem:prodlp}Suppose that $(\eta_1^N)_{N=1}^\infty$ and $(\eta_2^N)_{N=1}^\infty$ are independent sequences of random probability measures defined on some common probability triplet $(\Omega,\cal{F},\mathbb{P})$, respectively taking values in  some measurable spaces $(S_1,\cal{S}_1)$ and $(S_2,\cal{S}_2)$, and satisfying
\begin{equation}\label{eq:etalp}\Ebb{\mmag{\eta^N_k(\varphi)-\eta_{k}(\varphi)}^p}^{\frac{1}{p}}\leq\frac{C\norm{\varphi}}{N^{1/2}}\quad\forall N>0,\enskip \varphi\in \cal{B}_b(S_{k}),\enskip k\in\{1, 2\},\end{equation}
for some $p\geq 1$, with limits $\eta_1$ and $\eta_2$ that are also probability measures. Then,
\begin{equation}\label{eq:etaprodlp}
\Ebb{\mmag{(\eta^N_1\times\eta^N_2)(\varphi)-(\eta_1\times\eta_2)(\varphi)}^p}^{\frac{1}{p}}\leq\frac{C\norm{\varphi}}{N^{1/2}}\quad\forall N>0,\enskip \varphi\in\cal{B}_b(S_1\times S_2).
\end{equation}
\end{lemma}
\begin{proof}Fix any $N>0$ and $\varphi$ in $\cal{B}_b(S_1\times S_2)$ and note that
\begin{align}
\label{eq:decomp}\enot(\varphi)-\eot(\varphi)=&(\eno-\eo)\times(\ent-\et)(\varphi)\\
&+(\eno-\eo)\times\et(\varphi)+\eo\times(\ent-\et)(\varphi).\nonumber
\end{align}
Hence, Minkowski's inequality  implies that
\begin{align}
\Ebb{\mmag{\enot(\varphi)-\eot(\varphi)}^p}^{\frac{1}{p}}\leq& \Ebb{\mmag{(\eno-\eo)\times(\ent-\et)(\varphi)}^p}^{\frac{1}{p}}\label{eq:nfme78wahfe78ahfa}\\
&+\Ebb{\mmag{\eno(\et(\varphi))-\eo(\et(\varphi))}^p}^{\frac{1}{p}}\nonumber\\
&+\Ebb{\mmag{\ent(\eo(\varphi))-\et(\eo(\varphi))}^p}^{\frac{1}{p}}\nonumber.
\end{align}
Because $\eta_1$ and $\eta_2$ are probability measures, $\eo(\varphi)$ and $\et(\varphi)$ are  bounded above by $\norm{\varphi}$. For this reason, the $L^p$ inequalities in~\eqref{eq:etalp} imply that 
\begin{align}
\Ebb{\mmag{\eno(\et(\varphi))-\eo(\et(\varphi))}^p}^{\frac{1}{p}}&\leq \frac{C\norm{\eta_2(\varphi)}}{N^{1/2}}\leq \frac{C\norm{\varphi}}{N^{1/2}},\\
\Ebb{\mmag{\ent(\eo(\varphi))-\et(\eo(\varphi))}^p}^{\frac{1}{p}}&\leq\frac{C\norm{\eta_1(\varphi)}}{N^{1/2}}\leq \frac{C\norm{\varphi}}{N^{1/2}}.
\end{align}
To control the remaining term in~\eqref{eq:nfme78wahfe78ahfa}, let $\mathcal{F}_2$ denote the $\sigma$-algebra generated by the $\eta_2^N$s. For each $\omega$ in $\Omega$, $x_1\mapsto \int\varphi(x_1,x_2)\eta_2^N(\omega, dx_2)-\eta_2(\varphi)(x_1)$ is a bounded function on $S_1$. Hence, \eqref{eq:etalp} and the independence of $(\eta_1^N)_{N=1}^\infty$ and $(\eta_2^N)_{N=1}^\infty$ imply that
\begin{align}
&\Ebb{\mmag{(\eno-\eo)\times(\ent-\et)(\varphi)}^p}^{\frac{1}{p}}\label{eq:fme89ahaf78haeufafe}\\
&\qquad\qquad=\Ebb{\mmag{\eno(\ent(\varphi)-\et(\varphi))-\eo(\ent(\varphi)-\et(\varphi))}^p}^{\frac{1}{p}}\nonumber\\
&\qquad\qquad=\Ebb{\Ebb{\mmag{\eno(\ent(\varphi)-\et(\varphi))-\eo(\ent(\varphi)-\et(\varphi))}^p|\cal{F}_2}}^{\frac{1}{p}}\nonumber\\
&\qquad\qquad\leq\frac{C\Ebb{\norm{\ent(\varphi)-\et(\varphi)}^p}^{\frac{1}{p}}}{N^{1/2}}\leq \frac{C\Ebb{(\norm{\ent(\varphi)}+\norm{\et(\varphi)})^p}^{\frac{1}{p}}}{N^{1/2}}\nonumber\\
&\qquad\qquad\leq \frac{C\Ebb{(2\max\{\norm{\ent(\varphi)},\norm{\et(\varphi)}\})^p}^{\frac{1}{p}}}{N^{1/2}}\leq\frac{2C\norm{\varphi}}{N^{1/2}}.\nonumber
\end{align}
Putting (\ref{eq:nfme78wahfe78ahfa}--\ref{eq:fme89ahaf78haeufafe}) together, we obtain the $L^p$ inequality~\eqref{eq:etaprodlp} for the product.
\end{proof}
\begin{proof}[Proof of Lemma~\ref{lem:lpcoa}]
Because $(\pi^N_{u1})_{N=1}^\infty,\dots,(\pi^N_{uc_u})_{N=1}^\infty$ are independent sequences of probability measures by construction, so are $(\pi^N_{u[k]})_{N=1}^\infty$ and $(\pi^N_{u(k+1)})_{N=1}^\infty$ for all $k<c_u$, where $u[k]:=\{u1,\dots,uk\}$ denotes the set  containing the first $k$ children of $u$ and $\pi^N_{u[k]}:=\prod_{v\in u[k]}\pi^N_{v}$ the corresponding product of $\pi_v^N$s. Hence, starting from the premise and repeatedly applying  Lemma~\ref{lem:prodlp}, we obtain the $L^p$ inequality for $\pi^N_{u[c_u]}=\pi_{\cal{C}_u}^N$.
\end{proof}
Emulating the approach of~\cite[Lemma 4]{Crisan2002}~and~\cite[Lemma 1]{miguez2013convergence}, we find that the correction step also respects the inequality:
\begin{lemma}[Correction step]\label{lem:lpcor}If, in addition to Assumptions~\ref{ass:abscont}--\ref{ass:boundweights},  \eqref{eq:lpcoa} is satisfied for some $u$ in $\mathbb{T}^{\not \partial}$ and $p\geq1$, then
\begin{equation}\label{eq:lpcor}\Ebb{\mmag{\pi_{u_-}^{N}(\varphi)-\pi_{u_-}(\varphi)}^p}^{\frac{1}{p}}\leq\frac{C\norm{\varphi}}{N^{1/2}}\quad\forall N>0,\enskip \varphi\in \cal{B}_b(\bm{E}_{\cal{C}_u}).\end{equation}
\end{lemma}

\begin{proof}Fix any $p\geq1$, $N>0$, and $\varphi$ in $\cal{B}_b(\bm{E}_{\cal{C}_u})$. Recall the definitions in Section~\ref{sec:dacalg}:
\begin{align*}\gamma_{u_-}^N(d\bm{x}_{\cal{C}_u})=w_{u_-}(\bm{x}_{\cal{C}_u})\gamma_{\cal{C}_u}^N(d\bm{x}_{\cal{C}_u}),\quad\gamma_{\cal{C}_u}^N(d\bm{x}_{\cal{C}_u})&=\cal{Z}_{\cal{C}_u}^N\pi_{\cal{C}_u}^N(d\bm{x}_{\cal{C}_u}).
\end{align*}
It follows that 
$$\pi^N_{u_-}(\varphi)=\frac{\gamma_{u_-}^N(\varphi)}{\gamma_{u_-}^N(\bm{E}_{\cal{C}_u})}=\frac{\gamma^N_{\cal{C}_u}(w_{u_-}\varphi)}{\gamma^N_{\cal{C}_u}(w_{u_-})}=\frac{\pi^N_{\cal{C}_u}(w_{u_-}\varphi)}{\pi^N_{\cal{C}_u}(w_{u_-})}.$$
Because $w_{u_-}=d\gamma_{u_-}/d\gamma_{\cal{C}_u}$, we similarly have that 
$$\pi_{u_-}(\varphi)=\frac{\gamma_{u_-}(\varphi)}{\gamma_{u_-}(\bm{E}_{\cal{C}_u})}=\frac{\gamma_{\cal{C}_u}(w_{u_-}\varphi)}{\gamma_{\cal{C}_u}(w_{u_-})}=\frac{\pi_{\cal{C}_u}(w_{u_-}\varphi)}{\pi_{\cal{C}_u}(w_{u_-})}.$$
Hence,
\begin{align}\label{eq:mdw89a0hd7w8ahdw7nau}\mmag{\pi_{u_-}^{N}(\varphi)-\pi_{u_-}(\varphi)}\leq \mmag{\pi_{u_-}^{N}(\varphi)-\frac{\pi_{\cal{C}_u}^N(w_{u_-}\varphi)}{\pi_{\cal{C}_u}(w_{u_-})}}+\mmag{\frac{\pi_{\cal{C}_u}^N(w_{u_-}\varphi)}{\pi_{\cal{C}_u}(w_{u_-})}-\frac{\pi_{\cal{C}_u}(w_{u_-}\varphi)}{\pi_{\cal{C}_u}(w_{u_-})}}.
 \end{align}
To control the first term on the right-hand side, we use
\begin{align}\label{eq:mdw89a0hd7w8ahdw7nau2}\mmag{\pi_{u_-}^{N}(\varphi)-\frac{\pi_{\cal{C}_u}^N(w_{u_-}\varphi)}{\pi_{\cal{C}_u}(w_{u_-})}}&\leq \frac{\mmag{\pi_{u_-}^N(\varphi)}\mmag{\pi_{\cal{C}_u}(w_{u_-})-\pi_{\cal{C}_u}^N(w_{u_-})}}{\pi_{\cal{C}_u}(w_{u_-})}\\
&\leq \frac{\norm{\varphi}\mmag{\pi_{\cal{C}_u}(w_{u_-})-\pi_{\cal{C}_u}^N(w_{u_-})}}{\pi_{\cal{C}_u}(w_{u_-})}.\nonumber\end{align}
Because, $\pi_{\cal{C}_u}(w_{u_-})=\cal{Z}_{\cal{C}_u}^{-1}\gamma_{\cal{C}_u}(w_{u_-})=\cal{Z}_{\cal{C}_u}^{-1}\gamma_{u_-}(\bm{E}_{\cal{C}_u})=\cal{Z}_{\cal{C}_u}^{-1}\cal{Z}_u$, the desired $L^p$~inequality~\eqref{eq:lpcor} follows from~(\ref{eq:lpcoa},\ref{eq:mdw89a0hd7w8ahdw7nau},\ref{eq:mdw89a0hd7w8ahdw7nau2}) and Minkowski's inequality: 
\begin{align*}&\Ebb{\mmag{\pi_{u_-}^{N}(\varphi)-\pi_{u_-}(\varphi)}^p}^{\frac{1}{p}}
\\
&\qquad\qquad\leq \frac{\norm{\varphi}\Ebb{\mmag{\pi_{\cal{C}_u}(w_{u_-})-\pi_{\cal{C}_u}^N(w_{u_-})}^p}^{\frac{1}{p}}+\Ebb{\mmag{\pi_{\cal{C}_u}(w_{u_-}\varphi)-\pi_{\cal{C}_u}^N(w_{u_-}\varphi)}^p}^{\frac{1}{p}}}{\pi_{\cal{C}_u}(w_{u_-})}\\
&\qquad\qquad\leq \frac{C\norm{\varphi}\norm{w_{u_-}}+C\norm{w_{u_-}\varphi}}{N^{1/2}\pi_{\cal{C}_u}(w_{u_-})}\leq \left(\frac{2C\cal{Z}_{\cal{C}_u}\norm{w_{u_-}}}{\cal{Z}_u}\right)\frac{\norm{\varphi}}{N^{1/2}}.\end{align*}
\end{proof}
To show that the resampling step also preserves $L^p$ inequalities, we tweak~\cite[Lemma~5]{Crisan2002} (recall that $\epsilon_u^N:=N^{-1}\sum_{n=1}^N\delta_{\bm{X}^{n,N}_{u_-}}$ denotes the resampled particles'  empirical distribution):
\begin{lemma}[Resampling step]\label{lem:lpres}If, in addition to Assumptions~\ref{ass:abscont}--\ref{ass:boundweights}, the inequality~\eqref{eq:lpcor} is satisfied for some $u$ in $\mathbb{T}^{\not \partial}$ and $p\geq1$, then
\begin{equation}\label{eq:lpres}
\Ebb{\mmag{\epsilon_u^N(\varphi)-\pi_{u_-}(\varphi)}^p}^{\frac{1}{p}}\leq\frac{C\norm{\varphi}}{N^{1/2}}\quad\forall N>0,\enskip \varphi\in \cal{B}_b(\bm{E}_{\cal{C}_u}).
\end{equation}
\end{lemma}
Key to proving the above is  the following conditional Marcinkiewicz-Zygmund inequality:
\begin{lemma}[{e.g.\ \cite[Theorem~3.3]{Yuan2015}}]\label{lem:733}Let $Y_1,\dots, Y_N$ denote bounded real-valued random variables on a probability space $(\Omega,\cal{F},\mathbb{P})$ that are independent when conditioned on some sigma algebra $\cal{G}\subseteq\cal{F}$ and satisfy $\Ebb{Y_n|\cal{G}}=0$ almost surely for all $n$ in $[N]$. For any given $p\geq1$, there exists a constant $C$  independent of  $Y_1,\dots, Y_N$ and $N$ such that
$$\Ebb{\left.\mmag{\sum_{n=1}^NY_n}^p\right|\cal{G}}\leq C\Ebb{\left.\left(\sum_{n=1}^NY_n^2\right)^{\frac{p}{2}}\right|\cal{G}}\quad\text{almost surely}.$$
\end{lemma}
Lemma~\ref{lem:lpres}'s proof is now straightforward:
\begin{proof}[Proof of Lemma~\ref{lem:lpres}] Fix any $p\geq1$, $N>0$, and $\varphi$ in $\cal{B}_b(\bm{E}_{\cal{C}_u})$. Conditioning on $\cal{F}_{\cal{C}_u}^N$ (the sigma-algebra generated by the mutated particles indexed by $u$'s children), we have that
$$\Ebb{\varphi(\bm{X}_{u_-}^{n,N})|\cal{F}_{\cal{C}_u}^N}=\pi^N_{u_-}(\varphi)\enskip\text{almost surely, for all }n\leq N.$$
Hence, applying Lemma~\ref{lem:733} with $\cal{G}:=\cal{F}_{\cal{C}_u}^N$ and $Y_n:=N^{-1}[\varphi(\bm{X}_{u_-}^{n,N})-\pi^N_{u_-}(\varphi)]$, the tower rule,  and the bound $Y_n^2\leq 4N^{-2}\norm{\varphi}^2$,  we obtain
\begin{align}\label{eq:ndsua8dnuw8andawuy}\Ebb{\mmag{\epsilon_u^N(\varphi)-\pi_{u_-}^N(\varphi)}^p}^{\frac{1}{p}}\leq\frac{C\norm{\varphi}}{N^{1/2}}.\end{align}
Inequality~\eqref{eq:lpres} then follows from the above,  Minkowski's inequality, and~\eqref{eq:lpcor}.
\end{proof}
For the mutation step, we adapt~\cite[Lemma~3]{Crisan2002}:
\begin{lemma}[Mutation step]\label{lem:lpmut}If, in addition to Assumptions~\ref{ass:abscont}--\ref{ass:boundweights},~\eqref{eq:lpres} is satisfied for some $u$ in $\mathbb{T}^{\not \partial}$ and $p\geq1$, then~\eqref{eq:lp0} is also satisfied for the same $u$ and $p$.
\end{lemma}
\begin{proof}Fix any $p\geq1$, $N>0$, and $\varphi$ in $\cal{B}_b(\bm{E}_{\cal{C}_u})$. Minkowski's inequality implies that
\begin{align}\Ebb{\mmag{\pi^N_u(\varphi)-\pi_u(\varphi)}^p}^{\frac{1}{p}}\leq&\Ebb{\mmag{\pi^N_u(\varphi)-\nu_u^N(\varphi)}^p}^{\frac{1}{p}}+\Ebb{\mmag{\nu_u^N(\varphi)-\pi_u(\varphi)}^p}^{\frac{1}{p}},\label{eq:nds7a8fnyaehbfwa2}\end{align}
where $\nu^N_u:=\epsilon_u^N\times K_u$. Because $\pi_u=\pi_{u_-}\times K_u$ and $\norm{K_u\varphi}\leq \norm{\varphi}$ as $K_u$ is a Markov kernel,~\eqref{eq:lpres} implies that
\begin{equation}\Ebb{\mmag{\nu_u^N(\varphi)-\pi_u(\varphi)}^p}^{\frac{1}{p}}\leq \frac{C\norm{K_u\varphi}}{N^{1/2}}\leq \frac{C\norm{\varphi}}{N^{1/2}}.\end{equation}
To control the other term in~\eqref{eq:nds7a8fnyaehbfwa2}, we condition on $\cal{F}_{u_-}^N$ (the   sigma-algebra generated by the resampled particles indexed by $u$) and obtain
$$\Ebb{\varphi(\bm{X}_u^{n,N})|\cal{F}_{u_-}^N}=(K_u\varphi)(\bm{X}_{u_-}^{n,N})\quad\forall n\leq N.$$
Hence, applying Lemma~\ref{lem:733} similarly as for~\eqref{eq:ndsua8dnuw8andawuy}, only this time with $\cal{G}:=\cal{F}_{u_-}^N$ and $Y_n:=N^{-1}[\varphi(\bm{X}_{u}^{n,N})-(K_u\varphi)(\bm{X}_{u_-}^{n,N})]$, we find that
\begin{align}\Ebb{\mmag{\pi^N_u(\varphi)-\nu_u^N(\varphi)}^p}^{\frac{1}{p}}\leq  \frac{C\norm{\varphi}}{N^{1/2}}\label{eq:nds7a8fnyaehbfwa3}.\end{align}
Combining~(\ref{eq:nds7a8fnyaehbfwa2}--\ref{eq:nds7a8fnyaehbfwa3}) completes the proof.
\end{proof}
\begin{proof}[Proof of Theorem~\ref{thrm:lpapp}] The theorem follows by repeatedly applying Lemmas~\ref{lem:lpcoa},~\ref{lem:lpcor},~\ref{lem:lpres},~and~\ref{lem:lpmut}, starting  from~\eqref{eq:lp0}.\end{proof}
\begin{proof}[Proof of Theorem~\ref{THRM:LP}]Suppose that we are able to argue that
\begin{align}
\Ebb{\mmag{\cal{Z}^N_{u}-\cal{Z}_{u}}^p}^{\frac{1}{p}}&\leq\frac{C}{N^{1/2}}\quad\forall N>0,\label{eq:lpz}
\end{align}
for all $u$ in $\mathbb{T}$ and $p\geq1$. Because, for all $\varphi$ in $\cal{B}_b(\bm{E}_u)$,
\begin{align*}
\mmag{\gamma^N_u(\varphi)-\gamma_u(\varphi)}&=\mmag{\cal{Z}_{u}^N\pi^N_u(\varphi)-\cal{Z}_{u}\pi_u(\varphi)}\\
&\leq \mmag{\pi^N_u(\varphi)}\mmag{\cal{Z}_{u}^N-\cal{Z}_{u}}+\cal{Z}_{u}\mmag{\pi^N_u(\varphi)-\pi_u(\varphi)}\\
&\leq \norm{\varphi}\mmag{\cal{Z}_{u}^N-\cal{Z}_{u}}+\cal{Z}_{u}\mmag{\pi^N_u(\varphi)-\pi_u(\varphi)}
\end{align*}
we would then obtain an $L^p$ inequality for $\gamma_u^N$ from that for  $\pi_u^N$ (Theorem~\ref{thrm:lpapp}). Given that 
$$\rho_u^N(d\bm{x}_u)=w_u(\bm{x}_u)\gamma_u^N(d\bm{x}_u),\quad Z^N_u=\gamma_u^N(w_u),\quad \mu_u^N(d\bm{x}_u)=\frac{w_u(\bm{x}_u)\gamma_u^N(d\bm{x}_u)}{Z^N_u},$$
and similarly for $\gamma_u$, $\rho_u$, $Z_u$, and $\mu_u$, the inequalities in Theorem~\ref{THRM:LP} would then follow using arguments of the type in Lemma~\ref{lem:lpcor}'s proof. 

Fix any $p\geq 1$ and $N>0$. In the case of a leaf $u$,~\eqref{eq:lpz} is trivially satisfied because $\cal{Z}_u^N=\cal{Z}_u=1$ by definition. 
Suppose, instead, that $u$ is not a leaf and that the $L^p$ inequality holds for each of its children (i.e.\ \eqref{eq:lpz} holds with $v$ replacing $u$ therein, for each $v$ in $\cal{C}_u$). 
Because $\cal{Z}_{u1}^N,\dots,\cal{Z}_{uc_u}^N$ are independent by definition, and using the following multinomial expansion
$$\cal{Z}_{\cal{C}_u}^N=\prod_{v\in \cal{C}_u}\cal{Z}_{v}^N=\prod_{v\in \cal{C}_u} [(\cal{Z}_{v}^N - \cal{Z}_{v}) + \cal{Z}_{v}] =\sum_{A\subseteq \cal{C}_u}\left(\prod_{v\in A}(\cal{Z}_{v}^N-\cal{Z}_{v})\right)\cal{Z}_{\cal{C}_u}^{\not A}$$
where $\cal{Z}_{\cal{C}_u}^{\not A}:=\cal{Z}_{\cal{C}_u\backslash A}$, Minkowski's inequality  implies that
\begin{align*}\Ebb{\mmag{\cal{Z}_{\cal{C}_u}^N-\cal{Z}_{\cal{C}_u}}^p}^\frac{1}{p}&\leq \sum_{\emptyset\neq A\subseteq\cal{C}_u}\Ebb{\mmag{\prod_{v\in A}(\cal{Z}_{v}^N-\cal{Z}_{v})}^p}^\frac{1}{p}\cal{Z}_{\cal{C}_u}^{\not A}\\
&=\sum_{\emptyset\neq A\subseteq\cal{C}_u}\cal{Z}_{\cal{C}_u}^{\not A}\prod_{v\in A}\Ebb{\mmag{\cal{Z}_{v}^N-\cal{Z}_{v}}^{p}}^\frac{1}{p}\leq  \frac{\sum_{\emptyset\neq A\subseteq\cal{C}_u}\cal{Z}_{\cal{C}_u}^{\not A}C}{N^{1/2}}.
\end{align*}
Given the above and
\begin{align*}
\mmag{\cal{Z}^N_u-\cal{Z}_u}&=\mmag{\cal{Z}_{\cal{C}_u}^N\pi_{\cal{C}_u}^N(w_{u_-})-\cal{Z}_{\cal{C}_u}\pi_{\cal{C}_u}(w_{u_-})}\\
&\leq \mmag{\pi_{\cal{C}_u}^N(w_{u_-})}\mmag{\cal{Z}_{\cal{C}_u}^N-\cal{Z}_{\cal{C}_u}}+\cal{Z}_{\cal{C}_u}\mmag{\pi_{\cal{C}_u}^N(w_{u_-})-\pi_{\cal{C}_u}(w_{u_-})}\\
&\leq \norm{w_{u_-}}\mmag{\cal{Z}_{\cal{C}_u}^N-\cal{Z}_{\cal{C}_u}}+\cal{Z}_{\cal{C}_u}\mmag{\pi_{\cal{C}_u}^N(w_{u_-})-\pi_{\cal{C}_u}(w_{u_-})},
\end{align*}
\eqref{eq:lpz} follows from the $L^p$ inequality for $\pi_{\cal{C}_u}^N$ in~\eqref{eq:lpcoa} (obtained in the proof of Theorem~\ref{thrm:lpapp}) and Minkowski's inequality. Hence, starting from the nodes whose children are leaves and working our way inductively up the tree, we obtain~\eqref{eq:lpz} for all $u$ in $\mathbb{T}$.
\end{proof}
\section{Proof of Theorem~\ref{THRM:UNBIAS}}\label{app:unbias}
Throughout this appendix, we use the notation described in Appendix~\ref{app:notation}. 
Because $\rho_u^N(\varphi)=\gamma_u^N(w_u\varphi)$ and $w_u$ is bounded (Assumption~\ref{ass:boundweights}), we need only show that $\gamma_u^N$ is unbiased. That is, for all $u$ in $\mathbb{T}$,
\begin{equation}\label{eq:fmwuanfwuafa}
\Ebb{\gamma^N_u(\varphi)}=\gamma_u(\varphi)\quad\forall N>0, \enskip \varphi\in\cal{B}_b(\bm{E}_u).
\end{equation}
If $u$ is a leaf, the above holds trivially. For any other $u$, note that $\bm{X}_u^{n,N}$, by definition, has law $\pi_{u_-}^N\times K_u$ when conditioned on $\cal{F}_{\cal{C}_u}^N$. Because $\cal{Z}_u^N$ is $\cal{F}_{\cal{C}_u}^N$-measurable, it follows that
\begin{align*}
\Ebb{\gamma_u^N(\varphi)}&=\frac{1}{N}\sum_{n=1}^N\Ebb{\cal{Z}_u^N\varphi(\bm{X}_u^{n,N})}=\frac{1}{N}\sum_{n=1}^N\Ebb{\Ebb{\cal{Z}_u^N\varphi(\bm{X}_u^{n,N})|\cal{F}_{\cal{C}_u}^N}}\\
&=\frac{1}{N}\sum_{n=1}^N\Ebb{\cal{Z}_u^N\Ebb{\varphi(\bm{X}_u^{n,N})|\cal{F}_{\cal{C}_u}^N}}=\frac{1}{N}\sum_{n=1}^N\Ebb{\cal{Z}_u^N\pi_{u_-}^N(K_u\varphi)}=\Ebb{\gamma_{u_-}^N(K_u\varphi)}
\end{align*}
for all $N>0$ and $\varphi$ and $\cal{B}_b(\bm{E}_u)$. For this reason, and because $\gamma_u=\gamma_{u_-}\times K_u$, we need only show that $\gamma_{u_-}^N(\varphi)$ is unbiased to argue~\eqref{eq:fmwuanfwuafa}. We achieve this by extending  the arguments used in~\cite[Section~16.4.1]{chopin2020} to establish the analogous result for SMC:
\begin{theorem}For any $u$ in $\mathbb{T}^{\not \partial}$, 
\begin{equation}\label{eq:fmwuanfwuafa2}
\Ebb{\gamma^N_{u_-}(\varphi)}=\gamma_{u_-}(\varphi)\quad\forall N>0, \enskip \varphi\in\cal{B}_b(\bm{E}_{\cal{C}_u}).
\end{equation}
\end{theorem}

\begin{proof}We argue the result inductively, starting from the nodes whose children are leaves and recursively moving our way  up the tree.   
For the base case ($u$'s children are leaves), note that the independence of $(\bm{X}_{u1}^{n,N})_{n=1}^N,\dots,(\bm{X}_{uc_u}^{n,N})_{n=1}^N$ and the fact that $\bm{X}_{v}^{1,N},\dots,\bm{X}_{v}^{N,N}\sim K_v$ for each $v$ in $\cal{C}_u$ imply that
\begin{align*}
\Ebb{\gamma^N_{u_-}(\varphi)}&=\frac{1}{N^{c_u}}\sum_{\bm{n}\in [N]^{c_u}}\Ebb{w_{u_-}(\bm{X}^{\bm{n},N}_{\cal{C}_u})\varphi(\bm{X}^{\bm{n},N}_{\cal{C}_u})}=\frac{1}{N^{c_u}}\sum_{\bm{n}\in [N]^{c_u}}\left(\prod_{v\in\cal{C}_u}K_v\right)(w_{u_-}\varphi)\\
&=\left(\prod_{v\in\cal{C}_u}K_v\right)(w_{u_-}\varphi)=\gamma_{\cal{C}_u}(w_{u_-}\varphi)=\gamma_{u_-}(\varphi)\quad\forall N>0, \enskip \varphi\in\cal{B}_b(\bm{E}_{\cal{C}_u}).
\end{align*}

For the inductive step, suppose instead that~\eqref{eq:fmwuanfwuafa2} holds for each of $u$'s non-leaf children:
\begin{equation}
\label{eq:fmwuanfwuafa3}
\Ebb{\gamma^N_{v_-}(\varphi)}=\gamma_{v_-}(\varphi),\quad\forall N>0, \enskip \varphi\in\cal{B}_b(\bm{E}_{\cal{C}_v}),\enskip v\in\cal{C}_u^{\not\partial}.
\end{equation}
We need to show that~\eqref{eq:fmwuanfwuafa2} also holds for $u$ itself. To do so, note that $(\bm{X}_{u1}^{n,N})_{n=1}^N,\dots,(\bm{X}_{uc_u}^{n,N})_{n=1}^N$ are independent and that, for each $v$ in $\cal{C}_u^{\not \partial}$, $\bm{X}_{v}^{1,N},\dots,\bm{X}_{v}^{N,N}$ have law $\pi_{v_-}^N\times K_v$ conditioned on $\cal{F}_{\cal{C}_v}^N$. Hence,
\begin{align*}
\Ebb{\gamma^N_{u_-}(\varphi)}&=\frac{1}{N^{c_u}}\sum_{\bm{n}\in [N]^{c_u}}\Ebb{\cal{Z}_{\cal{C}_u}^Nw_{u_-}(\bm{X}^{\bm{n},N}_{\cal{C}_u})\varphi(X^{\bm{n},N}_{\cal{C}_u})}\\
&=\frac{1}{N^{c_u}}\sum_{\bm{n}\in [N]^{c_u}}\Ebb{\Ebb{\left(\prod_{v\in\cal{C}_u^{\not\partial}}\cal{Z}_v^N\right)w_{u_-}(\bm{X}^{\bm{n},N}_{\cal{C}_u})\varphi(X^{\bm{n},N}_{\cal{C}_u})\left|\bigvee_{v\in\cal{C}_u^{\not\partial}}\cal{F}_{\cal{C}_v}^N\right.}}\\
&=\frac{1}{N^{c_u}}\sum_{\bm{n}\in [N]^{c_u}}\Ebb{\prod_{v\in\cal{C}_u^{\not\partial}}\cal{Z}_v^N\Ebb{w_{u_-}(\bm{X}^{\bm{n},N}_{\cal{C}_u})\varphi(X^{\bm{n},N}_{\cal{C}_u})\left|\bigvee_{v\in\cal{C}_u^{\not\partial}}\cal{F}_{\cal{C}_v}^N\right.}}\\
&=\frac{1}{N^{c_u}}\sum_{\bm{n}\in [N]^{c_u}}\Ebb{\prod_{v\in\cal{C}_u^{\not\partial}}\cal{Z}_v^N\left(\prod_{v\in\cal{C}_u^{\not\partial}}\pi_{v_-}^N\times K_v\right)\left(\left(\prod_{v\in\cal{C}_u^\partial}K_v\right)(w_{u_-}\varphi)\right)}\\
&=\Ebb{\left(\prod_{v\in\cal{C}_u^{\not\partial}}\gamma_{v_-}^N\right)(K_{\cal{C}_u}[w_{u_-}\varphi])}=\Ebb{\left(\prod_{v\in\cal{C}_u^{\not\partial}}\gamma_{v_-}\right)(K_{\cal{C}_u}[w_{u_-}\varphi])}\\
&=\gamma_{\cal{C}_u}(w_{u_-}\varphi)=\gamma_{u_-}(\varphi),
\end{align*}
where the third-to-last equality follows from applying~\eqref{eq:fmwuanfwuafa3} and the tower rule, and $K_{\cal{C}_u}$ denotes the kernel defined by
\begin{equation}\label{eq:Kcu}K_{\cal{C}_u}(\bm{x}_{\cal{C}_{u_-}},dx_{\cal{C}_u}):=\left(\prod_{v\in\cal{C}_u^{\partial}}K_v(dx_v)\right)\left(\prod_{v\in\cal{C}_u^{\not \partial}}K_v(\bm{x}_{\cal{C}_v},dx_v)\right)\quad\forall \bm{x}_{\cal{C}_{u_-}}\in \prod_{v\in\cal{C}_u^{\not \partial}}\bm{E}_{\cal{C}_v}.\end{equation}
\end{proof}
\section{Proof of Theorem~\ref{THRM:BIAS}}\label{app:bias}
Throughout this appendix, we use the notation described in Appendix~\ref{app:notation}. 
By definition, 
$$\mu_u^N(\varphi)=\frac{\rho_u^N(\varphi)}{\rho_u^N(\bm{E}_u)}=\frac{\gamma_u^N(w_u\varphi)}{\gamma_u^N(w_u)}=\frac{\pi_u^N(w_u\varphi)}{\pi_u^N(w_u)},\quad\mu_u(\varphi)=\frac{\rho_u(\varphi)}{\rho_u(\bm{E}_u)}=\frac{\gamma_u(w_u\varphi)}{\gamma_u(w_u)}=\frac{\pi_u(w_u\varphi)}{\pi_u(w_u)},$$
for all $\varphi$ in $\cal{B}_b(\bm{E}_u)$ and $u$ in $\mathbb{T}$. For this reason, and because $w_u$ is bounded above and below by theorem's premise, the bias bounds for $(\mu_u^N)_{u\in\mathbb{T}}$ in Theorem~\ref{THRM:BIAS} follow using standard arguments (see~\cite[p.~35]{liu2001monte} or (\ref{eq:bias_decomp_dac}--\ref{eq:pimbias}) below) from those for $(\pi_u^N)_{u\in\mathbb{T}}$ in Theorem~\ref{THRM:BIASapp} below.
\begin{theorem}[Bias estimates for   $(\pi_u^N)_{u\in\mathbb{T}}$]\label{THRM:BIASapp} If Assumptions~\ref{ass:abscont}--\ref{ass:boundweights} hold, the weight functions are bounded below (i.e.\, for every $u$ in $\mathbb{T}^{\not\partial}$, $w_{u_-}\geq \beta_u$  for some constant $\beta_u>0$), then for all  $u$ in $\mathbb{T}$,
\begin{equation}\label{eq:pibiasbound}
\mmag{\Ebb{\pi^N_u(\varphi)}-\pi_u(\varphi)}\leq\frac{C\norm{\varphi}}{N}\quad\forall N>0,\enskip \varphi\in\cal{B}_b(\bm{E}_{u}).
\end{equation}
\end{theorem}
To argue the above, we follow an approach similar to the one we took in Appendix~\ref{app:lplln} for the $L^p$ inequalities: we show that the bias bound holds for the leaves, prove that each step of the algorithm preserves the bound, and inductively work our way up the tree. To this end, we require Lemma~\ref{lem:prodbias}, our main innovation in this appendix, showing that the product operation respects bias bounds.
\begin{lemma}
\label{lem:prodbias}Suppose that $(\eta_1^N)_{N=1}^\infty$ and $(\eta_2^N)_{N=1}^\infty$ are independent sequences of random probability measures defined on some common probability triplet $(\Omega,\cal{F},\mathbb{P})$, respectively taking values in  some measurable spaces $(S_1,\cal{S}_1)$ and $(S_2,\cal{S}_2)$, and satisfying bias estimates:
\begin{equation}\label{eq:etabias}\mmag{\Ebb{\eta^N_k(\varphi_{k})}-\eta_{k}(\varphi_{k})}\leq\frac{C\norm{\varphi_{k}}}{N}\quad\forall N>0,\enskip \varphi_{k}\in \cal{B}_b(S_{k}),\enskip k\in\{1, 2\},\end{equation}
with limits $\eta_1$ and $\eta_2$ that are also probability measures. The sequence of products $(\eta^N_1\times\eta^N_2)_{N=1}^\infty$ also satisfies a bias estimate: 
\begin{equation}\label{eq:etaprodbias}
\mmag{\Ebb{(\eta^N_1\times\eta^N_2)(\varphi)}-(\eta_1\times\eta_2)(\varphi)}\leq\frac{C\norm{\varphi}}{N}\quad\forall N>0,\enskip \varphi\in \cal{B}_b(S_1\times S_2).
\end{equation}
\end{lemma}
\begin{proof}[Proof of Lemma~\ref{lem:prodbias}] Fix any $\varphi$ in $\cal{B}_b(S_1\times S_2)$ and $N>0$, and recall the decomposition in~\eqref{eq:decomp}. Because $\eta_1,\eta_2$ are probability measures, $\norm{\eta_1(\varphi)}$ and $\norm{\eta_2(\varphi)}$ are both bounded above by $\norm{\varphi}$ and we can apply the bias estimates in~\eqref{eq:etabias} to control the expected value of the rightmost two terms in~\eqref{eq:decomp}:
\begin{align}
\mmag{\Ebb{\eno(\et(\varphi))}-\eo(\et(\varphi))}&\leq \frac{C\norm{\eta_2(\varphi)}}{N}\leq \frac{C\norm{\varphi}}{N},\label{eq:bias_one}\\
\mmag{\Ebb{\ent(\eo(\varphi))}-\et(\eo(\varphi))}&\leq\frac{C\norm{\eta_1(\varphi)}}{N}\leq \frac{C\norm{\varphi}}{N}.\label{eq:bias_two}
\end{align}
To control the expected value of the remaining term in the right-hand side of~\eqref{eq:decomp}, let $\mathcal{F}_2$ denote the $\sigma$-algebra generated by the $\eta_2^N$s. Because, for each $\omega$ in $\Omega$,  $x_1\mapsto \int\varphi(x_1,x_2)\eta_2^N(\omega, dx_2)-\eta_2(\varphi)(x_1)$ is a bounded function on $S_1$, \eqref{eq:etabias} and the independence of $(\eta_1^N)_{N=1}^\infty$ and $(\eta_2^N)_{N=1}^\infty$ imply that
\begin{align*}
&\mmag{\Ebb{(\eno-\eo)\times(\ent-\et)(\varphi)}}\\
&\qquad\qquad=\mmag{\Ebb{\eno(\ent(\varphi)-\et(\varphi))-\eo(\ent(\varphi)-\et(\varphi))}}\\
&\qquad\qquad=\mmag{\Ebb{\Ebb{\eno(\ent(\varphi)-\et(\varphi))-\eo(\ent(\varphi)-\et(\varphi))|\cal{F}_2}}}\nonumber\\
&\qquad\qquad\leq\Ebb{\mmag{\Ebb{\eno(\ent(\varphi)-\et(\varphi))-\eo(\ent(\varphi)-\et(\varphi))|\cal{F}_2}}}\nonumber\\
&\qquad\qquad\leq\frac{C\Ebb{\norm{\ent(\varphi)-\et(\varphi)}}}{N}\leq\frac{2C\norm{\varphi}}{N}.
\end{align*}
Putting the above and~(\ref{eq:bias_one},\ref{eq:bias_two}) together, we obtain the bias estimate for the product in~\eqref{eq:etaprodbias}.
\end{proof}
Armed with Lemma~\ref{lem:prodbias}, the remainder of Theorem~\ref{THRM:BIASapp}'s proof follows the approach of~\cite{olsson2004bootstrap}:
\begin{proof}[Proof of Theorem~\ref{THRM:BIASapp}]In the case of a leaf $u$, the bound~\eqref{eq:pibiasbound} holds trivially:
\begin{align}
\label{eq:bias_base}
\Ebb{\pi^N_u(\varphi)}=\frac{1}{N}\sum_{n=1}^N\Ebb{\varphi(\bm{X}_u^{n,N})}=K_u(\varphi)=\pi_u(\varphi)\quad\forall \varphi\in\cal{B}_b(\bm{E}_u),\end{align}
because $\cal{Z}_u^N=1$ by definition and $\bm{X}_u^{1,N},\dots,\bm{X}_u^{N,N}$ are drawn directly from $K_u$.

For all other nodes, we argue the result inductively, starting from those whose children are leaves and recursively moving our way  up the tree.   
For the base case ($u$'s children are leaves), note that $(\pi^N_{u1})_{N=1}^\infty,\dots,(\pi^N_{uc_u})_{N=1}^\infty$ are independent sequences of probability measures by construction;  consequently, for any $k<c_u$, $(\pi^N_{u[k]})_{N=1}^\infty$ and $(\pi^N_{u(k+1)})_{N=1}^\infty$  are also independent sequences (where $u[k]:=\{u1,\dots,uk\}$ and $\pi^N_{u[k]}:=\prod_{v\in u[k]}\pi^N_{v}$). For this reason, and because the bound~\eqref{eq:pibiasbound} holds for all of $u$'s children, repeated applications of Lemma~\ref{lem:prodbias}  yield 
\begin{align}
\label{eq:bias_lambda}
\mmag{\Ebb{\pi_{\cal{C}_u}^N(\varphi)}-\pi_{\cal{C}_u}(\varphi)} \leq \frac{C\norm{\varphi}}{N}\quad \forall N>0, \enskip \varphi\in\cal{B}_b(\bm{E}_{\cal{C}_u}).
\end{align}

To proceed, fix any $\varphi$ in $\cal{B}_b(\bm{E}_{\cal{C}_u})$ and $N>0$. As shown at the start of Lemma~\ref{lem:lpcor}'s proof,
$$\pi_{u_-}(\varphi)=\frac{\pi_{\cal{C}_u}(w_{u_-}\varphi)}{\pi_{\cal{C}_u}(w_{u_-})},\quad \pi_{u_-}^N(\varphi)=\frac{\pi_{\cal{C}_u}^N(w_{u_-}\varphi)}{\pi_{\cal{C}_u}^N(w_{u_-})}.$$
Hence, applying the triangle inequality, we find that
\begin{align}
\label{eq:bias_decomp_dac}
\mmag{\Ebb{\pi_{u_-}^N(\varphi)}-\pi_{u_-}(\varphi)} \leq& \mmag{\Ebb{\frac{\pi_{\cal{C}_u}^N(w_{u_-}\varphi)}{\pi_{\cal{C}_u}^N(w_{u_-})}}- \frac{\Ebb{\pi_{\cal{C}_u}^N(w_{u_-}\varphi)}}{\Ebb{\pi_{\cal{C}_u}^N(w_{u_-})}}} \\
&+ \mmag{\frac{\Ebb{\pi_{\cal{C}_u}^N(w_{u_-}\varphi)}}{\Ebb{\pi_{\cal{C}_u}^N(w_{u_-})}}-\frac{\pi_{\cal{C}_u}(w_{u_-}\varphi)}{\pi_{\cal{C}_u}(w_{u_-})}}.\nonumber
\end{align}
To control the first term, we apply a bivariate Taylor expansion with Lagrange remainder of $(y_1,y_2)\mapsto y_1/y_2$ at $(\Ebb{\pi_{\cal{C}_u}^N(w_{u_-}\varphi)},\Ebb{\pi_{\cal{C}_u}^N(w_{u_-})})$ and obtain
\begin{align}
\label{eq:biaslagrange}
\frac{\pi_{\cal{C}_u}^N(w_{u_-}\varphi)}{\pi_{\cal{C}_u}^N(w_{u_-})} =&\frac{\Ebb{\pi_{\cal{C}_u}^N(w_{u_-}\varphi)}}{\Ebb{\pi_{\cal{C}_u}^N(w_{u_-})}} + \frac{\pi_{\cal{C}_u}^N(w_{u_-}\varphi) - \Ebb{\pi_{\cal{C}_u}^N(w_{u_-}\varphi)}}{\Ebb{\pi_{\cal{C}_u}^N(w_{u_-})}} \\
&-\frac{\Ebb{\pi_{\cal{C}_u}^N(w_{u_-}\varphi)}}{\Ebb{\pi_{\cal{C}_u}^N(w_{u_-})}^2}\left( \pi_{\cal{C}_u}^N(w_{u_-}) - \Ebb{\pi_{\cal{C}_u}^N(w_{u_-})}\right) + R^N,\notag
\end{align}
where
\begin{align*}
R^N :=& \frac{\theta_1}{\theta_2^3}\left( \pi_{\cal{C}_u}^N(w_{u_-}) - \Ebb{\pi_{\cal{C}_u}^N(w_{u_-})}\right)^2 \\
&- \frac{1}{\theta_2^2}\left( \pi_{\cal{C}_u}^N(w_{u_-}) - \Ebb{\pi_{\cal{C}_u}^N(w_{u_-})}\right) \left( \pi_{\cal{C}_u}^N(w_{u_-}\varphi) - \Ebb{\pi_{\cal{C}_u}^N(w_{u_-}\varphi)}\right),
\end{align*}
for some random point $(\theta_1, \theta_2)$ on the line segment joining $(\pi_{\cal{C}_u}^N(w_{u_-}\varphi), \pi_{\cal{C}_u}^N(w_{u_-}))$ and $\left(\Ebb{\pi_{\cal{C}_u}^N(w_{u_-}\varphi)}, \Ebb{\pi_{\cal{C}_u}^N(w_{u_-})}\right)$. 
Taking expectations of~\eqref{eq:biaslagrange}, we find that
\begin{align}\label{eq:fme8uwafneyau8feya}
&\mmag{\Ebb{\frac{\pi_{\cal{C}_u}^N(w_{u_-}\varphi)}{\pi_{\cal{C}_u}^N(w_{u_-})}}- \frac{\Ebb{\pi_{\cal{C}_u}^N(w_{u_-}\varphi)}}{\Ebb{\pi_{\cal{C}_u}^N(w_{u_-})}}} \leq \Ebb{\mmag{\frac{\theta_1}{\theta_2^3}\left( \pi_{\cal{C}_u}^N(w_{u_-}) - \Ebb{\pi_{\cal{C}_u}^N(w_{u_-})}\right)^2 }}\\
&\qquad+\Ebb{\mmag{ \frac{1}{\theta_2^2}\left( \pi_{\cal{C}_u}^N(w_{u_-}) - \Ebb{\pi_{\cal{C}_u}^N(w_{u_-})}\right) \left( \pi_{\cal{C}_u}^N(w_{u_-}\varphi) - \Ebb{\pi_{\cal{C}_u}^N(w_{u_-}\varphi)}\right)}}.\nonumber
\end{align}
As we will show in Lemma~\ref{lemma:2.2olsson} below, for all non-negative integers satisfying $l+k\geq 1$,
\begin{align}\label{eq:olsons22}
&\Ebb{\mmag{\pi_{\cal{C}_u}^M(\psi_1)-\Ebb{\pi_{\cal{C}_u}^M(\psi_1)}}^l\mmag{\pi_{\cal{C}_u}^M(\psi_2)-\Ebb{\pi_{\cal{C}_u}^M(\psi_2)}}^k}\leq\frac{C\norm{\psi_1}^l\norm{\psi_2}^k}{M^{(l+k)/2}}
\end{align}
for all $M>0$ and $\psi_1,\psi_2$ in $\cal{B}_b(\bm{E}_{\cal{C}_u})$. Given our assumption that $w_{u_-}$ is bounded below by some constant $\beta>0$, we have that
\begin{align*}
\theta_2\geq\min\{\pi_{\cal{C}_u}^N(w_{u_-}),\Ebb{\pi_{\cal{C}_u}^N(w_{u_-})}\}\geq\min\{\beta\pi_{\cal{C}_u}^N(\bm{E}_{\cal{C}_u}),\Ebb{\beta\pi_{\cal{C}_u}^N(\bm{E}_{\cal{C}_u})}\}= \beta.
\end{align*}
Similarly,   $\theta_1\leq \norm{w_{u_-}}\norm{\varphi}$ and it follows from (\ref{eq:fme8uwafneyau8feya},\ref{eq:olsons22}) that
\begin{align}
\label{eq:bias1}
\mmag{\Ebb{\frac{\pi_{\cal{C}_u}^N(w_{u_-}\varphi)}{\pi_{\cal{C}_u}^N(w_{u_-})}}- \frac{\Ebb{\pi_{\cal{C}_u}^N(w_{u_-}\varphi)}}{\Ebb{\pi_{\cal{C}_u}^N(w_{u_-})}}} &\leq \left(\frac{\norm{w_{u_-}}^3\norm{\varphi}}{\beta^3}+\frac{\norm{w_{u_-}}^2\norm{\varphi}}{\beta^2}\right)\frac{C}{N}.
\end{align}
To deal with the second term in~\eqref{eq:bias_decomp_dac}, we apply the triangle inequality:
\begin{align*}
&\mmag{\frac{\Ebb{\pi_{\cal{C}_u}^N(w_{u_-}\varphi)}}{\Ebb{\pi_{\cal{C}_u}^N(w_{u_-})}}-\frac{\pi_{\cal{C}_u}(w_{u_-}\varphi)}{\pi_{\cal{C}_u}(w_{u_-})}}\\
&\leq \frac{\mmag{\Ebb{\pi_{\cal{C}_u}^N(w_{u_-}\varphi)}-\pi_{\cal{C}_u}(w_{u_-}\varphi)}}{\Ebb{\pi_{\cal{C}_u}^N(w_{u_-})}}+\frac{\mmag{\pi_{\cal{C}_u}(w_{u_-}\varphi)}\mmag{\pi_{\cal{C}_u}(w_{u_-})-\Ebb{\pi_{\cal{C}_u}^N(w_{u_-})}}}{\pi_{\cal{C}_u}(w_{u_-})\Ebb{\pi_{\cal{C}_u}^N(w_{u_-})}} \notag\\
&\leq \frac{\mmag{\Ebb{\pi_{\cal{C}_u}^N(w_{u_-}\varphi)}-\pi_{\cal{C}_u}(w_{u_-}\varphi)}}{\beta}+\frac{\norm{\varphi}}{\beta}\mmag{\pi_{\cal{C}_u}(w_{u_-})-\Ebb{\pi_{\cal{C}_u}^N(w_{u_-})}},
\end{align*}
where, in the final inequality, we used $\mmag{\pi_{\cal{C}_u}(w_{u_-}\varphi)/\pi_{\cal{C}_u}(w_{u_-})}=\mmag{\pi_{u_-}(\varphi)}\leq \norm{\varphi}$. Applying the bias estimate for $\pi_{\cal{C}_u}^N$ in~\eqref{eq:bias_lambda}, we then find that
$$
\mmag{\frac{\Ebb{\pi_{\cal{C}_u}^N(w_{u_-}\varphi)}}{\Ebb{\pi_{\cal{C}_u}^N(w_{u_-})}}-\frac{\pi_{\cal{C}_u}(w_{u_-}\varphi)}{\pi_{\cal{C}_u}(w_{u_-})}}\leq\frac{2\norm{w_{u_-}}\norm{\varphi}C}{\beta N}.
$$
Combining the above with~(\ref{eq:bias_decomp_dac},\ref{eq:bias1}), we obtain a bias bound for $\pi_{u_-}^N$:
\begin{equation}\label{eq:pimbias}
\mmag{\Ebb{\pi^N_{u_-}(\varphi)}-\pi_{u_-}(\varphi)} \leq \frac{C\norm{\varphi}}{N}\quad \forall N>0, \enskip \varphi\in\cal{B}_b(\bm{E}_{\cal{C}_u}).
\end{equation}

Lastly, given that $\bm{X}_{u}^{1,N},\dots,\bm{X}_{u}^{N,N}$ have law $\pi_{u_-}^N\times K_v$ when conditioned on $\cal{F}_{\cal{C}_u}^N$,
\begin{align*}
\Ebb{\pi^N_u(\varphi)}&=\frac{1}{N}\sum_{n=1}^N\Ebb{\varphi(\bm{X}_u^{n,N})}=\frac{1}{N}\sum_{n=1}^N\Ebb{\Ebb{\varphi(\bm{X}_u^{n,N})|\cal{F}_{\cal{C}_u}^N}}=\Ebb{\pi_{u_-}^N(K_u\varphi)}
\end{align*}
for all $N>0$ and $\varphi$ in $\cal{B}_b(\bm{E}_u)$. Because $\pi_u=\pi_{u_-}\times K_u$, the desired bias estimate~\eqref{eq:pibiasbound} for $\pi_u^N$ follows from that for $\pi^N_{u_-}$ in~\eqref{eq:pimbias} for all non-leaf nodes $u$ whose children are leaves.

For the inductive step, take any non-leaf node $u$, assume that the bound~\eqref{eq:pibiasbound} holds for all of its non-leaf children, and repeat the exact same argument given above for the base case.
\end{proof}
We have one loose end left to tie up:
\begin{lemma}\label{lemma:2.2olsson}If Assumptions~\ref{ass:abscont}-\ref{ass:boundweights} are satisfied and $l,k$ are non-negative integers satisfying $l+k\geq 1$, then~\eqref{eq:olsons22} holds for all $M>0$ and $\psi_1,\psi_2$ in $\cal{B}_b(\bm{E}_{\cal{C}_u})$.
\end{lemma}
\begin{proof}Given the $L^p$ inequality for $\pi_{\cal{C}_u}^N$ in~\eqref{eq:lpcoa} (obtained in the proof of Theorem~\ref{thrm:lpapp}), this proof consists of applying Jensen's inequality and the Cauchy-Schwarz inequality as was done in the proof of \cite[Lemma 2.2]{olsson2004bootstrap}. 
\end{proof}
\section{Proof of Theorem~\ref{THRM:CLT}}\label{app:clt}
Throughout this appendix, we use the notation described in Appendix~\ref{app:notation}.  The aim of the appendix is to establish the CLTs in Theorem~\ref{THRM:CLT}. The meat of the matter entails deriving a CLT for the estimator $\gamma_{u}^N$ of the extended auxiliary measure $\gamma_{u}$:
\begin{theorem}[CLT for the unnormalized flow]\label{THRM:CLTun}\hspace{-2pt}If Assumptions~\ref{ass:abscont}--\ref{ass:boundweights} hold and $u$ lies in $\mathbb{T}$,
$$N^{1/2}\left(\gamma_{u}^N(\varphi)-\gamma_{u}(\varphi)\right)\Rightarrow\cal{N}(0,\sigma^2_{\gamma_u}(\varphi))\enskip\text{as}\enskip N\to\infty,\quad\forall\varphi\in \cal{B}_b(\bm{E}_{\cal{C}_u})$$
where $\sigma^2_{\gamma_u}(\varphi):=\sum_{v\in\mathbb{T}_u}\pi_v([\cal{Z}_{v}\Gamma_{v,u}\varphi-\gamma_u(\varphi)]^2)$.
\end{theorem}
Theorem~\ref{THRM:CLT} follows from Theorem~\ref{THRM:CLTun} using standard arguments:
\begin{proof}[Proof of Theorem~\ref{THRM:CLT}]Fix any $u$ in $\mathbb{T}$ and  $\varphi$ in $\cal{B}_b(\bm{E}_{u})$. Because $\rho_u^N(\varphi)=\gamma_u^N(w_u\varphi)$, the CLT for $\rho_u^N$ follows immediately from that for $\gamma_u^N$ and
\begin{align*}
\sigma^2_{\gamma_u}(w_u\varphi)=\sum_{v\in\mathbb{T}_u}\pi_v([\cal{Z}_{v}\Gamma_{v,u}\left[w_u\varphi\right]-\gamma_u(w_u\varphi)]^2)=\sum_{v\in\mathbb{T}_u}\pi_v([\cal{Z}_{v}\Gamma_{v,u}\left[w_u\varphi\right]-\rho_u(\varphi)]^2).
\end{align*}
For $\mu_u^N(\varphi)$ note that, just as in the standard SMC case \cite[Section 9.4.2]{DelMoral2004}, 
\begin{equation}\label{eq:fne8aungfyawengw}\mu_u^N(\varphi)-\mu_u(\varphi)=\frac{Z_{u}}{Z_{u}^N}\rho_u^N\left(\frac{\varphi-\mu_u(\varphi)}{Z_{u}}\right).\end{equation}
Theorem~\ref{THRM:LLN} tells us that  $Z_{u}^N$ converges almost surely to $Z_{u}$ as $N$ tends to infinity. Moreover,
$$\rho_u\left(\frac{\varphi-\mu_u(\varphi)}{Z_{u}}\right)=\mu_u(\varphi-\mu_u(\varphi))=0.$$
Hence, we obtain the CLT for $\mu_u^N$ applying Slutsky's theorem to~\eqref{eq:fne8aungfyawengw}:
$$N^{1/2}(\mu_u^N(\varphi)-\mu_u(\varphi))\Rightarrow\cal{N}(0,\sigma^2_{\rho_u}(Z_{u}^{-1}[\varphi-\mu_u(\varphi)]))\enskip\text{as}\enskip N\to\infty.$$

To complete the proof, note that
\begin{align*}
\sigma^2_{Z_u}=\sigma^2_{\rho_u}(1)=\sum_{v\in\mathbb{T}_u}\pi_v([\cal{Z}_{v}\Gamma_{v,u}w_u-Z_u]^2).
\end{align*}
But, for any $v$ in $\mathbb{T}_u$ and $\psi$ in $\cal{B}_b(\bm{E}_v)$,
\begin{align*}
\pi_v(\cal{Z}_{v}\Gamma_{v,u}w_u\psi)&=\gamma_v(\Gamma_{v,u}w_u\psi)=\gamma_u(w_u\psi)=\rho_u(\psi)=Z_u\mu_u(\psi)=Z_u\mu_u^v(\psi).
\end{align*}
In other words, $\cal{Z}_{v}\Gamma_{v,u}w_u=Z_ud\mu_u^v/d\pi_v$ and~\eqref{eq:sigZ} follows.
\end{proof}
For the proof of Theorem~\ref{THRM:CLTun}, we adapt the approach taken in \cite[Chapter~9]{DelMoral2004} to establish the analogous result for standard SMC. In particular, we first show  that the \emph{local errors},  $\pi^N_v-\pi_{v_-}^N\times K_v$, are $\cal{O}(N^{-1/2})$, asymptotically normal, and independent (Lemma~\ref{lem:Vs}). Next, we break down the \emph{global error} $\gamma_{u}^N-\gamma_{u}$ into a sum of the local errors and their products, both `propagated up the tree' via $(\Gamma_{v,u})_{v\in\mathbb{T}_u}$ (Lemma~\ref{lem:wdecomp}~and~\eqref{eq:gamlocaldecomp}). Using the $L^2$ product inequality in Lemma~\ref{THRM:CROSS}, we show that the products are  $o(N^{-1/2})$ and, hence, vanish in the rescaled limit. An application of the continuous mapping theorem then completes the proof. Let's begin: emulating~\cite[Theorem~10.6.1]{DelMoral2013}'s proof (or, similarly, those of \cite[Theorem~9.31; Corollary~9.3.1]{DelMoral2004}), we obtain the following:
\begin{lemma}\label{lem:Vs} Fix any $u$ in $\mathbb{T}$ and a collection $(\varphi_v)_{v\in\mathbb{T}_u}$ with $\varphi_v$ belonging to $\cal{B}_b(\bm{E}_v)$ for each $v$ in $\mathbb{T}_u$. Let 
\begin{equation}\label{eq:VvN}V^N_v(\varphi_v):=\left\{\begin{array}{lc}N^{1/2}(\pi^N_v(\varphi_v)-K_v(\varphi_v))&\text{if }v\text{ is a leaf}\\
N^{1/2}(\pi^N_v(\varphi_v)-\pi_{v_-}^N(K_v\varphi_v))&otherwise\end{array}\right.\quad\forall N>0,\enskip v\in\mathbb{T}_u.\end{equation}
Under Assumptions~\ref{ass:abscont}--\ref{ass:boundweights},  the collection $(V^N_v(\varphi_v))_{v\in\mathbb{T}_u}$ converges weakly  as $N\to\infty$ to a collection $(V_v(\varphi_v))_{v\in\mathbb{T}_u}$ of independent zero mean random variables with Gaussian laws:
\begin{equation}\label{eq:Vv}V_v(\varphi_v)\sim\cal{N}\left(0,\pi_v([\varphi_v-\pi_v(\varphi_v)]^2)\right)\quad\forall v\in\mathbb{T}_u.\end{equation}
\end{lemma}

\begin{proof}Because $(K_v)_{v\in\mathbb{T}_u}$ are Markov kernels,  Lemma~\ref{lem:borelcantelli} and the $L^p$ inequality for $\pi_{v_-}^N$ in~\eqref{eq:lpcor} (obtained in the proof of Theorem~\ref{thrm:lpapp}) imply that
$$\lim_{N\to\infty}\pi_{v_-}^N(K_v\psi)=\pi_{v_-}(K_v\psi)=\pi_v(\psi)\quad\text{almost surely},\quad\forall\psi\in\cal{B}_b(\bm{E}_v),\enskip v\in\mathbb{T}_u^{\not \partial}.$$
Armed with the above, this proof is entirely analogous to that of \cite[Theorem~10.6.1]{DelMoral2013}, we only need to tweak the definition of the $U^N_k(\varphi)$s therein.  
In order to do this, we define a bijection  $b=(b_1,b_2)$ from the one-dimensional index set $\{1,2,\dots,N\mmag{\mathbb{T}_u}\}$ to the two-dimensional index set $\mathbb{T}_u\times \{1,2,\dots,N\}$  that `preserves $\mathbb{T}_u$'s structure' in the sense that:
$$b^{-1}(v',n')\leq b^{-1}(v,n)\quad\forall v'\in\mathbb{T}_{v},\enskip n',n\leq N,\enskip v\in\mathbb{T}_u.$$
That is, $(\bm{X}_{b_1(s)}^{b_2(s),N})_{s=1}^{r}$ contains all particles sub-indexed by descendants of $b_1(r)$ and no particles sub-indexed by ancestors of $b_1(r)$. 
This mapping imposes an order on the contributions of the local errors, allowing us to introduce the following sequence of martingale increments (defined with respect to an appropriate filtration):
$$U_k^N(\varphi):=\frac{1}{N^{1/2}}\left[\varphi_{b_1(k)}(\bm{X}_{b_1(k)}^{b_2(k),N})-\pi^N_{b_1(k)}(K_{b_1(k)}\varphi_{b_1(k)}))\right];$$
and \cite[Theorem~10.6.1]{DelMoral2013}'s proof applies as is.
\end{proof}
To obtain the aforementioned local error decomposition for $\gamma_u^N-\gamma_u$,  it is easiest to first derive an analogous result for $\gamma_{\cal{C}_u}^N-\gamma_{\cal{C}_u}$. In particular, we express $\gamma_{\cal{C}_u}^N-\gamma_{\cal{C}_u}$ in terms of the local errors and an $\cal{O}(N^{-1})$ remainder term that is a linear combination of their products. When doing so, we find it convenient to introduce the kernel $\Upsilon_{v,u}:\bm{E}_{\cal{C}_v}\times\bm{\cal{E}}_{\cal{C}_u}\to[0,\infty)$ mapping $\gamma_{\cal{C}_v}$ to $\gamma_{\cal{C}_u}$ (i.e.\  $\gamma_{\cal{C}_v}(\Upsilon_{v,u}\varphi)=\gamma_{\cal{C}_u}(\varphi)$ for all $ \varphi$ in $\cal{B}_b(\bm{E}_{\cal{C}_u})$): $\Upsilon_{u,u}(\bm{x}_{\cal{C}_u},d\bm{y}_{\cal{C}_u})=\delta_{\bm{x}_{\cal{C}_u}}(d\bm{y}_{\cal{C}_u})$ and 
\begin{align}\label{eq:Uvu}\Upsilon_{r,u}(\bm{x}_{\cal{C}_r},d\bm{y}_{\cal{C}_u})=\delta_{\bm{x}_{\cal{C}_r}}(d\bm{y}_{\cal{C}_r})w_{v_-}(\bm{y}_{\cal{C}_r})K_v(\bm{y}_{\cal{C}_r},dy_r)\Gamma_{r,v}(\bm{y}_r,d\bm{y}_v)\gamma_{\cal{C}_u}^{\not v}(d\bm{x}_{\cal{C}_u}^{\not v})\end{align}
for all $r\neq u$ in $\mathbb{T}_u^{\not \partial}$, where $v$ denotes the child of $u$ that is an ancestor of $r$ (with $v=r$ if $r$ itself is a child of $u$) and $\Gamma_{r,v}$ is as in~(\ref{eq:Pivu}--\ref{eq:Pivu2}). It is straightforward to check that
\begin{align}
\Upsilon_{r,u}\varphi&=\Upsilon_{r,v}\Upsilon_{v,u}\varphi\quad\forall \varphi\in\cal{B}_b(\bm{E}_u),\enskip r\in\mathbb{T}_v^{\not\partial},\enskip v\in\mathbb{T}_u^{\not\partial},\enskip u\in\mathbb{T}^{\not\partial}, \label{eq:gamprop1}\\
\Gamma_{r,u}\varphi &=\gamma_{\cal{C}_v}^{\not r}(\Upsilon_{v,u}[w_{u_-}K_u\varphi])\quad\forall  \varphi\in\cal{B}_b(\bm{E}_u),\enskip r\in\cal{C}_v,\enskip v\in\mathbb{T}_u^{\not\partial},\enskip u\in\mathbb{T}^{\not \partial}, \label{eq:gamprop2}
\end{align}
two properties that will be of use in the following proofs.
\begin{lemma}\label{lem:wdecomp} If Assumptions~\ref{ass:abscont}--\ref{ass:boundweights} hold, $u$ lies in $\mathbb{T}^{\not \partial}$, and $\varphi$ in $\cal{B}_b(\bm{E}_{\cal{C}_u})$, then
\begin{align}\gamma_{\cal{C}_u}^N(\varphi)-\gamma_{\cal{C}_u}(\varphi)=\frac{1}{N^{1/2}}\sum_{v\in\mathbb{T}_u^{\not \partial}}\sum_{r\in\cal{C}_v}\cal{Z}_r^NV_r^N(\gamma_{\cal{C}_v}^{\not r}(\Upsilon_{v,u}\varphi))+R_u^N(\varphi),\label{eq:wun}\end{align}
where $(V_v^N(\Upsilon_{v,u}\varphi))_{v\in\mathbb{T}_u^{\not u}}$ is as in~\eqref{eq:VvN}  and  $R_u^N(\varphi)$ 
satisfies $N^{1/2}R_u^N(\varphi)\Rightarrow0$ as $N\to\infty$.
\end{lemma}
\begin{proof}
Fix any $u$ in $\mathbb{T}^{\not \partial}$ and $\varphi$ in $\cal{B}_b(\bm{E}_{\cal{C}_u})$. Note that, similarly to \eqref{eq:lem10},
\begin{equation}\label{eq:mfd89afnau8wfwa2}\gamma^N_{\cal{C}_u}-\gamma_{\cal{C}_u}=\prod_{v\in \cal{C}_u}[\gamma^N_{v}-\gamma_{v}+\gamma_{v}]-\gamma_{\cal{C}_u}=\sum_{\emptyset\neq A\subseteq\cal{C}_u}\Delta_{A}^N\times\gamma_{\cal{C}_u}^{\not A},\end{equation}
where $\Delta_{A}^N:=\prod_{v\in A}(\gamma_v^N-\gamma_v)$ and $\gamma_{\cal{C}_u}^{\not{A}}:=\gamma_{\cal{C}_u\backslash A}$ for all subsets $A$ of $\cal{C}_u$. 
The same arguments as those in~\eqref{eq:sameargsss}, only with the inequality in~\eqref{eq:crossprop} replaced by that in Theorem~\ref{THRM:CROSS}, show that 
\begin{align*}\Ebb{\Delta_{A}^N(\gamma_{\cal{C}_u}^{\not A}(\varphi))^2}^{\frac{1}{2}}\leq N^{-1}C\norm{\gamma_{\cal{C}_u}^{\not A}(\varphi)}\leq N^{-1}C\cal{Z}_{\cal{C}_u}^{\not A}\norm{\varphi}
\end{align*}
if $\mmag{A}>1$. Markov's inequality then implies that 
$$\Pbb{\left\{\mmag{\Delta_{A}^N(\gamma_{\cal{C}_u}^{\not A}(\varphi))}\geq\varepsilon\right\}}\leq \varepsilon^{-2}\Ebb{\Delta_{A}^N(\gamma_{\cal{C}_u}^{\not A}(\varphi))^2}\leq (\varepsilon N)^{-2}C\cal{Z}_{\cal{C}_u}^{\not A}\norm{\varphi}\quad\forall\varepsilon>0;$$
from which it follows that
\begin{equation}\label{eq:ddn0}N^{1/2}\Delta_{A}^N(\gamma_{\cal{C}_u}^{\not A}(\varphi))\Rightarrow0\enskip\text{as }N\to\infty\quad\forall A\subseteq\cal{C}_u:\mmag{A}>1.\end{equation}

We now focus on the case $\mmag{A} =1$ where $A$ is a singleton $\lbrace v\rbrace$. Comparing with~\eqref{eq:Uvu}, 
\begin{align*}
\gamma_v^N(\gamma_{\cal{C}_u}^{\not v}(\varphi))&=\cal{Z}_v^N\pi_v^N(\gamma_{\cal{C}_u}^{\not v}(\varphi))=N^{-1/2}\cal{Z}_v^NV_v^N(\gamma_{\cal{C}_u}^{\not v}(\varphi))+\cal{Z}_v^N\pi_{v_-}^N(K_v\gamma_{\cal{C}_u}^{\not v}(\varphi))\\
&=N^{-1/2}\cal{Z}_v^NV_v^N(\gamma_{\cal{C}_u}^{\not v}(\varphi))+\gamma_{v_-}^N(K_v\gamma_{\cal{C}_u}^{\not v}(\varphi))\\
&=N^{-1/2}\cal{Z}_v^NV_v^N(\gamma_{\cal{C}_u}^{\not v}(\varphi))+\gamma_{\cal{C}_v}^N(\Upsilon_{v,u}\varphi),\\
\gamma_v(\gamma_{\cal{C}_u}^{\not v}(\varphi))&=\gamma_{v_-}(K_v\gamma_{\cal{C}_u}^{\not v}(\varphi))=\gamma_{\cal{C}_v}(\Upsilon_{v,u}\varphi),
\end{align*}
for all non-leaf children $v$ of $u$. For these reasons,  
\begin{align}
(\Delta_{v}^N\times\gamma_{\cal{C}_u}^{\not v})(\varphi)&=\frac{\cal{Z}_v^NV_v^N(\gamma_{\cal{C}_u}^{\not v}(\varphi))}{N^{1/2}}\quad\forall v\in\cal{C}_u^\partial\label{eq:gnr7ea9gnr8yanguiar0}\\
\label{eq:gnr7ea9gnr8yanguiar}
(\Delta_{v}^N\times\gamma_{\cal{C}_u}^{\not v})(\varphi)&=\gamma_v^N(\gamma_{\cal{C}_u}^{\not v}(\varphi))-\gamma_v(\gamma_{\cal{C}_u}^{\not v}(\varphi))\\
&=\cal{Z}_v^NV_v^N(\gamma_{\cal{C}_u}^{\not v}(\varphi))N^{-1/2}+\gamma_{\cal{C}_v}^N(\Upsilon_{v,u}\varphi)-\gamma_{\cal{C}_v}(\Upsilon_{v,u}\varphi)\quad\forall v\in\cal{C}_u^{\not \partial}.\nonumber
\end{align}
If all of $u$'s children are leaves, then~\eqref{eq:wun} follows by setting $R_u^N(\varphi):=\sum_{A\subseteq \cal{C}_u:\mmag{A}>1}(\Delta_{A}^N\times\gamma_{\cal{C}_u}^{\not A})(\varphi)$ and combining (\ref{eq:mfd89afnau8wfwa2}--\ref{eq:gnr7ea9gnr8yanguiar0}). 
For all other $u$, we argue the claim inductively: suppose that, for each non-leaf child $v$ of $u$ and $\psi$ in $\cal{B}_b(\bm{E}_{\cal{C}_v})$,~\eqref{eq:wun} (with $\psi$ replacing $\varphi$) holds for some $R_v^N(\psi)$ satisfying $N^{1/2}R_v^N(\psi)\Rightarrow0$ as $N\to\infty$. Given~\eqref{eq:gamprop1}, we can re-write~\eqref{eq:gnr7ea9gnr8yanguiar} as 
\begin{align*}&N^{1/2}(\Delta_{v}^N\times\gamma_{\cal{C}_u}^{\not v})(\varphi)\\
&=\cal{Z}_v^NV_v^N(\gamma_{\cal{C}_u}^{\not v}(\varphi))+\sum_{v'\in\mathbb{T}_v^{\not \partial}}\sum_{r\in\cal{C}_{v'}}\cal{Z}_r^NV_r^N(\gamma_{\cal{C}_{v'}}^{\not r}(\Upsilon_{v',u}\varphi))+N^{1/2}R_v^N(\Upsilon_{v,u}\varphi)\quad\forall v\in\cal{C}_u^{\not\partial}.\nonumber
\end{align*}
Because $\mathbb{T}_u^{\not \partial}=\{u\}\bigcup_{v\in\cal{C}_u}\mathbb{T}_v^{\not \partial}$, we obtain~\eqref{eq:wun} by combining the above with (\ref{eq:mfd89afnau8wfwa2}--\ref{eq:gnr7ea9gnr8yanguiar0}) and setting
$$R_u^N(\varphi):=\sum_{v\in\cal{C}_u^{\not \partial}}R_v^N(\Upsilon_{v,u}\varphi)+\sum_{A\subseteq\cal{C}_u:\mmag{A}>1}(\Delta_{A}^N\times\gamma_{\cal{C}_u}^{\not A})(\varphi).$$ 
\end{proof}
Theorem~\ref{THRM:CLTun} follows from Lemmas~\ref{lem:Vs}--\ref{lem:wdecomp} and the continuous mapping theorem:
\begin{proof}[Proof of Theorem~\ref{THRM:CLTun}] If $u$ is a leaf, then the  CLT reduces to that for i.i.d. sequences. For any other $u$, fix  $\varphi$ in $\cal{B}_b(\bm{E}_{u})$ and note that, by definition,
\begin{align}\label{eq:gamlocaldecomp}\gamma_u^N(\varphi)-\gamma_u(\varphi)&=\gamma_u^N(\varphi)-\gamma_{u_-}^N(K_u\varphi)+\gamma_{u_-}^N(K_u\varphi)-\gamma_{u_-}(K_u\varphi)\\
&=N^{-1/2}\cal{Z}_u^NV_u(\varphi)+\gamma_{\cal{C}_u}^N(w_{u_-}K_u\varphi)-\gamma_{\cal{C}_u}(w_{u_-}K_u\varphi)\nonumber.\end{align}
Using Lemma~\ref{lem:wdecomp} and~\eqref{eq:gamprop2}, we then find that 
$$N^{1/2}(\gamma_u^N(\varphi)-\gamma_u(\varphi))=f((\cal{Z}_{v}^N)_{v\in\mathbb{T}_u},(V^N_v(\varphi_v))_{v\in\mathbb{T}_u},N^{1/2}R^N_u(w_{u_-}K_u\varphi))$$
where $\varphi_v$ denotes $\Gamma_{v,u}\varphi$ and
$$f(\bm{z},\bm{\nu},R):=\sum_{v\in\mathbb{T}_u}z_v\nu_v+R\quad\forall \bm{z},\bm{\nu}\in\r^{\mmag{\mathbb{T}_u}},\enskip R\in\r.$$
The $L^p$ inequalities for $(\cal{Z}_{v}^N)_{v\in\mathbb{T}_u}$ together with Lemma~\ref{lem:borelcantelli} imply that  $(\cal{Z}_{v}^N)_{v\in\mathbb{T}_u}$ converges almost surely to $(\cal{Z}_{v})_{v\in\mathbb{T}_u}$ as $N$ tends to infinity. Moreover, Lemmas~\ref{lem:Vs}--\ref{lem:wdecomp} show that $N^{1/2}R^N_u(w_{u_-}K_u\varphi)$ converges weakly to zero and $(V^N_v(\varphi_v))_{v\in\mathbb{T}_u}$ converges to the  family $(V_v(\varphi_v))_{v\in\mathbb{T}_u}$  of independent Gaussian random variables in~\eqref{eq:Vv}. For these reasons,  the continuous mapping theorem tells us that
$$N^{1/2}(\gamma_u^N(\varphi)-\gamma_u(\varphi))\Rightarrow\sum_{v\in\mathbb{T}_u}\cal{Z}_{v}V_v(\varphi_v)=\sum_{v\in\mathbb{T}_u}\cal{Z}_{v}V_v(\Gamma_{v,u}\varphi)\enskip\text{as}\enskip N\to\infty.$$
The right-hand side is a sum of zero mean independent Gaussian random variables and, so, also a zero mean Gaussian random with variance
\begin{align*}
    \text{Var}\left(\sum_{v\in\mathbb{T}_u}\cal{Z}_{v}V_v(\Gamma_{v,u}\varphi)\right)&=\sum_{v\in\mathbb{T}_u}\cal{Z}_{v}^2\text{Var}(V_v(\Gamma_{v,u}\varphi))=\sum_{v\in\mathbb{T}_u}\cal{Z}_{v}^2\pi_v([\Gamma_{v,u}\varphi-\pi_v(\Gamma_{v,u}\varphi)]^2)\\
    &=\sum_{v\in\mathbb{T}_u}\pi_v([\cal{Z}_{v}\Gamma_{v,u}\varphi-\gamma_v(\Gamma_{v,u}\varphi)]^2)\\
    &=\sum_{v\in\mathbb{T}_u}\pi_v([\cal{Z}_{v}\Gamma_{v,u}\varphi-\gamma_u(\varphi)]^2).
\end{align*}
\end{proof}
\section{Proof of the  Theorem~\ref{THRM:CROSS}}  \label{app:cross}
Throughout this appendix, we use the notation described in Appendix~\ref{app:notation}. The aim of the appendix is to prove the product $L^2$ inequality for DaC-SMC in Theorem~\ref{THRM:CROSS}. Similarly to the argument for the $L^p$ inequality in Appendix~\ref{app:lplln}, we prove the above by  obtaining a product $L^2$ inequality for the empirical distributions of the particles indexed by the tree's leaves and showing that each step of the algorithm preserves this inequality. To start, note that, if $u$ and $v$ are leaves, then
\cite[Lemma~1]{Kuntz2021}  implies that 
\begin{equation}\label{eq:crossprop}\Ebb{(\pi_u^N-\pi_u)\times(\pi_v^N-\pi_v)(\varphi)^2}^\frac{1}{2}\leq \frac{C\norm{\varphi}}{N}\quad\forall N>0,\enskip \varphi\in\cal{B}_b(\bm{E}_u\times \bm{E}_v).\end{equation}
Taking the products of $\pi_{u'}^N$ and $\pi_{v'}^N$  over all children $u'$ and $v'$ of nodes $u$ and $v$ to obtain $\pi_{\cal{C}_u}^N$ and $\pi_{\cal{C}_v}^N$  preserves~\eqref{eq:crossprop}:
\begin{lemma}[Product step]\label{lem:crosscoal}If, in addition to Assumptions~\ref{ass:abscont}--\ref{ass:boundweights}, \eqref{eq:crossprop} is satisfied for each child $u'$ and $v'$ (i.e.\ it holds with $u'$ and $v'$ replacing $u$ and $v$ therein) of two non-leaf nodes $u$ and $v$ lying on separate branches (i.e.\ $u\not\in\mathbb{T}_v$ and $v\not\in\mathbb{T}_u$), then 
\begin{equation}\label{eq:crosscoal}\Ebb{(\pi_{\cal{C}_u}^N-\pi_{\cal{C}_u})\times(\pi_{\cal{C}_v}^N-\pi_{\cal{C}_u})(\varphi)^2}^{\frac{1}{2}}\leq \frac{C\norm{\varphi}}{N}\quad\forall N>0,\enskip \varphi\in \cal{B}_b(\bm{E}_{\cal{C}_u}\times\bm{E}_{\cal{C}_v}).\end{equation}
\end{lemma}

\begin{proof}Fix any $N>0$ and $\varphi$ in $\cal{B}_b(\bm{E}_{\cal{C}_u}\times\bm{E}_{\cal{C}_v})$. Using a multinomial expansion as in the proof of Theorem~\ref{thrm:lpapp}, we have that
\begin{equation*}
\pi^N_{\cal{C}_u}=\prod_{r\in \cal{C}_u}[\pi^N_{r}-\pi_{r}+\pi_{r}]=\sum_{A\subseteq{\cal{C}_u}}\Delta_{A}^N\times\pi_{\cal{C}_u}^{\not{A}},\qquad\pi^N_{\cal{C}_v}=\dots =\sum_{A\subseteq{\cal{C}_v}}\Delta_{A}^N\times\pi_{\cal{C}_v}^{\not{A}},
\end{equation*}
where $\Delta_{A}^N:=\prod_{r\in A}(\pi_r^N-\pi_r)$ and  $\pi_{A}^{\not{B}}:=\pi_{A\backslash B}$ for all $B\subseteq A\subseteq \mathbb{T}$. Hence,
\begin{equation}
    \label{eq:lem10}
\pi_{\cal{C}_u}^N-\pi_{\cal{C}_u}=\sum_{\emptyset\neq A\subseteq\cal{C}_u}\Delta_{A}^N\times\pi_{\cal{C}_u}^{\not{A}},\qquad \pi_{\cal{C}_v}^N-\pi_{\cal{C}_u}=\sum_{\emptyset\neq A\subseteq\cal{C}_v}\Delta_{A}^N\times\pi_{\cal{C}_v}^{\not{A}},
\end{equation}
and it follows that
\begin{align*}
(\pi_{\cal{C}_u}^N-\pi_{\cal{C}_u})\times(\pi_{\cal{C}_v}^N-\pi_{\cal{C}_v})(\varphi)=\sum_{\emptyset \neq A\subseteq\cal{C}_u}\sum_{\emptyset \neq B\subseteq\cal{C}_v}\Delta_{A\cup B}^N(\pi_{\cal{C}_u\cup \cal{C}_v}^{\cancel{A\cup B}}(\varphi)).
\end{align*}
Given Minkowski's inequality, we need only to show that
$$\Ebb{\Delta_{A\cup B}^N(\pi_{\cal{C}_u\cup \cal{C}_v}^{\cancel{A\cup B}}(\varphi))^2}^{\frac{1}{2}}\leq\frac{C\norm{\varphi}}{N}\quad\forall N>0,\enskip \emptyset \neq A\subseteq\cal{C}_u,\enskip \emptyset \neq B\subseteq\cal{C}_v.$$
To do so, pick any nodes $u',v'$ respectively in any two subsets $A,B$ of $\cal{C}_u,\cal{C}_v$,  and note that
$$\Ebb{\Delta_{A\cup B}^N(\pi_{\cal{C}_u\cup \cal{C}_v}^{\cancel{A\cup B}}(\varphi))^2}^{\frac{1}{2}}=\Ebb{(\pi_{u'}^N-\pi_{u'})\times(\pi_{v'}^N-\pi_{v'})(\psi)^2}^{\frac{1}{2}},$$
where $\psi:=\Delta_{A\cup B\backslash\{u',v'\}}^N(\pi_{\cal{C}_u\cup \cal{C}_v}^{\cancel{A\cup B}}(\varphi))$.  
Because $u$ and $v$ lie on separate branches, and particles indexed by distinct nodes in $\cal{C}_u$ (or $\cal{C}_v$) are independent,  $(\bm{X}^{n,N}_{r})_{n=1}^N$ and  $(\bm{X}^{n,N}_{r'})_{n=1}^N$ are independent collections of particles for all $r$ and $r'$ in $A\cup B$ satisfying $r\neq r'$. Hence, $\pi_{u'}^N$ and $\pi_{v'}^N$ are independent of the sigma-algebra $\cal{G}$ generated by $((\bm{X}^{n,N}_{r})_{n=1}^N)_{r\in A\cup B\backslash\{u',v'\}}$ and the desired inequality follows from our assumption that~\eqref{eq:crossprop} holds for $u'$ and $v'$:
\begin{align}\nonumber\Ebb{\Delta_{A\cup B}^N(\pi_{\cal{C}_u\cup \cal{C}_v}^{\cancel{A\cup B}}(\varphi))^2}^\frac{1}{2}&=\Ebb{\Ebb{(\pi_{u'}^N-\pi_{u'})\times(\pi_{v'}^N-\pi_{v'})(\psi)^2|\cal{G}}}^\frac{1}{2}\leq\frac{C\Ebb{\norm{\psi}^2}^{\frac{1}{2}}}{N}\\
 \label{eq:sameargsss} &\leq \frac{C2^{(\mmag{A}+\mmag{B}-2)/2}\norm{\pi_{\cal{C}_u\cup \cal{C}_v}^{\cancel{A\cup B}}(\varphi)}}{N}\leq \frac{C2^{(\mmag{A}+\mmag{B}-2)/2}\norm{\varphi}}{N}.\end{align}
\end{proof}
The correction step also preserves these inequalities:
\begin{lemma}[Correction step]\label{lem:crosscorrect}If, in addition to Assumptions~\ref{ass:abscont}--\ref{ass:boundweights},~\eqref{eq:crosscoal} is satisfied for some non-leaf nodes $u$ and $v$ lying in separate branches (i.e.\ $u\not\in\mathbb{T}_v$ and $v\not\in\mathbb{T}_u$), then 
\begin{equation}\label{eq:crosscorrect}\Ebb{(\pi_{u_-}^N-\pi_{u_-})\times(\pi_{v_-}^N-\pi_{v_-})(\varphi)^2}^{\frac{1}{2}}\leq \frac{C\norm{\varphi}}{N}\quad\forall N>0,\enskip \varphi\in\cal{B}_b(\bm{E}_{\cal{C}_u}\times\bm{E}_{\cal{C}_v}).\end{equation}
\end{lemma}
\begin{proof}Fix any $N>0$ and $\varphi$ in $\cal{B}_b(\bm{E}_{\cal{C}_u}\times\bm{E}_{\cal{C}_v})$, and let
\begin{align*}\varrho_r^N(d\bm{x}_{\cal{C}_r}):=w_{r_-}(\bm{x}_{\cal{C}_r})\pi_{\cal{C}_r}^N(d\bm{x}_{\cal{C}_r}),\quad
\varrho_{r}(d\bm{x}_{\cal{C}_r}):=w_{r_-}(\bm{x}_{\cal{C}_r})\pi_{\cal{C}_r}(d\bm{x}_{\cal{C}_r}),\quad\forall r=v,u.\end{align*}
Because $\pi_{u_-}^N=\varrho_u^N/\pi_{\cal{C}_u}^N(w_{u_-})$ and $\pi_{u_-}=\varrho_u/\pi_{\cal{C}_u}(w_{u_-})$,
\begin{align*}\pi_{u_-}^N-\pi_{u_-}=\pi_{u_-}^N-\frac{\varrho_u}{\pi_{\cal{C}_u}(w_{u_-})}
=\frac{\varrho_u^N-\varrho_u-(\pi_{\cal{C}_u}^N(w_{u_-})-\pi_{\cal{C}_u}(w_{u_-}))\pi_{u_-}^N}{\pi_{\cal{C}_u}(w_{u_-})},\end{align*}
and similarly for $v$. Hence,
\begin{align}\pi_{\cal{C}_u}(w_{u_-})&\pi_{\cal{C}_v}(w_{v_-})(\pi_{u_-}^N-\pi_{u_-})\times(\pi_{v_-}^N-\pi_{v_-})(\varphi)\label{eq:mfe7w8agnfay8nga}\\
=&(\varrho_u^N-\varrho_u)\times(\varrho_v^N-\varrho_v)(\varphi)\nonumber\\
&-(\pi_{\cal{C}_v}^N(w_{v_-})-\pi_{\cal{C}_v}(w_{v_-}))(\varrho_u^N-\varrho_u)(\pi_{v_-}^N(\varphi))\nonumber\\
&-(\pi_{\cal{C}_u}^N(w_{u_-})-\pi_{\cal{C}_u}(w_{u_-}))(\varrho_v^N-\varrho_v)(\pi_{u_-}^N(\varphi))\nonumber\\
&+(\pi_{\cal{C}_u}^N(w_{u_-})-\pi_{\cal{C}_u}(w_{u_-}))(\pi_{\cal{C}_v}^N(w_{v_-})-\pi_{\cal{C}_v}(w_{v_-}))(\pi_{u_-}^N\times\pi_{v_-}^N)(\varphi).\nonumber
\end{align}
Because $\varrho^N_u(\varphi)=\pi_{\cal{C}_u}^N(w_{u_-}\varphi)$, $\varrho_u(\varphi)=\pi_{\cal{C}_u}(w_{u_-}\varphi)$, and similarly for $v$,~\eqref{eq:crosscoal} implies that
\begin{align}
\Ebb{\mmag{(\varrho_u^N-\varrho_u)\times(\varrho_v^N-\varrho_v)^2}}^{\frac{1}{2}}&=\Ebb{\mmag{(\pi_{\cal{C}_u}^N-\pi_{\cal{C}_u})\times(\pi_{\cal{C}_v}^N-\pi_{\cal{C}_v})(w_{u_-}w_{v_-}\varphi)^2}}^{\frac{1}{2}}\\
&\leq \frac{C\norm{w_{u_-}w_{v_-}\varphi}}{N}\leq \frac{C\norm{w_{u_-}}\norm{w_{v_-}}\norm{\varphi}}{N}.\nonumber
\end{align}
Next, because $u$ and $v$ lie in different branches, $\cal{F}_{\cal{C}_u}^N$ and $\cal{F}_{\cal{C}_v}^N$ are independent. By definition, $\pi_{\cal{C}_u}^N(w_{u_-})$ is $\cal{F}_{\cal{C}_u}^N$-measurable and $\pi_{\cal{C}_v}^N(w_{u_-})$ is $\cal{F}_{\cal{C}_v}^N$-measurable. Hence, the   $L^p$ inequalities for $\pi_{\cal{C}_u}^N$ and $\pi_{\cal{C}_v}^N$ in~\eqref{eq:lpcoa} (obtained in the proof of Theorem~\ref{thrm:lpapp}) imply that
\begin{align}
&\Ebb{(\pi_{\cal{C}_v}^N(w_{v_-})-\pi_{\cal{C}_v}(w_{v_-}))^2(\varrho_u^N-\varrho_u)(\pi_{v_-}^N(\varphi))^2}^{\frac{1}{2}}\\
&\qquad\qquad\qquad\qquad=\Ebb{(\pi_{\cal{C}_v}^N(w_{v_-})-\pi_{\cal{C}_v}(w_{v_-}))^2\Ebb{(\varrho_u^N-\varrho_u)(\pi_{v_-}^N(\varphi))^2|\cal{F}_v^N}}^{\frac{1}{2}}\nonumber\\
&\qquad\qquad\qquad\qquad\leq \Ebb{(\pi_{\cal{C}_v}^N(w_{v_-})-\pi_{\cal{C}_v}(w_{v_-}))^2\frac{C^2\norm{w_{u_-}\pi_{v_-}^N(\varphi)}^2}{N}}^{\frac{1}{2}}\nonumber\\
&\qquad\qquad\qquad\qquad\leq \frac{C\norm{w_{u_-}}\norm{\varphi}}{N^{1/2}}\Ebb{(\pi_{\cal{C}_v}^N(w_{v_-})-\pi_{\cal{C}_v}(w_{v_-}))^2}^{\frac{1}{2}}\nonumber\\
&\qquad\qquad\qquad\qquad\leq \frac{C\norm{w_{u_-}}\norm{w_{v_-}}\norm{\varphi}}{N}.\nonumber
\end{align}
Similarly, we find that 
\begin{align}&\Ebb{(\pi_{\cal{C}_u}^N(w_{u_-})-\pi_{\cal{C}_u}(w_{u_-}))^2(\varrho_v^N-\varrho_v)(\pi_{u_-}^N(\varphi))^2}^{\frac{1}{2}}\leq \frac{C\norm{w_{v_-}}\norm{w_{u_-}}\norm{\varphi}}{N},\\
&\Ebb{(\pi_{\cal{C}_u}^N(w_{u_-})-\pi_{\cal{C}_u}(w_{u_-}))^2(\pi_{\cal{C}_v}^N(w_{v_-})-\pi_{\cal{C}_v}(w_{v_-}))^2(\pi_{u_-}^N\times\pi_{v_-}^N)(\varphi)^2}^{\frac{1}{2}}\label{eq:mfe7w8agnfay8nga2}\\
&\qquad\leq \Ebb{(\pi_{\cal{C}_u}^N(w_{u_-})-\pi_{\cal{C}_u}(w_{u_-}))^2}^{\frac{1}{2}}\Ebb{(\pi_{\cal{C}_v}^N(w_{v_-})-\pi_{\cal{C}_v}(w_{v_-}))^2}^{\frac{1}{2}}\norm{\varphi}\nonumber\\
&\qquad\leq\frac{C\norm{w_{u_-}}\norm{w_{v_-}}\norm{\varphi}}{N}.\nonumber
\end{align}
Combining~(\ref{eq:mfe7w8agnfay8nga}--\ref{eq:mfe7w8agnfay8nga2}) and applying Minkowski's inequality then completes the proof.
\end{proof}
For the resampling step, recall that $\epsilon_u^N:=N^{-1}\sum_{n=1}^N\delta_{\bm{X}^{n,N}_{u_-}}$ denotes the empirical distribution of the resampled particles. We then have the following:
\begin{lemma}[Resampling step]\label{lem:crossres}If, in addition to Assumptions~\ref{ass:abscont}--\ref{ass:boundweights}, \eqref{eq:crosscorrect} is satisfied for some non-leaf nodes $u$ and $v$  lying in separate branches (i.e.\ $u\not\in\mathbb{T}_v$ and $v\not\in\mathbb{T}_u$), then
\begin{equation}\label{eq:crossres}\Ebb{(\epsilon_u^N-\pi_{u_-})\times(\epsilon_v^N-\pi_{v_-})(\varphi)^2}^{\frac{1}{2}}\leq \frac{C\norm{\varphi}}{N}\quad\forall N>0,\enskip \varphi\in\cal{B}_b(\bm{E}_{\cal{C}_u}\times\bm{E}_{v}).\end{equation}
\end{lemma}

\begin{proof}Fix any $N>0$ and $\varphi$ in $\cal{B}_b(\bm{E}_{\cal{C}_u}\times\bm{E}_{v})$. Note that
$$S_N:=N^2(\epsilon_u^N-\pi_{u_-})\times(\epsilon_v^N-\pi_{v_-})(\varphi)=\sum_{i=1}^N\sum_{j=1}^N\psi(\bm{X}^{i,N}_{u_-},\bm{X}^{j,N}_{v_-}),$$
where $\psi:=\varphi-\pi_{v_-}(\varphi)-\pi_{u_-}(\varphi)+(\pi_{u_-}\times\pi_{v_-})(\varphi)$. 
In other words,
$$S_N=\sum_{i=1}^N\sum_{j=1}^NY^N_{i,j}\quad\text{where}\quad Y^N_{i,j}:=\psi(\bm{X}^{i,N}_{u_-},\bm{X}^{j,N}_{v_-})\quad\forall i,j\leq N.$$
Conditioned on ${\cal{G}}:=\cal{F}_{\cal{C}_u}^N\vee\cal{F}_{\cal{C}_v}^N$, the resampled particles $\bm{X}^{1,N}_{u_-},\dots,\bm{X}^{N,N}_{u_-},$ $\bm{X}^{1,N}_{v_-},\dots,\bm{X}^{N,N}_{v_-}$ are independent with
$$\bm{X}^{n,N}_{u_-}\sim \pi_{u_-}^N,\quad\bm{X}^{n,N}_{v_-}\sim \pi_{v_-}^N,\quad\forall n\leq N.$$
We proceed by cases:
\begin{itemize}
\item \textit{(${i_1\neq i_2}$ and ${j_1\neq j_2}$).} With probability one,
\begin{align*}\Ebb{Y^N_{i_1,j_1}Y^N_{i_2,j_2}|{\cal{G}}}&=\Ebb{\psi(\bm{X}^{i_1,N}_{u_-},\bm{X}^{j_1,N}_{v_-})\psi(\bm{X}^{i_2,N}_{u_-},\bm{X}^{j_2,N}_{v_-})|{\cal{G}}}\\
&=\Ebb{\psi(\bm{X}^{i_1,N}_{u_-},\bm{X}^{j_1,N}_{v_-})|{\cal{G}}}\Ebb{\psi(\bm{X}^{i_2,N}_{u_-},\bm{X}^{j_2,N}_{v_-})|{\cal{G}}}\\
&=\pi_{u_-}^N\times \pi_{v_-}^N(\psi)^2.\end{align*}
But $\psi$'s definition implies that $\pi_{u_-}(\psi)=\pi_{v_-}(\psi)=0$ and $\norm{\psi}\leq 4\norm{\varphi}$. Hence, by~\eqref{eq:crosscorrect},
\begin{equation}\label{eq:dnmwu98abfha8wyfba}\Ebb{Y^N_{i_1,j_1}Y^N_{i_2,j_2}}\leq \frac{C^2\norm{\psi}^2}{N^2}\leq \frac{16C^2\norm{\varphi}^2}{N^2}\quad\text{if }i_1\neq i_2,\enskip j_1\neq j_2.\end{equation}

\item \textit{(${i_1\neq i_2}$ and ${j_1= j_2}$).} Let $\cal{G}_{j_1}:=\cal{G}\vee\sigma(\bm{X}^{j_1,N}_{v_-})$ so that
$$\Ebb{Y^N_{i_1,j_1}Y^N_{i_2,j_2}|\cal{G}_{j_1}}=\pi_{u_-}^N(\psi)(\bm{X}^{j_1,N}_{v_-})^2\quad\text{almost surely}.$$
Because $\pi_{u_-}(\psi)=0$ by definition, independence and the $L^p$ inequality for $\pi_{u_-}^N$ in~\eqref{eq:lpcor} (obtained in the proof of Theorem~\ref{thrm:lpapp}) then imply that%
\begin{equation}\label{eq:dnmwu98abfha8wyfba2}\Ebb{Y^N_{i_1,j_1}Y^N_{i_2,j_2}}\leq \frac{C^2\norm{\psi}^2}{N}\leq \frac{16C^2\norm{\varphi}^2}{N}\quad\text{if }i_1\neq i_2,\enskip j_1=j_2.\end{equation}
\item \textit{(${i_1= i_2}$ and ${j_1\neq j_2}$).} Similarly, 
\begin{equation}\label{eq:dnmwu98abfha8wyfba3}\Ebb{Y^N_{i_1,j_1}Y^N_{i_2,j_2}}\leq \frac{C^2\norm{\psi}^2}{N}\leq \frac{16C^2\norm{\varphi}^2}{N}\quad\text{if }i_1= i_2,\enskip j_1\neq j_2.\end{equation}
\item \textit{(${i_1= i_2}$ and ${j_1= j_2}$).} Clearly,
\begin{equation}\label{eq:dnmwu98abfha8wyfba4}\Ebb{Y^N_{i_1,j_1}Y^N_{i_2,j_2}}\leq\norm{\psi}^2\leq 16\norm{\varphi}^2\quad\text{if }i_1=i_2,\enskip j_1= j_2.\end{equation}
\end{itemize}
In the sum
$$\Ebb{S_{N}^2}=\sum_{i_1=1}^N\sum_{j_1=1}^N\sum_{i_2=1}^N\sum_{j_2=1}^N\Ebb{Y^N_{i_1,j_1}Y^N_{i_2,j_2}},$$
there are $N^2(N-1)^2$ terms of the type in~\eqref{eq:dnmwu98abfha8wyfba}, $2 N^2 (N-1)$ of the type in (\ref{eq:dnmwu98abfha8wyfba2},\ref{eq:dnmwu98abfha8wyfba3}), and $N^2$ of the type in \eqref{eq:dnmwu98abfha8wyfba4}. Hence, (\ref{eq:dnmwu98abfha8wyfba}--\ref{eq:dnmwu98abfha8wyfba4}) imply that $\Ebb{S_{N}^2}\leq N^2C\norm{\varphi}^2$ and the result follows.
\end{proof}
Similarly for the mutation step:
\begin{lemma}[Mutation step]\label{lem:crossmut}If, in addition to Assumptions~\ref{ass:abscont}--\ref{ass:boundweights}, $u$ and $v$ lie in separate branches (i.e.\ $u\not\in\mathbb{T}_v$ and $v\not\in\mathbb{T}_u$) and either (a) neither $u$ or $v$ is a leaf and \eqref{eq:crossres} holds or (b) $u$ or $v$ is a leaf, then~\eqref{eq:crossprop} also holds for $u$ and $v$.
\end{lemma}

\begin{proof} For the sake of brevity, we focus on the case that neither $u$ nor $v$ are leaves and skip the analogous, albeit simpler, case that one of the two is a leaf. 
Fix any $N>0$ and $\varphi$ in $\cal{B}_b(\bm{E}_{u}\times\bm{E}_{v})$, and note that
\begin{align}\nonumber(\pi_u^N-\pi_u)\times(\pi_v^N-\pi_v)(\varphi)=&(\pi_u^N-\nu_u^N+\nu_u^N-\pi_u)\times(\pi_v^N-\nu_v^N+\nu_v^N-\pi_v)(\varphi)\\
=&(\pi_u^N-\nu_u^N)\times(\pi_v^N-\nu_v^N)(\varphi)+(\pi_u^N-\nu_u^N)\times(\nu_v^N-\pi_v)(\varphi)\label{eq:fmnw78a0hfw8yfaw0}\\
&+ (\nu_u^N-\pi_u)\times(\pi_v^N-\nu_v^N)(\varphi)+(\nu_u^N-\pi_u)\times(\nu_v^N-\pi_v)(\varphi),\nonumber\end{align}
where $\nu_u^N:=\epsilon_u^N\times K_u$ and $\nu_v^N:=\epsilon_v^N\times K_v$. Because $K_u$ and $K_v$ are Markov kernels, 
the inequality~\eqref{eq:crossres} bounds the fourth term in~\eqref{eq:fmnw78a0hfw8yfaw0}:
\begin{align}\label{eq:fmnw78a0hfw8yfaw}
\Ebb{(\nu_u^N-\pi_u)\times(\nu_v^N-\pi_v)(\varphi)^2}^{\frac{1}{2}}&=\Ebb{(\epsilon_u^N-\pi_{u_-})\times(\epsilon_v^N-\pi_{v_-})(K_uK_v\varphi)^2}^{\frac{1}{2}}\\
&\leq \frac{C\norm{K_uK_v\varphi}}{N}\leq \frac{C\norm{\varphi}}{N}.\nonumber
\end{align}

To control the other three terms in~\eqref{eq:fmnw78a0hfw8yfaw0}, we emulate the approach in Lemma~\ref{lem:crossres}'s proof. For the first term, we set
$$S_N:=N^2(\pi_u^N-\nu_u^N)\times(\pi_v^N-\nu_v^N)(\varphi)=\sum_{i=1}^N\sum_{j=1}^NY^N_{i,j},$$
where  $Y^N_{i,j}:=\psi(\bm{X}^{i,N}_{u},\bm{X}^{j,N}_{v})$  with  $\psi:=\varphi-K_v\varphi-K_u\varphi+K_uK_v\varphi$.
Suppose that $i_1\neq i_2$ and set $\cal{G}:=\sigma(\bm{X}^{j_1,N}_v,\bm{X}^{j_2,N}_v)\vee\cal{F}_{u_-}^N$.  Because $u$ and $v$ lie in different branches, $\bm{X}^{j_1,N}_v$~and~$\bm{X}^{j_2,N}_v$ are independent of $\cal{F}_{u_-}^N$. Moreover, conditioned on $\cal{F}_{u_-}^N$, $X^{i_1,N}_u$ and $X^{i_2,N}_u$ are independent with respective laws $K_u(\bm{X}^{i_1,N}_{u_-},\cdot)$ and $K_u(\bm{X}^{i_2,N}_{u_-},\cdot)$. 
For these reasons,
\begin{align*}\Ebb{Y^N_{i_1,j_1}Y^N_{i_2,j_2}|\cal{G}}&=\Ebb{\psi(\bm{X}^{i_1,N}_u,\bm{X}^{j_1,N}_v)\psi(\bm{X}^{i_2,N}_u,\bm{X}^{j_2,N}_v)|\cal{G}}\\
&=\Ebb{\psi(\bm{X}^{i_1,N}_u,\bm{X}^{j_1,N}_v)|\cal{G}}\Ebb{\psi(\bm{X}^{i_2,N}_u,\bm{X}^{j_2,N}_v)|\cal{G}}\\
&=(K_u\psi)(\bm{X}^{i_1,N}_{u_-},\bm{X}^{j_1,N}_{v})(K_u\psi)(\bm{X}^{i_2,N}_{u_-},\bm{X}^{j_2,N}_{v})\quad\text{almost surely}.\end{align*}
But $\psi$'s definition implies that $K_u\psi=0$ and we have that $\Ebb{Y^N_{i_1,j_1}Y^N_{i_2,j_2}}=0$ if $i_1\neq i_2$. The same argument shows that this is also the case if $j_1\neq j_2$, and it follows that
$$\Ebb{S_{N}^2}=\sum_{i_1=1}^N\sum_{j_1=1}^N\sum_{i_2=1}^N\sum_{j_2=1}^N\Ebb{Y^N_{i_1,j_1}Y^N_{i_2,j_2}}=\sum_{i=1}^N\sum_{j=1}^N\Ebb{(Y^N_{i,j})^2}\leq N^2\norm{\psi}^2.$$
Because $K_u$ and $K_v$ are Markov kernels, $\norm{\psi}\leq 4\norm{\varphi}$ and, so,
\begin{equation}\label{eq:fmnw78a0hfw8yfaw2}
\Ebb{(\pi_u^N-\nu_u^N)\times(\pi_v^N-\nu_v^N)(\varphi)^2}^{\frac{1}{2}}=\frac{\Ebb{S_{N}^2}^{\frac{1}{2}}}{N^2}\leq \frac{4\norm{\varphi}}{N}.
\end{equation}

For the remaining two terms in~\eqref{eq:fmnw78a0hfw8yfaw0}, set 
$$S_N:=N^2(\pi_u^N-\nu_u^N)\times(\nu_v^N-\pi_v)(\varphi)=\sum_{i=1}^N\sum_{j=1}^NY^N_{i,j},$$
where $Y^N_{i,j}:=\psi(\bm{X}^{i,N}_u,\bm{X}^{j,N}_{v_-})$  with  $\psi:=K_v\varphi-\pi_{v_-}(K_v\varphi)-K_uK_v\varphi+K_u\pi_{v_-}(K_v\varphi)$. We proceed by cases:
\begin{itemize}
\item \textit{(${i_1\neq i_2}$).} Because $u$ and $v$ lie in separate branches, conditional on $\cal{G}:=\cal{F}_{u_-}^N\vee\cal{F}_{v_-}^N$, $X^{i_1,N}_u$~and~$X^{i_2,N}_u$ are independent with respective laws $K_u(\bm{X}^{i_1,N}_{u_-},\cdot)$ and $K_u(\bm{X}^{i_2,N}_{u_-},\cdot)$. For this reason,
\begin{align*}\Ebb{Y^N_{i_1,j_1}Y^N_{i_2,j_2}|{\cal{G}}}&=\Ebb{\psi(\bm{X}^{i_1,N}_u,\bm{X}^{j_1,N}_{v_-})\psi(\bm{X}^{i_2,N}_u,\bm{X}^{j_2,N}_{v_-})|{\cal{G}}}\\
&=\Ebb{\psi(\bm{X}^{i_1,N}_u,\bm{X}^{j_1,N}_{v_-})|{\cal{G}}}\Ebb{\psi(\bm{X}^{i_2,N}_u,\bm{X}^{j_2,N}_{v_-})|{\cal{G}}}\\
&=(K_u\psi)(\bm{X}^{i_1,N}_{u_-},\bm{X}^{j_1,N}_{v_-})(K_u\psi)(\bm{X}^{i_2,N}_{u_-},\bm{X}^{j_2,N}_{v_-})\quad\text{almost surely}.\end{align*}
But $\psi$'s definition implies that $K_u\psi=0$ and so
\begin{equation}\label{eq:fnueaw9ngfeu9anguwaege}\Ebb{Y^N_{i_1,j_1}Y^N_{i_2,j_2}}=0\quad\text{if }i_1\neq i_2.
\end{equation}

\item \textit{(${i_1=i_2}$ and ${j_1\neq j_2}$).} If $\cal{G}:=\sigma(\bm{X}^{i_1,N}_u)\vee\cal{F}_{\cal{C}_v}^N$, then 
\begin{align*}\Ebb{Y^N_{i_1,j_1}Y^N_{i_2,j_2}|\cal{G}}&=\Ebb{\psi(\bm{X}^{i_1,N}_u,\bm{X}^{j_1,N}_{v_-})\psi(\bm{X}^{i_1,N}_u,\bm{X}^{j_2,N}_{v_-})|\cal{G}}\\
&=\Ebb{\psi(\bm{X}^{i_1,N}_u,\bm{X}^{j_1,N}_{v_-})|\cal{G}}\Ebb{\psi(\bm{X}^{i_1,N}_u,\bm{X}^{j_2,N}_{v_-})|\cal{G}}\\
&=\pi_{v_-}^N(\psi)(\bm{X}^{i_1,N}_u)^2\quad\textrm{almost surely}.\end{align*}
But, by definition, $\pi_{v_-}(\psi)=0$ and the $L^p$ inequality for $\pi_{u_-}^N$ in Theorem~\ref{THRM:LP}  shows that 
\begin{align}\label{eq:fnueaw9ngfeu9anguwaege2}\Ebb{Y^N_{i_1,j_1}Y^N_{i_2,j_2}}&\leq \frac{16C^2\norm{\varphi}^2}{N}\quad\text{if }i_1=i_2,\enskip j_1\neq j_2.
\end{align}
\item \textit{(${i_1= i_2}$ and ${j_1= j_2}$).} Clearly,
\begin{equation}\label{eq:fnueaw9ngfeu9anguwaege3}\Ebb{Y^N_{i_1,j_1}Y^N_{i_2,j_2}}\leq\norm{\psi}^2\leq16\norm{\varphi}^2\quad\text{if }i_1=i_2,\enskip j_1= j_2.\end{equation}
\end{itemize}
In the sum
$$\Ebb{S_{N}^2}=\sum_{i_1=1}^N\sum_{j_1=1}^N\sum_{i_2=1}^N\sum_{j_2=1}^N\Ebb{Y^N_{i_1,j_1}Y^N_{i_2,j_2}},$$
there are   $N^2(N-1)$ terms of the type in~\eqref{eq:fnueaw9ngfeu9anguwaege2} and $N^2$ of that in \eqref{eq:fnueaw9ngfeu9anguwaege3}. Hence, (\ref{eq:fnueaw9ngfeu9anguwaege}--\ref{eq:fnueaw9ngfeu9anguwaege3}) imply that $\Ebb{S_{N}^2}\leq N^2C^2\norm{\varphi}^2$ and we find that
$$
\Ebb{(\pi_u^N-\nu_u^N)\times(\nu_v^N-\pi_v)(\varphi)^2}^{\frac{1}{2}}\leq \frac{C\norm{\varphi}}{N}.
$$
For the same reasons, we also have that
$$\Ebb{(\nu_u^N-\pi_u)\times(\pi_v^N-\nu_v^N)(\varphi)^2}^{\frac{1}{2}}\leq \frac{C\norm{\varphi}}{N}.$$
Combining (\ref{eq:fmnw78a0hfw8yfaw0}--\ref{eq:fmnw78a0hfw8yfaw2}), the above, and Minkowski's inequality, completes the proof.
\end{proof}
\begin{proof}[Proof of Theorem~\ref{THRM:CROSS}]Repeatedly applying Lemmas~\ref{lem:crosscorrect}--\ref{lem:crossmut} starting from~\eqref{eq:crossprop}, we find that, for any  $u$ and $v$ lying in separate branches (i.e.\ $u\not\in\mathbb{T}_v$ and $v\not\in\mathbb{T}_u$), 
\begin{equation}\label{eq:fmeu9agney7a8gbeyaw}
\Ebb{(\pi^N_u-\pi_u)\times(\pi_v^N-\pi_v)(\varphi)^2}^{\frac{1}{2}}\leq \frac{C\norm{\varphi}}{N}\quad\forall N>0,\enskip \varphi\in\cal{B}_b(\bm{E}_u\times\bm{E}_v).\end{equation}
The remainder of this proof is quite similar to that of Lemma~\ref{lem:crosscorrect} and we only sketch it: 
\begin{align*}
(\gamma_u^N-\gamma_u)\times(\gamma_v^N-\gamma_v)&(\varphi)\\
=&(\gamma_u^N-\cal{Z}_u\pi_u^N+\cal{Z}_u\pi_u^N-\gamma_u)\times(\gamma_v^N-\cal{Z}_v\pi_v^N+\cal{Z}_v\pi_v^N-\gamma_v)(\varphi)\\
=&(\cal{Z}_u^N-\cal{Z}_u)(\cal{Z}_v^N-\cal{Z}_v)\pi_u^N\times \pi_v^N(\varphi)\\
&-(\cal{Z}_u^N-\cal{Z}_u)\cal{Z}_v(\pi_v^N-\pi_v)(\pi^N_u(\varphi))\\
&-(\cal{Z}_v^N-\cal{Z}_v)\cal{Z}_u(\pi_u^N-\pi_u)(\pi^N_v(\varphi))\\
&+\cal{Z}_u\cal{Z}_v(\pi_u^N-\pi_u)\times (\pi_v^N-\pi_v)(\varphi)\enskip\forall N>0,\enskip \varphi\in\cal{B}_b(\bm{E}_u\times\bm{E}_v).
\end{align*}
To control the  $L^2$ norm of the fourth term, use~\eqref{eq:fmeu9agney7a8gbeyaw}. For the first three, use independence, the $L^p$ inequalities for $\pi_u^N$ and $\pi_v^N$ in Theorem~\ref{thrm:lpapp}, and those for $\cal{Z}_u^N$ and $\cal{Z}_v^N$ in~\eqref{eq:lpz}. Then, apply Minkowski's inequality.
\end{proof}
\section{Proof of Theorem~\ref{THRM:LOCALLY}}\label{app:locally}
Throughout this appendix, we use the notation described in Appendix~\ref{app:notation}.  
The aim of the appendix is to prove (\ref{eq:locally11}--\ref{eq:locally2}) in Theorem~\ref{THRM:LOCALLY}. Before we do so, we take a moment to verify that all of the terms in the theorem's statement are well-defined. First off, note that $\omega_u(K_u)$ therein is well-defined, and can be chosen to be positive everywhere, because, by $\cal{S}$'s definition and Assumption~\ref{ass:abscont}, 
$$\rho_u\sim\gamma_{u_-}\times K_u\sim\gamma_{\cal{C}_u}\times K_u\sim\rho_{\cal{C}_u}\times K_u\quad\forall  (\gamma_{u_-},K_u)\in  \cal{S}.$$ 
The above further implies that $C\sqrt{K_u\omega_u^2}\rho_{\cal{C}_u}\sim\rho_{\cal{C}_u}$. Hence, if $(\gamma_{u_-},K_u)$ belongs to $\cal{S}$,
\begin{align*}&C\sqrt{K_u\omega_u^2}\rho_{\cal{C}_u}\times K_u\sim\rho_{\cal{C}_u}\times K_u\sim\gamma_{\cal{C}_u}\times K_u\times \gamma_{u_-}\times K_u\sim\rho_u,\\
&\gamma_{u_-}\times K_u\sim\rho_u\sim\mu_u\sim\rho_u^{u_-}\times M_u\Rightarrow\gamma_{u_-}\sim \mu_u^{u_-}\Rightarrow \gamma_{u_-}\times M_u\sim \mu_u^{u_-}\times M_u=\mu_u\sim\rho_u.\end{align*}
In other words, both $(C\sqrt{K_u\omega_u^2}\rho_{\cal{C}_u},K_u)$ and $(\gamma_{u_-},M_u)$ belong to $\cal{S}$ whenever $(\gamma_{u_-},K_u)$ does and all terms in the statement are well-defined.

To prove (\ref{eq:locally11}--\ref{eq:locally2}), we  need only argue~(\ref{eq:locally11}--\ref{eq:locally12}) as~\eqref{eq:locally2} follows immediately:
\begin{align*}
f(\gamma_{u_-},K_u)\geq f(\gamma_{u_-},M_u)\geq f(\sqrt{M_u\omega_u(M_u)^2}\rho_{\cal{C}_u},M_u)\quad\forall (\gamma_{u_-},K_u)\in\cal{S},
\end{align*}
and, given that $\rho_u=Z_u\mu_u=Z_u\mu_u^{u_-}\times M_u=\rho_u^{u_-}\times M_u$, 
\begin{align}
\omega_u(K_u)&=\frac{d\rho_u}{d\rho_{\cal{C}_u}\times K_u}=\frac{d\rho_u^{u_-}\times M_u}{d\rho_{\cal{C}_u}\times K_u}=\frac{d\rho_u^{u_-}}{d\rho_{\cal{C}_u}}\frac{dM_u}{dK_u}\nonumber\\
\Rightarrow
K_u\omega_u(K_u)^2&=\left(\frac{d\rho_u^{u_-}}{d\rho_{\cal{C}_u}}\right)^2K_u\left[\frac{dM_u}{dK_u}\right]^2\label{eq:ndwan7dwu8hda8}
\Rightarrow \sqrt{K_u\omega_u(M_u)^2}\rho_{\cal{C}_u}=\frac{d\rho_u^{u_-}}{d\rho_{\cal{C}_u}}\rho_{\cal{C}_u}=\rho_u^{u_-}.
\end{align}

To argue~(\ref{eq:locally11}--\ref{eq:locally12}), we express $Z_u^N$'s variance as 
\begin{equation}\label{eq:dsadsadfsafafwa}\text{Var}(Z_u^N)=A+\Ebb{Bg(X_{\cal{C}_u}^{N},\gamma_{u_-},K_u)}\end{equation}
where $A$ is a constant and $B$ is a random variable, neither dependent on $(\gamma_{u_-},K_u)$, $\bm{X}_{\cal{C}_u}^{N}:=(\bm{X}_{\cal{C}_u}^{\bm{n},N})_{\bm{n}\in[N]^{c_u}}^N$, and $g$ is a  real-valued function on $\bm{E}_{\cal{C}_u}^{N^{c_u}}\times\cal{S}$. We will then show that
\begin{align}\label{eq:dnwa9ndwu9nfasdauwaw1}g(\bm{x}_{\cal{C}_u}^N,\gamma_{u_-},K_u)&\geq g(\bm{x}_{\cal{C}_u}^N,\sqrt{K_u\omega_u(K_u)^2}\rho_{\cal{C}_u},K_u)\quad\forall \bm{x}_{\cal{C}_u}^N\in \bm{E}_{\cal{C}_u}^{N^{c_u}},\\
 \label{eq:dnwa9ndwu9nfasdauwaw2}g(\bm{x}_{\cal{C}_u}^N,\gamma_{u_-},K_u)&\geq g(\bm{x}_{\cal{C}_u}^N,\gamma_{u_-},M_u)\quad\forall \bm{x}_{\cal{C}_u}^N\in \bm{E}_{\cal{C}_u}^{N^{c_u}},
\end{align}
for all $(\gamma_{u_-},K_u)$ in $\cal{S}$. Given~\eqref{eq:dsadsadfsafafwa}, (\ref{eq:locally11}--\ref{eq:locally12}) then follow from the above.

To carry out these steps, we'll need several identities that we collect here (throughout the following, recall that $\gamma_u$ depends on $\gamma_{u_{-}}$ and $K_u$ via $\gamma_u=\gamma_{u_-}\times K_u$): 
$$\gamma_{u_-}([K_uw_u]\varphi)=\gamma_{u_-}(K_u[w_u\varphi])=\gamma_{u}(w_u\varphi)=\rho_u(\varphi)=\rho_u^{u_-}(\varphi)\quad \forall \varphi\in \cal{B}(\bm{E}_{\cal{C}_u}),$$
where $\rho_u^{u_-}$ denotes the $\bm{E}_{\cal{C}_u}$-marginal of $\rho_u$, from which it follows that 
\begin{align}
K_uw_u&=\frac{d\rho_u^{u_-}}{d\gamma_{u_-}}\nonumber\\
\Rightarrow w_{u_-}K_uw_u&=\frac{d\gamma_{u_-}}{d\gamma_{\cal{C}_u}}\frac{d\rho_u^{u_-}}{d\gamma_{u_-}}=\frac{d\rho_u^{u_-}}{d\gamma_{\cal{C}_u}},\label{eq:identity2}\\
w_u&=\frac{d\rho_u}{d\gamma_u}=\frac{d\rho_{\cal{C}_u}\times K_u}{d\gamma_u}\frac{d\rho_u}{d\rho_{\cal{C}_u}\times K_u}=\frac{d\rho_{\cal{C}_u}\times K_u}{d\gamma_{u_-}\times K_u}\omega_u(K_u)\nonumber\\
&=\frac{d\rho_{\cal{C}_u}}{d\gamma_{u_-}}\omega_u(K_u)\nonumber\\
\label{eq:identity4}\Rightarrow w_{u_-}K_uw_u^2&=\frac{d\gamma_{u_-}}{d\gamma_{\cal{C}_u}}\left(\frac{d\rho_{\cal{C}_u}}{d\gamma_{u_-}}\right)^2K_u\omega_u(K_u)^2=\frac{d\rho_{\cal{C}_u}}{d\gamma_{\cal{C}_u}}\frac{d\rho_{\cal{C}_u}}{d\gamma_{u_-}}K_u\omega_u(K_u)^2\\
&=\frac{w_{\cal{C}_u}^2}{w_{u_-}}K_u\omega_u(K_u)^2,\quad\text{where }w_{\cal{C}_u}:=\frac{d\rho_{\cal{C}_u}}{d\gamma_{\cal{C}_u}}.\nonumber
\end{align}

Let's start in earnest: by the law of total variance,
\begin{align}\label{eq:dsadasda}
\textrm{Var}(Z_u^N)=\textrm{Var}(\Ebb{Z_u^N|\cal{F}_{\cal{C}_u}^N})+\Ebb{\textrm{Var}(Z_u^N|\cal{F}_{\cal{C}_u}^N)}.
\end{align}
Using~\eqref{eq:identity2}, we find that
\begin{align*}\Ebb{Z_u^N|\cal{F}_{\cal{C}_u}^N}&=\frac{\cal{Z}_u^N}{N}\sum_{n=1}^N\Ebb{w_u(\bm{X}_u^{n,N})|\cal{F}_{\cal{C}_u}^N}=\cal{Z}_u^N\pi_{u_-}^N(K_uw_u)=\gamma_{u_-}^N(K_uw_u)\\
&=\frac{\cal{Z}_{\cal{C}_u}^N}{N^{c_u}}\sum_{\bm{n}\in [N]^{c_u}}w_{u_-}(\bm{X}_{\cal{C}_u}^{\bm{n},N})(K_uw_u)(\bm{X}_{\cal{C}_u}^{\bm{n},N})
=\frac{\cal{Z}_{\cal{C}_u}^N}{N^{c_u}}\sum_{\bm{n}\in [N]^{c_u}}\frac{d\rho_u^{u_-}}{d\gamma_{\cal{C}_u}}(\bm{X}_{\cal{C}_u}^{\bm{n},N}).
\end{align*}
It follows that the first term in~\eqref{eq:dsadasda} does not depend on $(\gamma_{u_-},K_u)$ and we absorb it into $A$ in~\eqref{eq:dsadsadfsafafwa}. Similarly,
\begin{align*}
\text{Var}(Z_u^N|\cal{F}_{\cal{C}_u}^N)&=\text{Var}(\gamma_u^N(w_u)|\cal{F}_{\cal{C}_u}^N)=[N^{-1}\cal{Z}_u^N]^2\sum_{n=1}^N\text{Var}(w_u(\bm{X}_u^{n,N})|\cal{F}_{\cal{C}_u}^N)\\
&=N^{-1}[\cal{Z}_u^N]^2[\pi_{u_-}^N(K_uw_u^2)-\pi_{u_-}^N(K_uw_u)^2]\\
&=N^{-1}\cal{Z}_u^N\gamma_{u_-}^N(K_uw_u^2)-N^{-1}\gamma_{u_-}^N(K_uw_u)^2\quad\text{a.s.},
\end{align*}
and, using again~\eqref{eq:identity2}, we obtain
\begin{align*}
\gamma_{u_-}^N(K_uw_u)=\frac{\cal{Z}_{\cal{C}_u}^N}{N^{c_u}}\sum_{\bm{n}\in [N]^{c_u}}w_{u_-}(\bm{X}_{\cal{C}_u}^{\bm{n},N})(K_uw_u)(\bm{X}_{\cal{C}_u}^{\bm{n},N})=\frac{\cal{Z}_{\cal{C}_u}^N}{N^{c_u}}\sum_{\bm{n}\in [N]^{c_u}}\frac{d\rho_u^{u_-}}{d\gamma_{\cal{C}_u}}(\bm{X}_{\cal{C}_u}^{\bm{n},N}),
\end{align*}
another term that does not depend on $(\gamma_{u_-},K_u)$ and that we absorb into $A$ in~\eqref{eq:dsadsadfsafafwa}. Next,
\begin{align*}
\frac{\cal{Z}_u^N}{N}\gamma_{u_-}^N(K_uw_u^2)=\frac{[\cal{Z}_{\cal{C}_u}^N]^2}{N^{2c_u+1}}\left(\sum_{\bm{m}\in [N]^{c_u}}w_{u_-}(\bm{X}_{\cal{C}_u}^{\bm{m},N})\right)\sum_{\bm{n}\in [N]^{c_u}}w_{u_-}(\bm{X}_{\cal{C}_u}^{\bm{n},N})(K_uw_u^2)(\bm{X}_{\cal{C}_u}^{\bm{n},N}),
\end{align*}
and~\eqref{eq:dsadsadfsafafwa} follows with $A$ as above, $B:=N^{-1}[\cal{Z}_{\cal{C}_u}^N]^2$, and 
\begin{align}g(\bm{x}_{\cal{C}_u}^N,\gamma_{u_-},K_u):&=\frac{1}{N^{2c_u}}\left(\sum_{\bm{m}\in [N]^{c_u}}w_{u_-}(\bm{x}_{\cal{C}_u}^{\bm{m},N})\right)\sum_{\bm{n}\in [N]^{c_u}}w_{u_-}(\bm{x}_{\cal{C}_u}^{\bm{n}})(K_uw_u^2)(\bm{x}_{\cal{C}_u}^{\bm{n}})\nonumber\\
&=\frac{1}{N^{2c_u}}\sum_{\bm{n}\in [N]^{c_u}}\frac{w_{\cal{C}_u}(\bm{x}_{\cal{C}_u}^{\bm{n}})^2}{\bar{w}_{u_-}(\bm{x}_{\cal{C}_u}^{\bm{n}})}(K_u\omega_u(K_u)^2)(\bm{x}_{\cal{C}_u}^{\bm{n},N})\label{eq:gdef},\end{align}
where $\bar{w}_{u_-}(\bm{x}_{\cal{C}_u}^{\bm{n}}):=\frac{w_{u_-}^n(\bm{x}_{\cal{C}_u}^{\bm{n},N})}{\sum_{\bm{m}\in [N]^{c_u}}w_{u_-}(\bm{x}_{\cal{C}_u}^{\bm{m}})}$ and the second equality follows from~\eqref{eq:identity4}. 

Because Jensen's inequality implies that
$$K_u\left[\frac{dM_u}{dK_u}\right]^2\geq \left[K_u\frac{dM_u}{dK_u}\right]^2=1=M_u\left[\frac{dM_u}{dM_u}\right]^2,$$
\eqref{eq:ndwan7dwu8hda8} implies that $K_u\omega_u(K_u)^2$ is bounded below by $M_u\omega_u(M_u)^2$, and \eqref{eq:dnwa9ndwu9nfasdauwaw2} follows from \eqref{eq:gdef}. Lastly, note that the only terms in~\eqref{eq:gdef} that depend on $\gamma_{u_-}$ are $(\bar{w}_{u_-}(\bm{x}_{\cal{C}_u}^{\bm{n}}))_{\bm{n}\in[N]^{c_u}}$. Treating these as free variables and using a Lagrange multiplier like in~\cite[p.153]{Asmussen2007} to minimize~\eqref{eq:gdef} subject to the constraint `$\sum_{\bm{n}\in[N]^{c_u}}\bar{w}_{u_-}(\bm{x}_{\cal{C}_u}^{\bm{n}})=1$', we obtain~\eqref{eq:dnwa9ndwu9nfasdauwaw2}. 
\section{Proof of Theorem~\ref{THRM:SMCVSDAC}}\label{app:dacvssmcproofs}
As shown in the proof of Theorem~\ref{THRM:CLT} in Appendix~\ref{app:clt}, for all $\varphi$ in $\cal{B}_b(\bm{E}_{\mathfrak{r}})$,
$$\sigma^2_{\rho_{\mathfrak{r}}}(\varphi)=\sigma^2_{\gamma_{\mathfrak{r}}}(w_{\mathfrak{r}}\varphi),\quad \sigma^2_{\mu_{\mathfrak{r}}}(\varphi)=\sigma^2_{\gamma_{\mathfrak{r}}}(w_{\mathfrak{r}}Z_{\mathfrak{r}}^{-1}[\varphi-\mu_{\mathfrak{r}}(\varphi)]),$$
where $\sigma^2_{\gamma_{\mathfrak{r}}}(\varphi)$ denotes the asymptotic variance of $\gamma_{\mathfrak{r}}^N(\varphi)$ (cf.\ Theorem~\ref{THRM:CLTun}). Because, as explained in the main text, Algorithm~\ref{alg:adacsmc} reduces to Algorithm~\ref{alg:asmc} if we replace $\mathbb{T}$ with $\mathbb{L}$, the same equations hold if we replace $(\sigma^2_{\mu_{\mathfrak{r}}},\sigma^2_{\rho_{\mathfrak{r}}},\sigma^2_{\gamma_{\mathfrak{r}}})$ with $(\sigma^2_{\mu_{\mathfrak{r}},smc},\sigma^2_{\rho_{\mathfrak{r}},smc},\sigma^2_{\gamma_{\mathfrak{r}},smc})$ where, for any $\varphi$ in $\cal{B}_b(\bm{E}_{\mathfrak{r}})$, $\sigma^2_{\gamma_{\mathfrak{r}},smc}(\varphi)$ denotes  the asymptotic variance of $\gamma_{\mathfrak{r}}^N(\varphi)$ using $\mathbb{L}$ instead of $\mathbb{T}$.  For these reasons, we only need to show that $\sigma^2_{\gamma_{\mathfrak{r}}}(\varphi)\leq \sigma^2_{\gamma_{\mathfrak{r}},smc}(\varphi)$ for all $\varphi$ in $\cal{B}_b(\bm{E}_{\mathfrak{r}})$.

Because we obtain $\mathbb{L}$ by repeatedly lumping the children of $\mathbb{T}$'s nodes, it suffices to show that this lumping does not lower the asymptotic variance. Proposition~\ref{prop:sameLambda} below shows that lumping of the children of a node $u$ only affects the terms in the asymptotic variance sums corresponding to those children, and it follows that
\begin{equation}\label{eq:mgfeu9agn7u9aenguaa}
\sigma_{\gamma_{\mathfrak{r}},L}^2(\varphi)-\sigma_{\gamma_{\mathfrak{r}}}^2(\varphi)=\pi_{\cal{C}_u}([\cal{Z}_{\cal{C}_u}w_{u_-}K_u\Gamma_{u,\mathfrak{r}}\varphi-\gamma_{\mathfrak{r}}(\varphi)]^2)-\sum_{v\in\cal{C}_u}\pi_v([\cal{Z}_{v}\Gamma_{v,\mathfrak{r}}\varphi-\gamma_{\mathfrak{r}}(\varphi)]^2),\end{equation}
where $\sigma_{\gamma_{\mathfrak{r}},L}^2(\varphi)$ denotes the asymptotic variance of the estimator for $\gamma_{\mathfrak{r}}(\varphi)$  obtained after lumping $u$'s children together. But,
\begin{align*}
\gamma_{\mathfrak{r}}(\varphi)&=\gamma_u(\Gamma_{u,\mathfrak{r}}\varphi)=\gamma_{u_-}(K_u\Gamma_{u,\mathfrak{r}}\varphi)=\gamma_{\cal{C}_u}(w_{u_-}K_u\Gamma_{u,\mathfrak{r}}\varphi)=\cal{Z}_{\cal{C}_u}\pi_{\cal{C}_u}(w_{u_-}K_u\Gamma_{u,\mathfrak{r}}\varphi),\\
\cal{Z}_{v}\Gamma_{v,\mathfrak{r}}\varphi&=\cal{Z}_{v}\Gamma_{v,u}\Gamma_{u,\mathfrak{r}}\varphi=\cal{Z}_{v}\gamma_{\cal{C}_u}^{\not v}(w_{u_-}K_u\Gamma_{u,\mathfrak{r}}\varphi)=\cal{Z}_{\cal{C}_u}\pi_{\cal{C}_u}^{\not v}(w_{u_-}K_u\Gamma_{u,\mathfrak{r}}\varphi)\quad\forall v\in\cal{C}_u;
\end{align*}
and it follows that
\begin{align*}
\cal{Z}_{\cal{C}_u}^{-2}[\sigma_{\gamma_{\mathfrak{r}},L}^2(\varphi)-\sigma_{\gamma_{\mathfrak{r}}}^2(\varphi)]=\pi_{\cal{C}_u}([\phi-\pi_{\cal{C}_u}(\phi)]^2)-\sum_{v\in\cal{C}_u}\pi_v(\pi_{\cal{C}_u}^{\not v}(\phi-\pi_v(\phi))^2),
\end{align*}
where $\phi:=w_{u_-}K_u\Gamma_{u,\mathfrak{r}}\varphi$. However,
\begin{align*}
\pi_{u1}(\pi_{\cal{C}_u}^{\not {u1}}(\phi-\pi_{u1} (\phi))^2)&=\pi_{u1}(\pi_{u_2}(\psi-\pi_{u1} (\psi))^2),\\
\pi_{u2}(\pi_{\cal{C}_u}^{\not {u2}}(\phi-\pi_{u2} (\phi))^2)&=\pi_{u2}(\pi_{u_1}(\psi-\pi_{u2} (\psi))^2),
\end{align*}
where $\psi:=\pi_{u[3:c_u]}(\phi)$ and $v[3:c_u]:=\{u3,\dots,uc_u\}$.  Jensen's inequality implies that 
$$\pi_{u2}(\pi_{u_1}(\psi-\pi_{u2} (\psi))^2)  \leq \pi_{u2}(\pi_{u_1}([\psi-\pi_{u2} (\psi)]^2)) ,$$
and it follows that 
\begin{align}
&\pi_{u1}(\pi_{\cal{C}_u}^{\not {u1}}(\phi-\pi_{u1} (\phi))^2)+\pi_{u2}(\pi_{\cal{C}_u}^{\not {u2}}(\phi-\pi_{u2} (\phi))^2)\nonumber\\
&\qquad\qquad\qquad\qquad\qquad\leq \pi_{u1}(\pi_{u_2}(\psi-\pi_{u1} (\psi))^2)+\pi_{u2}(\pi_{u_1}([\psi-\pi_{u2} (\psi)]^2))\nonumber\\
&\qquad\qquad\qquad\qquad\qquad=\pi_{u1}([\pi_{u_2}(\psi)-\pi_{u[2]} (\psi)]^2)+\pi_{u_1}(\pi_{u2}([\psi-\pi_{u2} (\psi)]^2))\nonumber\\
&\qquad\qquad\qquad\qquad\qquad=\pi_{u1}(\pi_{u_2}(\psi)^2)-\pi_{u[2]} (\psi)^2+\pi_{u_1}(\pi_{u2}(\psi^2)-\pi_{u2} (\psi)^2)\nonumber\\
&\qquad\qquad\qquad\qquad\qquad=\pi_{u[2]} (\psi^2)-\pi_{u[2]} (\psi)^2=\pi_{u[2]}([\psi-\pi_{u[2]}(\psi)]^2).\nonumber
\end{align}
Setting now $\psi:=\pi_{u[4:c_u]}(\phi)$, we have that
\begin{align*}\sum_{k=1}^2\pi_{uk}(\pi_{\cal{C}_u}^{\not {uk}}(\phi-\pi_{uk} (\phi))^2)&\leq \pi_{u[2]}(\pi_{u3}(\psi-\pi_{u[2]}(\psi))^2),\\
 \pi_{u3}(\pi_{\cal{C}_u}^{\not {u3}}(\phi-\pi_{u3} (\phi))^2)&=\pi_{u3}(\pi_{u[2]}(\psi-\pi_{u3}(\psi))^2).\end{align*}
Hence, applying Jensen's inequality again, we obtain
$$\sum_{k=1}^3\pi_{uk}(\pi_{\cal{C}_u}^{\not {uk}}(\phi-\pi_{uk} (\phi))^2)\leq \pi_{u[3]}([\psi-\pi_{u[3]}(\psi)]^2).$$
Iterating this argument gives us the desired bound:
\begin{align*}
\sum_{v\in\cal{C}_u}\pi_v(\pi_{\cal{C}_u}^{\not v}(\phi-\pi_v (\phi))^2)&=\sum_{k=1}^{c_u}\pi_{uk}(\pi_{\cal{C}_u}^{\not {uk}}(\phi-\pi_{uk} (\phi))^2)\leq \pi_{u[K]}(\pi_{\emptyset}(\phi-\pi_{u[K]}(\phi))^2)\\
&=\pi_{\cal{C}_u}([\phi-\pi_{\cal{C}_u}(\phi)]^2).
\end{align*}
We have one loose end to tie off: proving~\eqref{eq:mgfeu9agn7u9aenguaa}.
\begin{proposition}\label{prop:sameLambda}If Assumptions~\ref{ass:abscont}--\ref{ass:boundweights} are satisfied, then~\eqref{eq:mgfeu9agn7u9aenguaa} holds for all $\varphi$ in $\cal{B}_b(\bm{E}_{\mathfrak{r}})$.
\end{proposition}
\begin{proof}Fix any $\varphi$ in $\cal{B}_b(\bm{E}_{\mathfrak{r}})$. By its definition in Theorem~\ref{THRM:CLTun},
\begin{align}\sigma_{\gamma_{\mathfrak{r}}}^2(\varphi)=&\sum_{v\in\cal{C}_u}\pi_v([\cal{Z}_{v}\Gamma_{v,\mathfrak{r}}\varphi-\gamma_{\mathfrak{r}}(\varphi)]^2)+\sum_{v\in\mathbb{T}\backslash\cal{C}_u}\pi_v([\cal{Z}_{v}\Gamma_{v,\mathfrak{r}}\varphi-\gamma_{\mathfrak{r}}(\varphi)]^2)
\label{eq:dnw78abfyw8abfya8w1}\end{align}
where we have singled out the terms corresponding to $u$'s children. Because the lumping only affects the terms indexed by $u$'s children, we have that
\begin{align}\sigma_{\gamma_{\mathfrak{r}},L}^2(\varphi)=&\pi_{\cal{C}_u}([\cal{Z}_{\cal{C}_u}\Gamma_{\cal{C}_u,\mathfrak{r}}^L\varphi-\gamma_{\mathfrak{r}}(\varphi)]^2)+\sum_{v\in\mathbb{T}\backslash\cal{C}_u}\pi_v([\cal{Z}_{v}\Gamma_{v,\mathfrak{r}}^L\varphi-\gamma_{\mathfrak{r}}(\varphi)]^2),\label{eq:dnw78abfyw8abfya8w2}\end{align}
where $\Gamma_{v,\mathfrak{r}}^L$ denotes the kernel defined analogously to $\Gamma_{v,\mathfrak{r}}$ in~(\ref{eq:Pivu}--\ref{eq:Pivu2}) except that $u$'s children  have been lumped together into a single node which, with an abuse of notation, we denote $\cal{C}_u$ (with $\gamma_{{\cal{C}_u}_-}$ and $K_{\cal{C}_u}$ correspondingly set to $\prod_{v\in\cal{C}_u}\gamma_{v_-}$ and $\prod_{v\in\cal{C}_u}K_v$, cf.\ \eqref{eq:Kcu} for the latter). Given~(\ref{eq:dnw78abfyw8abfya8w1},\ref{eq:dnw78abfyw8abfya8w2}), we   need only show that
\begin{align}\label{eq:proofobj1}\Gamma_{r,\mathfrak{r}}^L\varphi&=\Gamma_{r,\mathfrak{r}}\varphi\quad\forall r\not\in  \cal{C}_u,\quad \Gamma_{\cal{C}_u,\mathfrak{r}}^L\varphi=w_{u_-}K_u\Gamma_{u,\mathfrak{r}}\varphi.
\end{align}
However, the lumping leaves $(\pi_{v_-})_{v\not\in\cal{C}_u}$ and $(K_{v})_{v\not\in\cal{C}_u}$ unchanged, and so,~\eqref{eq:Pivu1} implies that
$$\Gamma_{s,r}^L=\Gamma_{s,r}\quad\forall s\in\cal{C}_{r},\enskip r\in\mathbb{T}\backslash(\{u\}\cup\cal{C}_u).$$
Hence, if $a$ is an ancestor of $u$ ($a\in\mathbb{T}$ s.t.\ $u\in\mathbb{T}_a$),   $r$ a grandchild of $u$ ($r\in\cup_{v\in\cal{C}_u}\cal{C}_v$), and $s$ a descendant of $r$ ($s\in\mathbb{T}_r$), it follows from~\eqref{eq:Pivu2} that
\begin{equation}\Gamma_{u,r}^L=\Gamma_{u,r},\quad\Gamma_{a,\mathfrak{r}}^L=\Gamma_{a,\mathfrak{r}}.\label{eq:fndua8fbnya8bgfnaey8wegnau9}\end{equation}
By definition $\Gamma_{\cal{C}_u,u}^L=w_{u_-}K_u$ and the rightmost equation in~\eqref{eq:proofobj1} follows from that in~\eqref{eq:fndua8fbnya8bgfnaey8wegnau9}:
$$\Gamma_{\cal{C}_u,\mathfrak{r}}^L\varphi=\Gamma_{\cal{C}_u,u}^L\Gamma_{u,\mathfrak{r}}^L\varphi=w_{u_-}K_u\Gamma_{u,\mathfrak{r}}\varphi.$$
Moreover, for any grandchild $r$ of $u$ with parent $v$ ($v\in\cal{C}_u$ and $r\in\cal{C}_v$),
\begin{align*}\Gamma_{r,\cal{C}_u}^L\varphi&=\gamma^{\not r}_{\cup_{v'\in\cal{C}_u}\cal{C}_{v'}}(w_{{\cal{C}_u}_-}K_{\cal{C}_u}\varphi)=\gamma^{\not r}_{\cal{C}_v}\left(w_{v_-}\left(\prod_{v'\in\cal{C}_u^{\not v}}\gamma_{\cal{C}_{v'}}\right)\left(\left[\prod_{v'\in\cal{C}_u^{\not v}}w_{v'_-}\right]K_{\cal{C}_u}\varphi\right)\right)\\
&=\gamma^{\not r}_{\cal{C}_v}\left(w_{v_-}K_v\left(\prod_{v'\in\cal{C}_u^{\not v}}\gamma_{v'_-}\times K_{v'}\right)(\varphi)\right)=\gamma^{\not r}_{\cal{C}_v}(w_{v_-}K_v\gamma_{\cal{C}_u^{\not v}}(\varphi)).\end{align*}
Thus,
\begin{align*}\Gamma_{r,u}^L\varphi&=\Gamma_{r,\cal{C}_u}^L\Gamma_{\cal{C}_u,u}^L\varphi=\gamma^{\not r}_{\cal{C}_v}(w_{v_-}K_v\gamma_{\cal{C}_u^{\not v}}(w_{u_-}K_u\varphi))=\gamma^{\not r}_{\cal{C}_v}(w_{v_-}K_v\Gamma_{v,u}\varphi)\\
&=\Gamma_{r,v}\Gamma_{v,u}\varphi=\Gamma_{v,u}\varphi,\end{align*}
and the leftmost equation in~\eqref{eq:proofobj1} also follows from~\eqref{eq:fndua8fbnya8bgfnaey8wegnau9}.
\end{proof}
\end{appendix}

\begin{acks}[Data Access Statement]
No new data were generated or analysed during this study.
\end{acks}
\begin{funding}
JK and AMJ acknowledge support from the EPSRC (grant \# EP/T004134/1) and the Lloyd's Register Foundation Programme on Data-Centric Engineering at the Alan Turing Institute. FRC acknowledges support from the EPSRC and the MRC OXWASP Centre for Doctoral Training (grant \# EP/L016710/1). FRC and AMJ acknowledge further support  from the EPSRC (grant \#  EP/R034710/1).
\end{funding}



\bibliographystyle{imsart-number} 
\bibliography{dac_biblio}       


\end{document}